\documentclass[aps,prx,reprint,superscriptaddress,amsmath,amssymb,nofootinbib]{revtex4-2}

\usepackage{graphicx}
\graphicspath{{numerics/figs/}{./}}
\usepackage{amssymb,amsthm,amsfonts}
\usepackage{mathtools}
\usepackage{physics}
\usepackage{bbm}
\usepackage{bm}
\usepackage[dvipsnames]{xcolor}
\usepackage{booktabs}
\usepackage{hyperref}

\providecommand{\openone}{}\renewcommand{\openone}{\mathbbm{1}}
\newcommand{\TER}{\mathrm{TER}}
\newcommand{\TSR}{\mathrm{TSR}}
\newcommand{\TNR}{\mathrm{TNR}}
\newcommand{\NSIT}{\mathrm{NSIT}}
\newcommand{\PDO}{\mathrm{PDO}}

\newcommand{\LHV}{\mathrm{LHV}}

\newtheorem{theorem}{Theorem}
\newtheorem{proposition}[theorem]{Proposition}
\newtheorem{corollary}[theorem]{Corollary}
\newtheorem{lemma}[theorem]{Lemma}

\begin{document}

\title{Temporal nonlocality of a qudit resides in the input state, not the channel, and certifies temporal teleportation up to a fundamental limit}

\author{Karol Bartkiewicz}
\email{bark@amu.edu.pl}
\affiliation{Institute of Spintronics and Quantum Information, Faculty of Physics and Astronomy, Adam Mickiewicz University, 61-614 Pozna\'n, Poland}
\author{Patrycja Tulewicz}
\affiliation{Institute of Spintronics and Quantum Information, Faculty of Physics and Astronomy, Adam Mickiewicz University, 61-614 Pozna\'n, Poland}
\date{\today}

\begin{abstract}
Correlations between two moments in time can be too strong for any
classical explanation---and, remarkably, this can happen for a single
quantum system measured twice, with no second particle involved. We show
that when one qudit is sent through a noisy channel, the strength of this
``nonlocality in time''---the temporal nonlocality robustness $\TNR$---is
carried entirely by the starting state: it vanishes precisely when the
input is maximally mixed (completely random),
$\TNR(\rho_A,\mathcal{E})=0\Leftrightarrow\rho_A=\openone/d$, for the
standard noise families. The resource is not any coherence in the channel
but the back-action of the input's mixedness, and it survives even complete
decoherence. This is at once a power and a trap. As a power, $\TNR$
lower-bounds---device-independently, in the prepare-and-measure sense---the
fidelity of \emph{temporal teleportation}, sending an unknown state forward
in time and reaching $7/9$ at $d=3$. As a trap, because the
certified quantity is decoupled from the channel's actual coherence
transmission, it can certify more than the channel delivers: an injective
(reversible) unitary attains the maximal temporal-Bell signal yet teleports
below the classical baseline. We resolve this over-certification
completely---a universal cap $\TNR\le(d-1)/d$ with an exact
channel-resolved value, honest certification for the depolarizing channel
and for any sufficiently mixed probe, and a proof that no choice of probes
makes it channel-universal. Underpinning the results is a unified
semidefinite-programming hierarchy of the temporal entanglement, steering
and nonlocality robustnesses ($\TER$, $\TSR$, $\TNR$), with a strict lower
hierarchy and an upper one conditional on no-signaling in time ($\NSIT$).
All structure is verified numerically for $d=2$ through $5$.
\end{abstract}

\maketitle

\section{Introduction}\label{sec:intro}

Quantum correlations can be non-classical not only \emph{across
space}---between two particles measured far apart, as in a Bell
test---but \emph{across time}, between two measurements made on a
single system at different moments. For one qudit sent through a noisy
channel, what makes its two-time correlations non-classical? The answer
is clean and slightly surprising: the resource is how far the
\emph{input state} sits from being completely mixed, while the noise it
passes through is, for the standard noise families, immaterial---temporal
nonlocality \emph{resides in the input state, not the channel}. We make
this precise, place it in a hierarchy mirroring the spatial
entanglement--steering--nonlocality ladder, and turn it into an
operational guarantee for teleporting an unknown state \emph{forward in
time}. The four quantities below have plain readings: the
nonlocality/steering/entanglement robustnesses ($\TNR,\TSR,\TER^{\rm
sep}$) measure ``how much noise the temporal correlation survives'' at
each rung of the ladder, and the NSIT violation
$\mathcal{V}_{\rm NSIT}$ measures ``how much Alice's earlier measurement
disturbs Bob's later statistics''---the causal back-action that powers
the protocol.

\begin{figure}[t]
\centering
\includegraphics[width=\columnwidth]{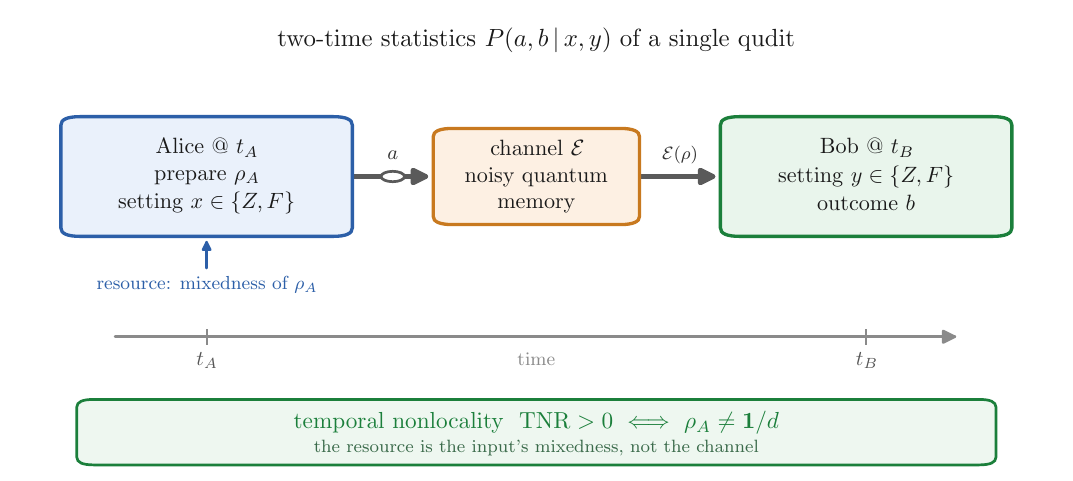}
\caption{The two-time scenario. A single qudit is prepared in $\rho_A$
and measured at $t_A$ (setting $x$), sent through a noisy channel
$\mathcal{E}$ (a quantum memory) to $t_B$, and measured again (setting
$y$). The non-classicality of the resulting two-time statistics
$P(a,b\,|\,x,y)$---temporal nonlocality---is non-zero \emph{if and only
if} the input $\rho_A$ is not maximally mixed (for the standard noise
families): the resource is the input's mixedness, not the channel.}
\label{fig:schematic}
\end{figure}

Three inequivalent forms of nonclassicality are firmly established for
spatially separated quantum systems---entanglement, EPR steering (in
which one party's measurements appear to remotely shape the other's
state), and Bell nonlocality---and they are arranged in a strict
hierarchy
\cite{PhysRevLett.98.140402,RevModPhys.86.419,RevModPhys.92.015001,Jirakova2021}.
Every Bell-nonlocal state is steerable, every steerable state is
entangled, and the inclusions are strict. Following the Leggett--Garg
program \cite{Leggett1985,Emary2013}, temporal counterparts have
been formalized: temporal steering
\cite{PhysRevA.89.032112,bartkiewicz2016temporal,bartkiewicz2016experimental,chen2016quantifying},
temporal Bell nonlocality, and time-like ``entanglement'' through
pseudo-density operators (PDOs---a single state-like operator that
packages the statistics of two separate times, at the price of allowing
negative eigenvalues) \cite{fitzsimons2015quantum} or
quasi-probability state-over-time constructions
\cite{WOOTTERS19871,Gibbons2004,Gross2006}. Recent work has put states
over time on a firm axiomatic footing \cite{LieFullwood2025} and, of
direct relevance here, has shown that a single operator representation
carries the qubit PDO over to systems of \emph{any} dimension
\cite{Fullwood2024,Jia2026}---a dimension-independent state-over-time
that our odd-prime Wigner construction realizes without the
contextuality artifacts noted above. For
two-level systems an analogous temporal hierarchy was demonstrated in
reference~\cite{PhysRevA.98.022104} under the no-signaling-in-time
(NSIT) condition---that the choice of the earlier measurement leave the
later statistics unchanged---\cite{Kofler2013,Halliwell2017,mal2016probing}. Partial
results for qutrits---including the robustness quantifiers and
numerical evidence for the hierarchy---were reported in
reference~\cite{Maskalaniec2021}. The temporal-correlation landscape
has since sharpened: the temporal Tsirelson bound and its sensitivity
to the dynamics and measurement scenario \cite{Wojcik2025}, extreme
Leggett--Garg violations under superposed unitaries
\cite{Chatterjee2025}, temporal nonlocality from indefinite causal
order \cite{Letertre2025}, and spatial-incompatibility witnesses for
temporal correlations \cite{Liu2025incompat} all refine the picture.
Whether the hierarchy survives
analytically in $d\geq 3$ has remained an open question.

Higher dimensions also unlock genuinely multilevel phenomena, the most
striking being contextuality \cite{Budroni2021,Klyachko2008}---the
impossibility of assigning measurement outcomes independently of which
other compatible measurements are performed alongside. This carries a
direct cost here: for $d\geq 3$ the natural observables (the generators
of $\mathrm{SU}(d)$) no longer share a single context, so a naive qudit
version of the qubit PDO is ill-defined---its spectrum would depend on
the measurement basis even for a closed system. The
discrete Wigner phase-space point operators $\{K_i\}_{i=0}^{d^2-1}$
\cite{Gibbons2004,Gross2006,Dawkins2015PRL}, however, provide a
manifestly non-contextual basis on the Wigner polytope and we use them
throughout.

The operational stake is that, just as spatial entanglement
underwrites teleportation across a spacelike interval
\cite{Bennett1993}, time-like correlations underwrite \emph{quantum
teleportation in time}
\cite{fitzsimons2015quantum,MacLean2017,Pisarczyk2019,Horsman2017},
realized experimentally e.g.\ via entanglement swapping between
photons that never coexist \cite{Megidish2013}. In plain terms,
ordinary teleportation moves an unknown state between two points in
\emph{space} by consuming a shared entangled pair together with a
classical message; teleportation in time instead moves an unknown
state between two \emph{moments}---an earlier preparation and a later
read-out---of a single system held in a quantum memory, with the
memory's two-time correlations playing the role that the shared
entangled pair plays in the spatial protocol. The input is never
copied (no-cloning forbids it): it is reconstructed at the later time,
so a protocol that succeeds certifies that genuine quantum
information---not merely classical data---has survived the wait. This
is the operational task our certificate addresses. Different
teleportation-in-time protocols make different device-trust
assumptions: the fully characterized version consumes a PDO-level
resource; a one-sided device-independent variant consumes only the
assemblage; the device-independent variant consumes only the
input--output behavior. Each lives at a distinct level of the
hierarchy of temporal correlations, so the question of universality is
also a question about which protocols can operate in a given
(state, channel) configuration.

We adopt \emph{robustness}---the minimal admixture of admissible noise
required to make the object classical---as the common currency for all
three correlations. We define \emph{temporal entanglement robustness}
(TER) and \emph{temporal nonlocality robustness} (TNR) as causal
counterparts of the spatial robustnesses of \cite{PhysRevA.59.141} and
\cite{PhysRevA.93.052112}, and combine them with the temporal steering
robustness (TSR) of \cite{PhysRevA.94.062126,Chen2017SciRep}. Each is
a linear semidefinite program.

\paragraph*{The central result: the resource is input mixedness.}
For the standard channel families, under the canonical two-MUB scheme (a
pair of mutually unbiased bases, MUBs), the temporal nonlocality robustness
obeys the strict equivalence
(Corollary~\ref{cor:tnr_iff})
\begin{equation}\label{eq:thesis}
\TNR(\rho_A,\mathcal{E})=0 \iff \rho_A=\openone/d,
\end{equation}
at every dimension admitting a Fourier MUB. \emph{The departure of the
input from maximal mixedness is the necessary and sufficient temporal-Bell
resource; the noise channel is dispensable.} The mechanism is
back-action: which basis Alice measures in disturbs the state Bob later
receives, and that disturbance---a signal forward in time---exists
precisely when the input is not already maximally mixed, and \emph{survives
even complete decoherence} of the channel (section~\ref{sec:TNRcoherence}).
At $d=2$ this recovers the qubit temporal hierarchy of Ku \emph{et al.}~\cite{PhysRevA.98.022104};
the equivalence is channel-conditional, not universal over all injective
channels, since an adversarial unitary can hide the resource from a fixed
scheme (Remark after Corollary~\ref{cor:tnr_iff},
cf.~\cite{Wojcik2025}).

\paragraph*{The power: a device-independent certifier.}
The prepare-and-measure temporal-teleportation fidelity equals the
average fidelity of the channel's Heisenberg--Weyl twirl,
$\mathcal{F}_{\rm DI}=1/d+(d-1)p/d$, and on the standard families the
certified robustness lower-bounds it,
$\mathcal{F}_{\rm DI}\geq 1/d+(d-1)T/d$ ($T:=\TNR$), up to a supremal
guaranteed value $7/9$ at $d=3$ (Theorem~\ref{thm:ditit})---the temporal
counterpart of the entanglement certification behind device-independent
quantum key distribution, with applications to secure communication over
time-bin channels and quantum memories.

\paragraph*{The trap: certifying more than the channel can deliver.}
Because the certified quantity ($\TNR$, a property of the input) is
\emph{decoupled} from what the channel actually transmits (its coherence
fidelity $F_e$), the certificate can promise more than the channel can
honor---\emph{over-certification}. The extreme case is stark: an injective
\emph{unitary} attains the maximal temporal-Bell signal $\TNR=(d-1)/d$
yet teleports at fidelity $1/(d+1)$, \emph{below} the no-resource baseline
$1/d$. We resolve the phenomenon completely (section~\ref{sec:overcert}):
a universal cap $\TNR\leq(d-1)/d$ with an exact channel-resolved value
(Lemma~\ref{lem:cap}); honest certification for the depolarizing channel
unconditionally and, for any channel, once the probe is sufficiently mixed
(Proposition~\ref{prop:mixsuff}); and a proof that \emph{no} choice of
probes makes the certificate channel-universal. This is a general caution
for temporal-correlation-based certification: a single-system
Bell-in-time test can pass maximally while certifying an operational
guarantee the system cannot honor.

\paragraph*{The machinery.}
The above rests on a unified semidefinite-programming treatment of the
three robustnesses (section~\ref{sec:robustness}) and their hierarchy
(section~\ref{sec:hier-frame}): $\TER^{\rm sep}\geq\TSR\geq\TNR\geq0$ on
$\openone/d$, with TER a causality monotone
(Theorem~\ref{thm:TERmonotone}), the lower inequality $\TSR\geq\TNR$
universal (Theorem~\ref{thm:tsr_tnr}, which the cap above uses), and the
upper one $\TER^{\rm sep}\geq\TSR$ conditional on no-signaling in time
(NSIT)---it fails off $\openone/d$, where the temporal correlations
signal, and is replaced by a single universal inequality
$\TSR\leq\TER^{\rm sep}+\tfrac12\mathcal{V}_{\rm NSIT}$ with the constant
$\tfrac12$ tight (Proposition~\ref{prop:univ_bound}). Every claim is
verified on $\rho_A$-adapted Monte-Carlo sweeps ($10^6$
configurations at $d=3$; $d=2$ through $5$).

\paragraph*{Relation to prior work.}
The qubit case~\cite{PhysRevA.98.022104} established a temporal
hierarchy under NSIT at $d=2$, and reference~\cite{Maskalaniec2021}
introduced robustness quantifiers and reported numerical evidence for
the qutrit hierarchy. The present contribution is fourfold and
distinct: (a)~the \emph{analytic} asymmetric state-boundness
theorem---$\TNR=0\iff\rho_A=\openone/d$ with the channel
dispensable---proved uniformly across dimensions, rather than observed
numerically at fixed $d$; (b)~the determination that the upper
inequality $\TER^{\rm sep}\geq\TSR$ is \emph{NSIT-conditional}, with an
explicit refutation of any state-independent (``NSIT-free'')
entanglement-over-steering bound; (c)~the operational
identification of TNR as a certificate of prepare-and-measure
temporal-teleportation fidelity; and (d)~the discovery and complete
resolution of \emph{over-certification}---a universal cap, an exact
mixedness criterion for honest certification, and a no-go for
channel-universal certification---exposing a general limit of
temporal-correlation-based device-independent certification. The SDP
implementations and qutrit benchmarks build on those of
reference~\cite{Maskalaniec2021}; the Supplementary Material tabulates
these contributions against the qubit case of
reference~\cite{PhysRevA.98.022104}.

The paper follows the arc above. Sections~\ref{sec:framework}
and~\ref{sec:robustness} set up the two-time scenario, the channel Choi
state, and the three robustness quantifiers. Section~\ref{sec:results}
proves the central result---state-boundness and its back-action
mechanism. Section~\ref{sec:hier-frame} develops the supporting
robustness hierarchy and the universal NSIT bound. The operational
reading is section~\ref{sec:ditit}: the temporal-teleportation
certificate and, as its sharp limit, over-certification
(section~\ref{sec:overcert}) with its complete resolution.
Section~\ref{sec:numerics} reports the numerical verification across
$d=2$--$5$ and section~\ref{sec:discussion} the broader implications;
proofs, channel representations, and reproduced figures are collected in
the Methods and Supplementary Material. A reader after the operational
message can go straight to sections~\ref{sec:results}
and~\ref{sec:ditit}.

\section{Theoretical framework}\label{sec:framework}

\subsection{Non-contextual states over time}\label{sec:noncontextual}

The substantive state-over-time object in what follows is the channel
\emph{Choi state}
$\Lambda_{\mathcal{E}}=(\mathrm{id}\otimes\mathcal{E})
|\Phi^+\rangle\langle\Phi^+|$, with
$|\Phi^+\rangle=\tfrac{1}{\sqrt d}\sum_k|kk\rangle$. Following
reference~\cite{fitzsimons2015quantum}, the two-time statistics also
define a \emph{pseudo-density operator} (PDO)
\begin{equation}\label{eq:new_PDO}
R^{\PDO}=\frac{1}{\mathcal{N}}\sum_{i,j=0}^{d^2-1}
\langle G_i\otimes G_j\rangle\,G_i\otimes G_j,
\end{equation}
in a Hilbert--Schmidt-orthogonal operator basis $\{G_i\}$
($\tr(G_iG_j)=d\,\delta_{ij}$), with $\langle G_i\otimes G_j\rangle=
\sum_{a,b}ab\,p(a,b|i,j)$ the time-ordered expectation of the product
of outcomes at $t_A$ and $t_B$
[$p(a,b|i,j)=\tr(\Pi^B_{j,b}\,\mathcal{E}(\Pi_{i,a}\rho_A\Pi_{i,a}))$].
The object in \eqref{eq:new_PDO} is the \emph{pseudo-density matrix}
(PDM) of \cite{fitzsimons2015quantum}, also employed in \cite{Comar2026};
we retain the name pseudo-density \emph{operator} (PDO) to stress that we
work with its basis-independent, non-contextual realization introduced
below.
The PDO is an equivalent representation of the Choi state; its operator
expansion, the Choi--Jamio\l{}kowski operator $E_{B|A}=d\Lambda_{\mathcal{E}}$,
and the non-contextual Wigner realization that makes $R^{\PDO}$
basis-independent at $d\geq 3$---where $\mathrm{SU}(d)$ contextuality
\cite{Budroni2021} would otherwise make it eigenbasis-dependent---are
collected in the Supplementary Material and may be skipped on a first
reading: the robustness quantifiers of section~\ref{sec:robustness} are
built directly from the Choi state and the measured assemblage and
behavior. The one consequence used below is that $R^{\PDO}$ and
$\Lambda_{\mathcal{E}}$ are \emph{inequivalent} at $d\geq 3$ on
$\rho_A=\openone/d$, motivating the channel-Choi-based quantifiers of
section~\ref{sec:robustness}.

\paragraph*{Causality monotones.}
A state over time is generally not positive semidefinite. Following
reference~\cite{fitzsimons2015quantum}, a function $\Phi(R)$ is a
\emph{causality monotone} iff (i) $\Phi(R)\geq 0$, with $\Phi(R)=0$ if
$R$ is completely positive and maximal for two consecutive measurements
on a closed system; (ii) $\Phi$ is invariant under unitary maps; (iii)
non-increasing under local CPTP maps; (iv) convex,
$\sum_i p_i\Phi(R_i)\geq\Phi(\sum_i p_iR_i)$. The trace-norm monotone
$f(R)=\norm{R}_{\tr}-1$ \cite{fitzsimons2015quantum} satisfies these axioms.

A space-like separable state over time is
\begin{equation}\label{eq:separableR}
R^{\mathrm{SEP}}=\sum_k p_k\,\rho_k^A\otimes\rho_k^B,
\quad p_k\geq 0,\;\sum_k p_k=1,
\end{equation}
with $\rho_k^{A,B}\geq 0$. Note that $R\geq 0$ does not imply
separability in this sense: there can remain space-like temporal
correlations.

\paragraph*{Measurements.} Throughout the paper all measurements
on Alice's and Bob's sides are projective, i.e.\ the POVM elements
are rank-one orthogonal projectors $M_{a|x}=\Pi_{a|x}$ with
$\sum_a\Pi_{a|x}=\openone$ for each setting $x$. Generalized Lüders
measurements with continuously variable strength $\eta\in[0,1]$ can
be substituted into the same framework, but they do not enter any
of the proofs below.

\section{Robustness-based quantifiers of two-time correlations}
\label{sec:robustness}

The robustness of a quantum object is the minimal weight of an
admissible noise admixture that renders the object classical. For
time-like entanglement we use four closely related but operationally
distinct quantifiers: the pseudo-density-operator robustness
$\TER^{\PDO}$---the causal counterpart of Vidal--Tarrach spatial
entanglement robustness applied to the non-contextual state-over-time
of section~\ref{sec:framework}---and three channel-Choi-based variants
($\TER^{\rm Choi\text{-}pos}\!\equiv\!0$, the separability-based
$\TER^{\rm sep}$, and its partial-transpose proxy $\TER^{\rm PPT}$)
that arise when one asks \emph{which} property of the Choi state
$\Lambda_{\mathcal{E}}$ controls the temporal hierarchy. Of these,
$\TER^{\rm sep}=\mathrm{ER}(\Lambda_{\mathcal{E}})$ is the only
physically justified temporal-entanglement quantifier; $\TER^{\PDO}$
and $\TER^{\rm PPT}$ are computable proxies that can fall below $\TSR$
off $\openone/d$, where even $\TER^{\rm sep}$ does
(section~\ref{sec:nsit-free}). Their precise definitions, ordering, and
the partial-transpose SDP form are collected in
the Supplementary Material. A closely related notion of temporal
entanglement---defined directly from the pseudo-density-matrix
structure and likewise linked to non-signaling in time and to temporal
Bell violations---was recently introduced by Comar~\emph{et
al.}~\cite{Comar2026}; in the present language it is the object whose
Vidal--Tarrach robustness we quantify as $\TER^{\PDO}$.

\subsection{Temporal entanglement robustness}\label{sec:TER}

The spatial entanglement robustness of Vidal and Tarrach
\cite{PhysRevA.59.141} is the minimal weight of an admixed (admissible)
noise that makes a state separable. Its temporal counterpart we define as
the minimal admixture of a pseudo-density operator that destroys time-like
entanglement,
\begin{eqnarray}\label{eq:TER}
\TER &=& \min_{\gamma}\;\gamma,\nonumber\\
\text{s.t. }&& \frac{R+\gamma\,\mathfrak{R}}{1+\gamma}=\rho,\nonumber\\
&& \mathfrak{R}\;\text{a pseudo-density operator},\nonumber\\
&& \rho\geq 0,\quad \gamma\geq 0.
\end{eqnarray}
Restricting $\mathfrak{R}$ to a unit-trace positive semi-definite
operator (the implicit assumption for $\gamma\mathfrak{R}$ to be a
``noise'' admixture rather than an arbitrary Hermitian compensator),
the substitution $\tilde{\mathfrak{R}}=\gamma\mathfrak{R}$ yields the
linear SDP
\begin{eqnarray}\label{eq:TER2}
\TER &=& \min\,\bigl(\tr\tilde{\mathfrak{R}}\bigr),\nonumber\\
\text{s.t. }&& R+\tilde{\mathfrak{R}}\succeq 0,\quad
\tilde{\mathfrak{R}}\succeq 0.
\end{eqnarray}
The constraint $\tilde{\mathfrak{R}}\succeq 0$ is essential: without
it the linear program is unbounded below, since arbitrary indefinite
Hermitian compensators with large negative trace would satisfy
$R+\tilde{\mathfrak{R}}\succeq 0$ at zero or negative cost. With it,
the SDP is solved efficiently by interior-point methods, and the
optimum coincides with the sum of absolute values of negative
eigenvalues of $R$.

\noindent The operator $R$ in \eqref{eq:TER}--\eqref{eq:TER2} is the
pseudo-density operator $R^{\PDO}$ of \eqref{eq:new_PDO}, so the
robustness defined here is the PDO-based quantifier
$\TER^{\PDO}$ singled out in section~\ref{sec:robustness}. We write the
explicit superscript $\TER^{\PDO}$ whenever it must be distinguished
from the channel-Choi-based variants $\TER^{\rm sep}$ and
$\TER^{\rm PPT}$, and retain the bare symbol $\TER$ where the context
fixes the PDO construction.

\begin{theorem}\label{thm:TERmonotone}
TER is a causality monotone in the sense of
\cite{fitzsimons2015quantum}.
\end{theorem}

\noindent Throughout, the proofs of the technical lemmas,
propositions and theorems are collected in the Supplementary Material;
the two central arguments---the necessity of state-boundness
(Theorem~\ref{thm:necessity}) and the universal bound
(Proposition~\ref{prop:univ_bound})---are given in full in the main
text. (For two-qubit or qubit--qutrit systems the Peres--Horodecki PPT
criterion \cite{PhysRevLett.77.1413,HORODECKI19961} turns
\eqref{eq:TER2} into a separability-faithful SDP and gives the proxy
ordering $\mathrm{ER}\geq\TER^{\PDO}$; see the Supplementary Material.)

\subsection{Temporal steering robustness}\label{sec:TSR}

Alice performs a projective measurement $\{\Pi_{a|x}\}$ at $t_A$ on
$\rho_A$; the post-measurement state propagates through the channel,
producing the assemblage
$\tilde\rho_{a|x}=\mathcal{E}(\Pi_{a|x}\rho_A\Pi_{a|x})$. It admits a
hidden-state model iff
$\tilde\rho_{a|x}=\sum_\lambda p(\lambda)D(a|x,\lambda)\rho_\lambda$
with $D(a|x,\lambda)$ deterministic. Following
references~\cite{PhysRevA.94.062126,Chen2017SciRep,PhysRevLett.114.060404},
the temporal steering robustness is the minimal admixture restoring
such a model,
\begin{eqnarray}\label{eq:TSR_SDP}
\TSR &=& \min\,\Bigl(\tr\sum_\lambda\tilde\rho_\lambda - 1\Bigr),\nonumber\\
\text{s.t. }&& \sum_\lambda D(a|x,\lambda)\,\tilde\rho_\lambda
\geq \tilde\rho_{a|x},\quad \tilde\rho_\lambda\geq 0\;\forall\lambda.
\end{eqnarray}
The general-POVM assemblage and the hidden-state-model definition are
given in the Supplementary Material.

\paragraph*{NSIT condition.}
A consistency requirement for the temporal Bell scenario is
no-signaling in time
\cite{Kofler2013,Halliwell2017,mal2016probing,PhysRevA.98.022104},
\begin{equation}\label{eq:NSIT}
\sum_a\tilde\rho_{a|x}=\tr_A\!\bigl(E_{B|A}\rho_A^T\bigr)\quad\forall x,
\end{equation}
i.e.\ Bob's marginal state at $t_B$ is independent of Alice's
choice of measurement basis at $t_A$. Eigenstates of $\rho_A$
mutually unbiased to a measurement basis violate \eqref{eq:NSIT};
thus pure initial states are typically NSIT-violating. We quantify
the departure from \eqref{eq:NSIT} by the NSIT-violation monitor
\begin{equation}\label{eq:Vnsit}
\mathcal{V}_{\NSIT}
=\max_{x}\Bigl\|\textstyle\sum_a\tilde\rho_{a|x}
-\tr_A\!\bigl(E_{B|A}\rho_A^T\bigr)\Bigr\|_1,
\end{equation}
the largest trace-norm distance between Bob's outcome-averaged
post-measurement state under setting $x$ and the no-intervention
marginal of \eqref{eq:NSIT}. This is the faithful witness of
\eqref{eq:NSIT}: $\mathcal{V}_{\NSIT}=0$ iff NSIT holds, which by
Theorem~\ref{thm:nsit_state} (for injective $\mathcal{E}$) happens iff
$\rho_A=\openone/d$. (A weaker pairwise quantity
$\max_{x,x'}\|\sum_a(\tilde\rho_{a|x}-\tilde\rho_{a|x'})\|_1$ is
\emph{not} a faithful witness---it can vanish on inputs that are
uniform in every measured basis yet are not $\openone/d$; see the
remark after Corollary~\ref{cor:three_way}.) We formally
define each robustness for arbitrary $\rho_A$ and $\mathcal{E}$
without invoking NSIT a priori; the entanglement-over-steering
interpretation of the upper hierarchy
(Theorem~\ref{thm:hierarchy_strict}) requires NSIT and genuinely
fails without it (section~\ref{sec:nsit-free}).

\subsection{Temporal nonlocality robustness}\label{sec:TNR}

In a Bell-type two-time scenario Alice and Bob measure $A_x, B_y$ with
joint distribution
\begin{equation}\label{eq:born}
P(a,b|x,y)=\tr\bigl[M_{b|y}\,\mathcal{E}\bigl(\sqrt{M_{a|x}}\rho_A\sqrt{M_{a|x}}\bigr)\bigr].
\end{equation}
A local hidden-variable (LHV) model exists iff
\begin{equation}\label{eq:LHV}
P(a,b|x,y)=\sum_{\mu,\nu} p(\mu,\nu)\,D(a|x,\mu)\,D(b|y,\nu).
\end{equation}
The temporal nonlocality robustness, in analogy with the spatial
quantity of reference~\cite{PhysRevA.93.052112}, is the minimal admixture of
an admissible behavior $Q$ that renders $P$ local,
\begin{eqnarray}\label{eq:TNR}
\TNR &=& \min\,\beta,\nonumber\\
\text{s.t. }&& \frac{P(a,b|x,y)+\beta\,Q(a,b|x,y)}{1+\beta}
= R(a,b|x,y),\nonumber\\
&& R \in \LHV,\quad Q\;\text{a behavior},\quad \beta\geq 0.
\end{eqnarray}
Eliminating $Q$ through the normalization constraints (see
the Supplementary Material) yields the linear SDP
\begin{eqnarray}\label{eq:TNR_SDP}
\TNR &=& \min\!\left[\frac{\sum_{x,y,a,b,\mu,\nu}\tilde r_{\mu\nu}\,D(a|x,\mu)D(b|y,\nu)}{\sum_{x,y,a,b}P(a,b|x,y)}-1\right],\nonumber\\
\text{s.t. }&& \sum_{\mu,\nu}\tilde r_{\mu\nu}\,D(a|x,\mu)D(b|y,\nu)\geq P(a,b|x,y),\nonumber\\
&& \tilde r_{\mu\nu}\geq 0,
\end{eqnarray}
with $\tilde r_{\mu\nu}=(1+\beta)\,r(\mu,\nu)$. Restricting $Q$ itself
to be LHV defines the LHV-TNR variant, useful when one wishes the noise
itself to be classical.

\section{State-boundness: the temporal-Bell resource is the input's
mixedness}\label{sec:results}\label{sec:TNRcoherence}

We now establish the central result: for the standard channel families
the temporal nonlocality robustness vanishes \emph{exactly} on the
maximally mixed input, so non-maximal mixedness of the input is the
necessary-and-sufficient resource and the noise channel is dispensable.
We begin with the physical mechanism.

\paragraph*{State-bound vs.\ channel-bound resources.}
A natural foundational question is whether each temporal correlation
is a property of the input state $\rho_A$, of the channel
$\mathcal{E}$, or jointly of both. Spatially, nonlocality is a property
of the shared state alone. Temporally the answer is more delicate, and
our results pin it down for the three robustness measures
asymmetrically.

\paragraph*{Physical picture: a back-action signal that dephasing cannot
erase.} The mechanism behind the results of this section is visible on a
single example. Send $\rho_A=|0\rangle$ through a channel that fully
dephases in the computational basis, and let Alice choose to measure
either in that basis ($Z$) or in the conjugate Fourier basis ($F$).
Measuring $Z$ leaves the eigenstate $|0\rangle$ undisturbed, so Bob
receives $|0\rangle$; measuring $F$ projects onto a Fourier vector, which
the dephasing then turns into white noise $\openone/d$. Bob's later
statistics therefore flip---sharp for $x=Z$, uniform for $x=F$---according
to \emph{which basis Alice chose}. This is signaling in time: Bob's
marginal carries information about Alice's setting [a violation of
no-signaling in time, equation~\eqref{eq:NSIT}], the temporal face of
measurement back-action. A behavior that signals cannot be local, so the
temporal-nonlocality robustness is non-zero \emph{even though the channel
has erased every trace of coherence}. The resource is not channel
coherence but the back-action of Alice's choice---which exists precisely
when the input is not already maximally mixed: on $\rho_A=\openone/d$
every measurement leaves the state $\openone/d$, Alice's choice disturbs
nothing, no signal reaches Bob, and $\TNR=0$. The propositions below make
this dichotomy exact; the same persistence of the signal through dephasing
is what makes the certification cap of section~\ref{sec:overcert}
delicate.

\begin{proposition}[State-boundness of TNR, sufficiency]\label{prop:TNRstatebound}
Let $\rho_A=\openone/d$, let $\mathcal{E}$ be any quantum channel,
and let $\{M_{a|x}\}=\{\Pi_{a|x}\}$ be projective measurements at
both times. The two-time joint distribution takes the factorized
form
\begin{equation}\label{eq:TNRHSM}
P(a,b|x,y)=\frac{1}{d}\,\tr[M_{b|y}\,\mathcal{E}(\Pi_{a|x})].
\end{equation}
Whenever the family of post-channel probability distributions
$\{q(b|y;a,x)=\tr[M_{b|y}\mathcal{E}(\Pi_{a|x})]\}_{a,x,y,b}$
admits an LHV decomposition --- i.e., there exists a probability
distribution $\{p(\lambda)\}$, deterministic responses
$D_A(a|x,\lambda)\in\{0,1\}$ with $\sum_a D_A(a|x,\lambda)=1$,
and deterministic responses $D_B(b|y,\lambda)\in\{0,1\}$
such that
$P(a,b|x,y)=\sum_\lambda p(\lambda)\,D_A(a|x,\lambda)\,D_B(b|y,\lambda)$ ---
then $\TNR(\openone/d,\mathcal{E})=0$.
\end{proposition}

The LHV reduction --- the analytical content above the
factorized form \eqref{eq:TNRHSM} --- is the substantive
condition. We verify it directly for the two unital standard
channels (depolarizing and phase damping), and confirm it
numerically for the non-unital amplitude-damping channel.

\begin{proposition}[Analytic LHV reduction for unital standard channels]
\label{prop:lhv_reduction}
On $\rho_A=\openone/d$ with the two-MUB projective settings
$x\in\{Z,F\}$ ($Z$ the eigenbasis of $\rho_A$, $F$ the Fourier MUB),
the two-time behavior of the depolarizing channel
$\mathcal{E}_{\rm depol}^{(e)}(\sigma)=e\,\sigma+(1-e)(\openone/d)\,\tr\sigma$
and of the phase-damping channel
$\mathcal{E}_{\rm phase}^{(e)}$ admit explicit local-hidden-variable
decompositions for every $e\in[0,1]$ (here $e=e^{-t}$ is the channel
visibility). Consequently
$\TNR(\openone/d,\mathcal{E}_{\rm depol}^{(e)})
=\TNR(\openone/d,\mathcal{E}_{\rm phase}^{(e)})=0$.
\end{proposition}

For the cascade amplitude-damping channel $\mathcal{E}_{\rm amp}$
(non-unital), $\mathcal{E}_{\rm amp}(\openone)\neq\openone$, and the
post-channel behavior does not decompose into the simple
``correlated--uncorrelated'' form of Proposition~\ref{prop:lhv_reduction}.
A different LHV construction, however, covers this case and indeed
\emph{any} channel whose computational-basis-to-computational-basis
action is diagonal (i.e., $\mathcal{E}(|j\rangle\langle j|)$ is
diagonal for every $j$):

\begin{proposition}[Product LHV for diagonal-action channels on
$\openone/d$]\label{prop:product_lhv}
Let $\mathcal{E}$ be any CPTP map such that
$\mathcal{E}(|j\rangle\langle j|)$ is diagonal in the computational
basis for every $j$ (this includes the three standard channels and,
more generally, any channel with real Kraus operators in the
computational basis). On $\rho_A=\openone/d$ with the two-MUB
scheme, the hidden-variable distribution
\begin{equation}\label{eq:product_lhv}
p(\alpha_Z,\alpha_F,\beta_Z,\beta_F)
= P(\alpha_Z,\beta_Z|Z,Z)\;\times\; P(\alpha_F,\beta_F|F,F)
\end{equation}
with deterministic responses $D_A(a|Z)=\delta_{a,\alpha_Z}$,
$D_A(a|F)=\delta_{a,\alpha_F}$, $D_B(b|Z)=\delta_{b,\beta_Z}$,
$D_B(b|F)=\delta_{b,\beta_F}$ is a valid LHV decomposition.
\end{proposition}

\noindent\emph{Remark.} The product LHV requires two conditions:
(a)~$\tr[\Pi_{b|F}\,\mathcal{E}(|a\rangle\langle a|)]$ is
independent of $a$ (so the cross-basis $P(a,b|Z,F)=1/d^2$), and
(b)~$[\mathcal{E}(|+_a\rangle\langle +_a|)]_{bb}$ is independent
of $a$ (so $P(a,b|F,Z)$ factorizes). The diagonal-action
hypothesis is sufficient for both but not necessary; the
construction applies to any channel satisfying (a) and~(b).

\noindent Proposition~\ref{prop:product_lhv} closes the amplitude-damping
gap in Corollary~\ref{cor:tnr_iff}: the cascade channel has
real Kraus operators in the computational basis
(Supplementary Material), hence diagonal action, and the
product LHV applies at every $d$ and every decay time $t\geq 0$.
The LHV construction is verified numerically to machine precision
($<10^{-15}$) across the full sweeps at $d=2,3,4,5$.

The converse direction strengthens this into a strict
necessary-and-sufficient characterization for the standard two-MUB
measurement setting:

\begin{theorem}[Necessity of state-boundness]\label{thm:necessity}
For any qudit state $\rho_A\neq\openone/d$ and any $d\geq 2$ admitting
mutually unbiased bases, there exist projective measurement settings
$\{M_{a|x}\}_{x=0,1}$ such that, with the identity channel, the
resulting two-time behavior $P(a,b|x,y)$ admits no LHV decomposition.
Hence $\TNR(\rho_A,\mathrm{id},\{M\})>0$.
\end{theorem}

\begin{proof}
Since $\rho_A\neq\openone/d$, $\rho_A$ has at least one eigenvalue
distinct from $1/d$. In its eigenbasis $\{|a\rangle\}$ the diagonal
entries $p_a=\langle a|\rho_A|a\rangle$ are non-uniform. Take Alice's
settings $x\in\{0,1\}$ to be the computational basis of
$\{|a\rangle\}$ and its Fourier MUB
$|+_k\rangle=(1/\sqrt d)\sum_j\omega^{kj}|j\rangle$ with
$\omega=e^{2\pi i/d}$, with identical settings for Bob. With identity
channel and projective measurements $M_{a|x}=\Pi_{a|x}$, the
two-time Born rule \eqref{eq:born} evaluates to
\begin{equation*}
P(a,b|x,y) =
\begin{cases}
p_a\,\delta_{a,b} & x=y=0,\\
p_a/d & x=0,\;y=1,\\
1/d^2 & x=1,\;y=0,\\
\delta_{a,b}/d & x=y=1.
\end{cases}
\end{equation*}
Suppose an LHV exists,
$P(a,b|x,y)=\sum_\lambda p(\lambda) D_A(a|x,\lambda)D_B(b|y,\lambda)$,
with deterministic responses $\alpha,\beta$. From
$P(a,b|0,0)=p_a\delta_{a,b}$ the joint $(\alpha(0),\beta(0))$ is
supported on the diagonal with marginal $p_a$, so $\beta(0)$ has
marginal $p_a$. From $P(a,b|1,0)=1/d^2$ the joint $(\alpha(1),\beta(0))$
is independent and uniform, so $\beta(0)$ has marginal $1/d$. Since
$\{p_a\}$ is non-uniform, the two requirements contradict each other,
and no LHV exists. The optimization defining TNR therefore has
strictly positive optimum.
\end{proof}

Combined with Proposition~\ref{prop:TNRstatebound}, this yields the
central result of the paper:
\begin{corollary}\label{cor:tnr_iff}
For any two-MUB projective measurement scheme
$\{M_{Z},M_{F}\}$ comprising the eigenbasis of $\rho_A$ and its
Fourier MUB (or any unitarily-equivalent pair, see
Proposition~\ref{prop:unitary_invariance} below),
\begin{equation}\label{eq:cor_iff}
\TNR(\rho_A,\mathcal{E};\{M_{Z},M_{F}\})=0 \;\;\iff\;\;
\rho_A=\openone/d,
\end{equation}
where the sufficiency direction (``$\Leftarrow$'') is proved
analytically for the depolarizing and phase-damping channels
(Propositions~\ref{prop:TNRstatebound},~\ref{prop:lhv_reduction})
and for the cascade amplitude-damping channel
(Propositions~\ref{prop:TNRstatebound},~\ref{prop:product_lhv}).
The necessity direction (``$\Rightarrow$'') holds for the identity
channel (Theorem~\ref{thm:necessity}) and for the standard decay
families (finite decay, $e^{-\kappa t}>0$), whose back-action
$\mathcal{E}(\rho_A-\openone/d)$ is visible to the two adapted MUBs; it
is \emph{not} universal over all injective channels (Remark below).
More broadly, the sufficiency holds for any channel with diagonal
action in the computational basis
(Proposition~\ref{prop:product_lhv}).
\end{corollary}

\noindent\emph{Remark (the channel is not fully dispensable).}
Equation~\eqref{eq:cor_iff} is established for the standard noise
families, where the input's departure from $\openone/d$ leaves a
trace---the back-action $\mathcal{E}(\rho_A-\openone/d)$ has a nonzero
diagonal in the eigenbasis or the Fourier MUB, so Bob's marginal
signals on Alice's setting and $\TNR>0$. It is \emph{not} universal:
for an adversarial unitary $\mathcal{E}=U(\cdot)U^\dagger$ with $U$
chosen so that $\mathcal{E}(\rho_A-\openone/d)$ is off-diagonal in
\emph{both} adapted bases, the back-action is invisible to the fixed
two-MUB scheme and $\TNR=0$ even for $\rho_A\neq\openone/d$ (numerically,
$\TNR<10^{-10}$ at $d=3$ for $\rho_A=\mathrm{diag}(0.5,0.3,0.2)$). The
input's resource is not lost: adapting the late measurement to the
channel output restores $\TNR=0.17$ for the same data. The limitation
is therefore one of the \emph{fixed} input-adapted scheme, not of
state-boundness, and mirrors the known scenario-dependence of temporal
quantum correlations \cite{Wojcik2025,Chatterjee2025}.

\begin{proposition}[Unitary invariance of the
TNR-zero region]\label{prop:unitary_invariance}
Let $U\in\mathrm{SU}(d)$ commute with $\rho_A$. If
$\TNR(\rho_A,\mathcal{E};\{M_{Z},M_{F}\})=0$ for the canonical
two-MUB scheme of Corollary~\ref{cor:tnr_iff}, then
$\TNR(\rho_A,U\mathcal{E}\,U^\dagger;\{UM_{Z}U^\dagger,UM_{F}U^\dagger\})=0$
as well.
\end{proposition}

\noindent Corollary~\ref{cor:tnr_iff} therefore holds for every measurement
scheme in the unitarily-equivalent orbit
$\{UM_{Z}U^\dagger, UM_{F}U^\dagger\}_{U\in\mathrm{SU}(d),
[U,\rho_A]=0}$ of the canonical two-MUB pair, against the
correspondingly conjugated channel. The sufficiency direction
uses two complementary LHV constructions: the
correlated--uncorrelated decomposition of
Proposition~\ref{prop:lhv_reduction} for the unital channels
(depolarizing and phase damping), and the product-form
decomposition of Proposition~\ref{prop:product_lhv} for the
cascade amplitude-damping channel (and, more broadly, any
diagonal-action channel). The necessity direction
(Theorem~\ref{thm:necessity}) holds for any injective $\mathcal{E}$,
so the iff~\eqref{eq:cor_iff} extends to any injective channel
for which either LHV construction applies. The only channels
not covered are those that create computational-basis coherence
from diagonal inputs, i.e.\ channels for which
$\mathcal{E}(|j\rangle\langle j|)$ has nonzero off-diagonal
elements; extending the iff to such channels is an open target.

The same canonical two-MUB scheme that drives the LHV decomposition
also yields a behavior-level equivalence with the NSIT condition:

\begin{theorem}[NSIT--state-boundness equivalence]\label{thm:nsit_state}
For any quantum channel $\mathcal{E}$ that is injective on the affine
span of trace-one Hermitian operators (in particular for the standard
channel families and for any non-state-collapsing CPTP map), and any
two-MUB projective measurement scheme $\{M_Z, M_F\}$ comprising the
eigenbasis of $\rho_A$ and its Fourier MUB,
\begin{equation}\label{eq:nsit_iff}
\mathcal{V}_{\rm NSIT}(\rho_A,\mathcal{E}) = 0
\;\Longleftrightarrow\;
\rho_A = \openone/d.
\end{equation}
\end{theorem}

Combining Theorem~\ref{thm:nsit_state} with
Corollary~\ref{cor:tnr_iff} gives a three-way equivalence:

\begin{corollary}\label{cor:three_way}
For any channel $\mathcal{E}$ in the standard family and the canonical
two-MUB scheme,
\begin{equation}
\mathcal{V}_{\rm NSIT}(\rho_A,\mathcal{E})=0
\;\Leftrightarrow\; \rho_A=\openone/d
\;\Leftrightarrow\; \TNR(\rho_A,\mathcal{E})=0.
\end{equation}
\end{corollary}

\noindent\emph{Remark (the scheme must be adapted to $\rho_A$, or
tomographically complete).} The two-MUB scheme above is
\emph{adapted}: its computational basis is the eigenbasis of
$\rho_A$. This is essential, not cosmetic. The necessity direction
works because in the eigenbasis Alice's marginal equals the spectrum
of $\rho_A$, which is non-uniform precisely when
$\rho_A\neq\openone/d$; equivalently the eigenbasis setting reproduces
the no-intervention marginal, so $\mathcal{V}_{\NSIT}$
[equation~\eqref{eq:Vnsit}] registers the full violation. For a
\emph{fixed} scheme not adapted to $\rho_A$, neither holds in general.
At $d=3$, for the fixed computational$+$Fourier pair the states
$\rho_A=\openone/3+M$ with $M=M^\dagger$ of vanishing computational
\emph{and} Fourier-basis diagonals form a four-parameter family
(two-parameter for any three fixed MUBs): each such
$\rho_A\neq\openone/3$ has uniform marginals in both fixed bases,
hence fixed-scheme $\TNR=0$ and vanishing \emph{pairwise} signaling,
although the no-intervention monitor~\eqref{eq:Vnsit} and the
$\rho_A$-adapted $\TNR$ remain strictly positive (e.g.\ $\TNR=0.18$
for a representative member). The family collapses to
$\{\openone/3\}$ only for a tomographically complete set of $d+1$
MUBs. The equivalence~\eqref{eq:nsit_iff} and
Corollary~\ref{cor:tnr_iff} therefore require the measurement to be
adapted to $\rho_A$ (as throughout this work and in the protocol of
section~\ref{sec:ditit}) or, for a single fixed scheme,
tomographically complete; this also explains why the redefined
no-intervention $\mathcal{V}_{\NSIT}$, rather than the pairwise
quantity, is the faithful witness.

This equivalence is borne out by our $\rho_A$-adapted sweeps: across
the $d=2,3,4,5$ Monte-Carlo ensembles the only zeros of
$\mathcal{V}_{\rm NSIT}$ and of $\TNR$ are the configurations at
$\rho_A=\openone/d$ (away from the $e\to0$ channel limit). Because the
measurement is adapted to the eigenbasis of $\rho_A$, this directly
tests the iff, with no reliance on avoiding the fixed-scheme zero
family of the Remark. Both quantities vanish
\emph{continuously} as $\rho_A\to\openone/d$---numerically linearly
along radial rays (figure~\ref{fig:tnr_zero}(c))---so there is no
finite separation gap; an apparent floor in a coarse sample is a
sampling artifact (section~\ref{sec:numerics}). The equivalence
elevates NSIT from a consistency requirement to a behavior-level test
that is operationally simpler than computing $\TNR$.

\begin{figure*}[t]
\includegraphics[width=\linewidth]{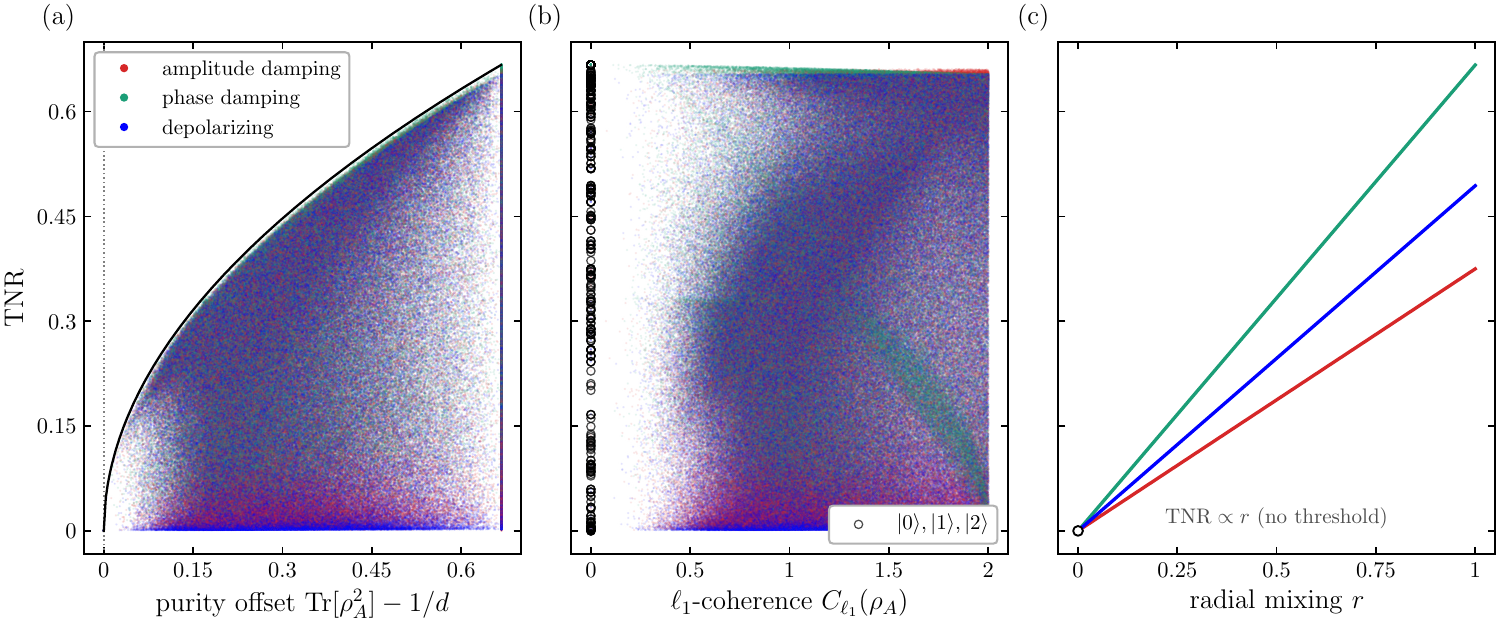}
\caption{The temporal-Bell resource is non-maximal mixedness, not
coherence (qutrit ensemble, $d=3$; $\rho_A$-adapted Monte-Carlo sweep
of $10^6$ configurations with random states, channels, and
continuous times). (a)~TNR against the purity
offset $\mathrm{tr}(\rho_A^2)-1/d$, with the analytic envelope
$\sqrt{\tfrac{d-1}{d}(\mathrm{tr}\rho_A^2-1/d)}$ (black) which the
adapted cloud now \emph{saturates} across the range:
every TNR$=0$ instance corresponds
exactly to $\rho_A=\openone/3$, and TNR$>0$ otherwise---the TNR-zero
region is exactly the singleton $\{\openone/3\}$ in the qutrit Bloch
ball (Corollary~\ref{cor:tnr_iff}), the headline result of this
paper; the same holds at $d=5$ (section~\ref{sec:numD5}). The solid
black curve is the analytic upper envelope
$\mathrm{TNR}=\sqrt{\tfrac{d-1}{d}\,[\mathrm{tr}(\rho_A^2)-1/d]}$,
saturated by phase damping along rays from $\openone/d$ to a
maximally-resourceful pure input (it follows from the radial law of
panel (c): $\mathrm{TNR}=c\,r$ with maximal slope $c=(d-1)/d$ and
offset $=\tfrac{d-1}{d}r^2$). (b)~TNR
versus $\ell_1$-coherence $C_{\ell_1}(\rho_A)=\sum_{i\neq j}|(\rho_A)_{ij}|$,
the sum of off-diagonal magnitudes of $\rho_A$ in the measurement
eigenbasis (a standard, basis-dependent coherence monotone). The
computational-basis states
$|0\rangle,|1\rangle,|2\rangle$ (rings) have zero coherence yet TNR up
to $(d-1)/d$, so coherence does not control the resource: the boundary
already \emph{reaches its maximum $(d-1)/d$ at zero coherence} and is
non-monotonic in coherence---the opposite of a coherence-monotone. (c)~Along
the radial ray $\rho_r=(1-r)\openone/d+r|0\rangle\langle0|$ from the
maximally mixed point, $\mathrm{TNR}$ is \emph{numerically exactly
linear} in $r$
for every channel (straight lines through the origin), so it vanishes
\emph{continuously}---$\mathrm{TNR}\propto r\propto
\sqrt{\mathrm{tr}(\rho_A^2)-1/d}$---with no finite threshold: the only
zero is $\rho_A=\openone/d$ itself. The apparent ``gap'' in a sparse
sweep is therefore a sampling artifact, not a real separation.}
\label{fig:tnr_zero}
\end{figure*}

The Leggett--Garg-type construction in
Theorem~\ref{thm:necessity} also shows that the back-action of
$\sqrt{M_{a|x}}$ alone supplies $\TNR>0$ without any non-trivial
channel structure. Together with Proposition~\ref{prop:TNRstatebound},
which closes off the channel-driven route on $\rho_A=\openone/d$, we
conclude:
\begin{itemize}
\item TNR \emph{requires} a non-maximally-mixed $\rho_A$ but not a
non-trivial $\mathcal{E}$;
\item TSR and TER can be supplied by either $\rho_A$ or $\mathcal{E}$
independently---both are jointly state--channel-bound without an
asymmetric prerequisite.
\end{itemize}
We therefore conclude that temporal nonlocality is \emph{state-bound
in an asymmetric sense}: the channel alone, however non-trivial,
cannot promote a behavior to Bell-nonlocality on a maximally mixed
input. The contrast with the spatial scenario---where nonlocality is a
property of the shared state alone---is illuminating: in the temporal
setting it is the structure of the input state, not of the
correlation-establishing channel, that carries the
device-independent resource.

This dichotomy has two immediate consequences. First, the closed-form
equality $\TSR=\TNR=(d-1)/d$ on $|0\rangle$+phase damping
(Proposition~\ref{prop:tsr_tnr_collapse}) reflects a regime in which the
state-bound TNR contribution to the behavior exhausts the entire
TSR resource: the channel-bound steering avenue is closed off by
NSIT violation, leaving only the input-state contribution, which is
the same on both sides of the hierarchy. Second, the
contextual-fraction connection \cite{Abramsky2017} sharpens:
deviations from a deterministic Alice-side response model require
non-uniform marginals $\{p_a\}$, which the maximally mixed input
forecloses by Proposition~\ref{prop:TNRstatebound}, while a
non-trivial channel does not by itself create such non-uniformity.

A central numerical observation is that whereas TER and TSR can be
nonzero on $\rho_A=\openone/3$, $\TNR$ vanishes there identically.
The contrast is sharp: on the empirical sweep
(figure~\ref{fig:tnr_zero}), the TNR-zero region is exactly the
singleton $\{\openone/3\}$, while TSR remains nonzero on a
substantial neighborhood of $\openone/3$ under amplitude or phase
damping. The TNR-positive region is therefore a strict subset of
the TSR-positive region, sharpening the universal lower hierarchy
of Theorem~\ref{thm:tsr_tnr}.

\section{The robustness hierarchy and the universal NSIT
bound}\label{sec:hier-frame}

The central result rests on, and is sharpened by, the relations among the
three robustnesses. We collect them here as the supporting framework for
the operational results of section~\ref{sec:ditit}.

\subsection{Hierarchy of two-time three-level correlations}
\label{sec:hierarchy}

The lower part of the temporal hierarchy generalizes cleanly to all
dimensions:
\begin{theorem}\label{thm:tsr_tnr}
For arbitrary initial state and channel, $\TSR(R)\geq\TNR(R)\geq 0$.
\end{theorem}

\noindent The inequality is confirmed numerically across all
$\rho_A$-adapted configurations in our sweep
(section~\ref{sec:numerics}): $\TSR-\TNR\geq 0$ in every instance
within solver tolerance.

The upper part requires more care. The qubit hierarchy proof of
Ku \emph{et al.}\ \cite{PhysRevA.98.022104} requires the relevant TER to quantify
robustness with respect to \emph{separability} of $R_{AB}$, not merely
its positivity. For qubits the Peres--Horodecki criterion
\cite{PhysRevLett.77.1413,HORODECKI19961} collapses positivity, PPT,
and separability, and the qubit hierarchy is automatic.\footnote{For
random sweeps at $d\geq 3$, where PPT is necessary but not
sufficient for separability (bound entanglement), we use the linear
correlation-matrix criterion of Sarbicki \emph{et al.}
\cite{PhysRevA.101.012341} as a complementary separability witness.
On the maximally mixed input $\rho_A=\openone/3$ under the three
standard channels the Sarbicki bound is satisfied whenever the
new-PDO of equation~\eqref{eq:new_PDO} is positive semidefinite, so the
separability-versus-positivity gap of $d\geq 3$ does not manifest
in practice for these channel families. The PDO of
equation~\eqref{eq:new_PDO} reduces to the standard qubit-Pauli PDO at
$d=2$.} For $d\geq 3$ these three notions are inequivalent. To formalize the proper
generalization we distinguish three quantifiers:
\begin{itemize}
\item $\TER^{\rm sep}$, equivalently $\TER^{\rm Choi}:=\mathrm{ER}(\Lambda_{\mathcal{E}})$,
the robustness with respect to admixture making the channel Choi
state separable---the one physically justified
temporal-entanglement quantifier.
\item $\TER^{\rm PPT}$, the SDP-computable \emph{proxy} obtained by
replacing separability with the Peres--Horodecki PPT condition; a
one-sided lower bound on $\TER^{\rm sep}$.
\item $\TER^{\PDO}$, the SDP \eqref{eq:TER2} quantifying
non-positivity of the new-PDO; a \emph{proxy} with no independent
physical justification, the PDO being generally non-positive (not a
quantum state).
\end{itemize}
For qubits all three coincide. For $d\geq 3$ they are distinct in
general, with $\TER^{\rm sep}\geq\TER^{\rm PPT}$ (separability is
strictly stronger than PPT in $d\geq 3$ due to bound-entangled states)
and no general ordering between $\TER^{\rm sep}$ and $\TER^{\PDO}$.
The proper qudit analogue of the qubit hierarchy is:
\begin{theorem}\label{thm:hierarchy_strict}
On $\rho_A=\openone/d$ and any quantum channel,
$\TER^{\rm sep}\geq\TSR\geq\TNR\geq 0$. The same chain holds with
$\TER^{\rm PPT}$ in place of $\TER^{\rm sep}$ whenever the Choi state
is on a manifold for which PPT is necessary and sufficient for
separability---e.g., the isotropic states reached by the depolarizing
channel.
\end{theorem}

We confirm this hierarchy numerically. On every sampled
$\rho_A=\openone/3$ point (15 configurations: 3 channels $\times$ 5
times), $\TER^{\rm PPT}\geq\TSR\geq\TNR\geq 0$ holds with strict
inequality wherever any of the three are nonzero
(section~\ref{sec:numHier}). The break observed off-$\openone/d$
(initially pure NSIT-violating states) is \emph{genuine}, not a
surrogate artifact: on the depolarizing/isotropic Choi states PPT is
necessary and sufficient for separability, so $\TER^{\rm sep}=\TER^{\rm PPT}$
there, and the observed $\TER^{\rm PPT}<\TSR$ is a true violation of
$\TER^{\rm sep}\geq\TSR$. The hierarchy holds exactly on the NSIT set
$\rho_A=\openone/d$ and its breakdown off that set is controlled by
the NSIT violation, not by the PPT/separability gap
(section~\ref{sec:nsit-free}).

\paragraph*{The new-PDO is the partial transpose of the Choi state.}
A by-product of our analysis is that the non-contextual Wigner-basis
PDO of equation~\eqref{eq:new_PDO} and the Choi--Jamio\l{}kowski state
are partial-transpose related rather than independent. Using
$\sum_i K_i\otimes K_i=d\,\mathrm{SWAP}$ and
$\mathrm{SWAP}=d\,(|\Phi^+\rangle\langle\Phi^+|)^{T_B}$, the identity
channel gives
\begin{equation}\label{eq:Rident}
R^{\PDO}_{\rm id}=\frac{1}{d^2}\sum_i K_i\otimes K_i
=\tfrac{1}{d}\,\mathrm{SWAP}=(|\Phi^+\rangle\langle\Phi^+|)^{T_B},
\end{equation}
the partial transpose of the rank-one maximally entangled Choi state,
with eigenvalues $\{+1/d\,(\times\tfrac{d(d+1)}{2}),\,
-1/d\,(\times\tfrac{d(d-1)}{2})\}$ (indefinite; $\{+\tfrac13\times6,
-\tfrac13\times3\}$ at $d=3$). Hence $\TER^{\rm Choi\text{-}pos}$
vanishes identically (Choi states are CP-positive) and is
uninformative, while $\TER^{\PDO}$ measures partial-transpose
negativity. This indefiniteness is present at every $d$, including
$d=2$; the genuinely dimension-sensitive feature is not the operators'
distinctness but the gap between PPT and separability (exact only at
$2\times2$ and $2\times3$), so the proper hierarchy quantifier is
$\TER^{\rm sep}$---lower-bounded by the SDP-computable
$\TER^{\rm PPT}$---on the Choi.
The asymmetric state-boundness of TNR established in
section~\ref{sec:TNRcoherence} is \emph{independent} of this discussion:
TSR and TNR are computed from the assemblage and behavior, neither
of which involves any PDO. Because projective back-action makes Bob's
marginal depend on Alice's setting, the behaviors are signaling in
time; $\TNR>0$ here certifies a no-signaling-in-time (NSIT) violation
rather than a genuine no-signaling Bell nonlocality.

\subsection{Hierarchy breaking off $\openone/d$: NSIT-conditionality of
the upper bound}\label{sec:NSITbreaks}

Pure initial states violate NSIT for at least one measurement basis,
and the violation has two distinct consequences. First, off
$\openone/d$ the genuine $\TER^{\rm sep}$ falls below $\TSR$ on
$\sim$22\,\% of sampled (state, time) configurations
(section~\ref{sec:numHier}), exclusively on NSIT-violating inputs---
exactly so on the depolarizing/isotropic Choi, where
$\TER^{\rm sep}=\TER^{\rm PPT}$. The same break appears in the PDO
proxy ($\TER^{\PDO}<\TSR$), where the new-PDO can be positive
semidefinite while the assemblage remains steerable. The upper hierarchy holds only on the
NSIT set $\rho_A=\openone/d$ (Theorem~\ref{thm:hierarchy_strict}); off
it the breakdown is governed by the NSIT violation, because $\TSR$
then carries a signaling contribution the Choi entanglement cannot
bound (section~\ref{sec:nsit-free}).

Second, on the precise (state, channel) pair $|0\rangle$ + phase
damping the gap between $\TSR$ and $\TNR$ \emph{closes} to numerical
zero. The closing is exact in $d$ and admits the closed-form

\begin{proposition}[$\TSR=\TNR=(d-1)/d$ on $|0\rangle$+phase damping]
\label{prop:tsr_tnr_collapse}
For $\rho_A=|0\rangle\langle 0|$ subjected to phase damping with
visibility $e=e^{-t}\in(0,1]$, two MUB measurements (computational and
Fourier of the eigenbasis of $\rho_A$) give the analytic upper bound
$\TNR\leq\TSR\leq(d-1)/d$ for every $e\in(0,1]$; the value is
saturated, $\TSR=\TNR=(d-1)/d$ independent of $e$, as confirmed by the
SDP optimum at $d=3,5$. (The bound is not tight at $e=0$, where the
fully dephased assemblage is unsteerable and $\TSR=0$, so the constant
$(d-1)/d$ holds on the open interval $e\in(0,1]$ with a discontinuity
at $e\to0^+$.)
\end{proposition}

The upper-bound value $(d-1)/d$ is therefore analytic and not a
measure-zero coincidence of the $d=3$ specialization; it follows from
the algebraic structure of $|0\rangle$ on its eigenbasis and the
phase-damping dephasing, with saturation confirmed numerically. The closed-form value $(d-1)/d$ matches the
numerical sweep at $d=3$ (the Supplementary Material) and the high-precision
test at $d=5$.

\medskip\noindent\emph{No NSIT-free upper hierarchy.}\label{sec:nsit-free}
\enspace Unlike $\TSR\geq\TNR$ (Theorem~\ref{thm:tsr_tnr}), which is
state-independent, the upper inequality $\TER^{\rm sep}\geq\TSR$ is
not, and \emph{no} state-independent (``NSIT-free'') version exists.
The proof of Theorem~\ref{thm:hierarchy_strict} uses
$\rho_A=\openone/d$ in an essential way. Writing
$\TER^{\rm Choi}(\mathcal{E})=\mathrm{ER}(\Lambda_{\mathcal{E}})$ for
the entanglement robustness of the channel's Choi state, the
entanglement-breaking measure-and-prepare form of the separablised
channel yields, for the assemblage, a candidate model
$\sum_k\tr\!\big(F_k\,\Pi_{a|x}\rho_A\Pi_{a|x}\big)\,\omega_k$ whose
hidden-variable weight $q(k|x)=\tr\!\big(F_k\,D_x(\rho_A)\big)$, with
$D_x(\rho_A)=\sum_a\Pi_{a|x}\rho_A\Pi_{a|x}$, \emph{depends on Alice's
setting} $x$. A genuine local-hidden-state model must be
no-signaling---its weight independent of $x$---so the construction is
valid only when $D_x(\rho_A)$ is $x$-independent, i.e.\ under NSIT,
which for the canonical two-MUB scheme means $\rho_A=\openone/d$
(Theorem~\ref{thm:nsit_state}). Off $\openone/d$ the assemblage
marginal $\sum_a\tilde\rho_{a|x}=\mathcal{E}(D_x(\rho_A))$ signals;
$\TSR$ then carries a signaling contribution that the Choi
entanglement---blind to the input---cannot bound.

Consequently the upper inequality genuinely \emph{fails} off
$\openone/d$, and for \emph{every} TER variant alike. For the
depolarizing channel with visibility $e=e^{-t}<1/(d+1)$ the Choi state
is separable, so the exact $\TER^{\rm sep}=\TER^{\rm Choi}=0$, yet
$\TSR=e\,(d-1)/d>0$; for $|0\rangle$+phase damping at strong dephasing
$\TSR\to(d-1)/d$ while $\TER^{\rm Choi}\to0$. The upper hierarchy
$\TER^{\rm sep}\geq\TSR\geq\TNR$ thus holds exactly on the NSIT set
$\{\rho_A=\openone/d\}$ (Theorem~\ref{thm:hierarchy_strict};
Corollary~\ref{cor:tnr_iff}), and its breakdown off that set is
governed by the NSIT violation $\mathcal{V}_{\rm NSIT}$, not by the
choice of TER quantifier.

This breakdown is, moreover, quantitatively bounded: the signaling
excess in $\TSR$ is at most \emph{half} the NSIT violation, which
restores a single inequality valid for \emph{every} input.

\begin{proposition}[NSIT-corrected universal upper bound]
\label{prop:univ_bound}
For $\rho_A$-adapted measurements and any input, channel and time,
\begin{equation}\label{eq:univ_bound}
0\;\leq\;\TNR\;\leq\;\TSR\;\leq\;\TER^{\rm sep}+\tfrac12\,
\mathcal{V}_{\rm NSIT},
\end{equation}
with the coefficient $\tfrac12$ tight: it is saturated on
$|0\rangle$+phase damping, where $\TER^{\rm sep}\to0$,
$\TSR=(d-1)/d$ and $\mathcal{V}_{\rm NSIT}=2(d-1)/d$.
\end{proposition}

\begin{proof}[Proof sketch]
The lower chain $0\leq\TNR\leq\TSR$ is Theorem~\ref{thm:tsr_tnr}. For
the upper bound one builds an explicit local-hidden-state over-cover of
the steering SDP \eqref{eq:TSR_SDP} from the measure-and-prepare form
$\mathcal{E}_{\rm sep}(\cdot)=\sum_k\tr(F_k\,\cdot)\,\omega_k$ of the
entanglement-breaking channel that separabilises
$\Lambda_{\mathcal{E}}$ at weight $\TER^{\rm sep}$. Its cost is
controlled by the ``traceless $\Rightarrow$ half the trace-norm''
identity applied to $\Delta=D_Z(\rho_A)-D_F(\rho_A)$, giving the
rigorous pre-channel bound $\TSR\leq\TER^{\rm sep}
+\tfrac12(1+\TER^{\rm sep})\|\Delta\|_1$ with the constant $\tfrac12$
tight. The sharp post-channel form \eqref{eq:univ_bound} replaces
$\|\Delta\|_1$ by $\mathcal{V}_{\rm NSIT}=\|\mathcal{E}(\Delta)\|_1$
(channel contraction) and drops the prefactor; this last step is the
only one resting on the numerics, verified with zero violations to
machine precision (minimum slack $-3.9\times10^{-8}$) across the full
$\rho_A$-adapted sweeps at $d=3$ ($10^6$) and $d=5$
($5.6\times10^4$), and at $d=2,4$ (Table~\ref{tab:d3d5_compare};
figure~\ref{fig:hier_nsit_compare}). The full construction is given in
the Supplementary Material.
\end{proof}

\subsection{Multi-time generalization}\label{sec:multitime}

In an $n$-time scenario with initial state $\rho_{A_1}$ and sequential
channels $\mathcal{E}_{1,2},\ldots,\mathcal{E}_{n-1,n}$ between events at
$t_1<\cdots<t_n$, the joint behavior is
\begin{equation}\label{eq:multitime}
\begin{aligned}
P(\vec a|\vec x)={}&
\tr\bigl[M_{a_n|x_n}\mathcal{E}_{n-1,n}\bigl(\cdots\\
&\;\;\mathcal{E}_{1,2}
(M_{a_1|x_1}^{1/2}\rho_{A_1}M_{a_1|x_1}^{1/2})\cdots\bigr)\bigr],
\end{aligned}
\end{equation}
$\vec a=(a_1,\ldots,a_n)$, $\vec x=(x_1,\ldots,x_n)$; let
$\TNR^{(n)}(\rho_{A_1},\mathcal{E})$ be the minimal admixture making
$P(\vec a|\vec x)$ admit an LHV decomposition with a single shared hidden
variable.

\begin{theorem}[Multi-time state-boundness, necessity]
\label{thm:multitime}
For any $n\geq 2$, any sequence of channels, and projective measurements
at every time, $\rho_{A_1}\neq\openone/d$ implies
$\TNR^{(n)}(\rho_{A_1},\mathcal{E})>0$.
\end{theorem}

The converse fails for $n\geq 3$: even on $\rho_{A_1}=\openone/d$ the
first measurement leaves the sub-normalized projector $\Pi_{a_1|x_1}/d$,
so the conditional state entering the next pair is generically not
$\openone/d$, and applying Theorem~\ref{thm:necessity} there yields a
non-LHV three-time behavior (Supplementary Material)---e.g.\
$\TNR^{(3)}(\openone/3,\mathrm{id},\mathcal{E}_{\rm phase})>0$ while
$\TNR^{(2)}(\openone/3,\mathcal{E})=0$. The clean
iff~(Corollary~\ref{cor:tnr_iff}) is thus \emph{intrinsically two-time};
for $n\geq3$ the natural analogue is the \emph{conjunctive} condition that
every conditional pre-measurement state equals $\openone/d$, preserved by
symmetric channels but not by amplitude or phase damping.

\section{Operational meaning and its fundamental
limit}\label{sec:ditit}

The asymmetric state-boundness of TNR has a sharp operational
meaning: TNR is the resource that controls
\emph{temporal teleportation}---the protocol that
delivers an unknown qudit state from $t_A$ to $t_B$ via a noisy
channel $\mathcal{E}$. We make the connection precise, then turn to its
fundamental limit---over-certification (section~\ref{sec:overcert}).
We stress at the outset that the protocol is device-independent in a
\emph{prepare-and-measure} sense: the test-round measurement bases are
prescribed (the two-MUB scheme tied to the certification state
$\rho_A$, not to the unknown state $|\psi\rangle$), and the
certification guarantees only the send-round fidelity conditioned on
passing the test---in contrast to spatial Bell-based DI, where the
measurement settings are free, untrusted parameters. The honest scope
of the guarantee is restated where it matters in
section~\ref{sec:discussion}.

A qubit-level, BB84-style picture of the protocol---a quantum one-time
pad sent through the time-channel and monitored by two-MUB test
rounds---is given in the Supplementary Material.

\paragraph*{Protocol DI-TIT (device-independent temporal teleportation).}
Alice at $t_A$ has access to a stream of $d$-dimensional qudit
inputs to the channel $\mathcal{E}$; Bob at $t_B$ has access to the
channel outputs. Alice holds an unknown qudit state $|\psi\rangle$
that she wishes to deliver to Bob with the highest possible
fidelity, using only untrusted measurement and preparation devices
and a public classical communication channel. Each round of DI-TIT
is one of two random branches.

\begin{enumerate}
\item[(T)] \emph{Test round} (probability $\gamma$). Alice
prepares the certification state $\rho_A$ on her qudit, samples
$x\in\{0,1\}$ uniformly, and projectively measures her qudit in the
basis $\{\Pi_{a|x}\}_{a=0}^{d-1}$, recording outcome $a$.
The post-measurement state propagates through $\mathcal{E}$ to
$t_B$. Bob samples $y\in\{0,1\}$ uniformly and projectively measures
in basis $\{M_{b|y}\}_{b=0}^{d-1}$, recording outcome $b$.

\item[(S)] \emph{Send round} (probability $1-\gamma$). Alice
prepares the unknown state $|\psi\rangle$ on her qudit, samples
$k=(k_1,k_2)\in\mathbb{Z}_d\times\mathbb{Z}_d$ uniformly at random,
and applies the Heisenberg--Weyl operator
$W_k=X^{k_1}Z^{k_2}$ (where $X|j\rangle=|j+1\bmod d\rangle$ and
$Z|j\rangle=\omega^j|j\rangle$ with $\omega=e^{2\pi i/d}$),
yielding the encoded qudit $W_k|\psi\rangle$. The encoded qudit
propagates through $\mathcal{E}$ to $t_B$. Alice publicly transmits
$k$ to Bob, who applies $W_k^\dagger$ to the channel output:
$\rho_{\rm out}^{(k)}=W_k^\dagger\,
\mathcal{E}(W_k|\psi\rangle\langle\psi|W_k^\dagger)\,W_k$.
\end{enumerate}

\noindent After $N$ rounds, Alice and Bob estimate the empirical
behavior $\hat P(a,b|x,y)$ from the
$\sim\gamma N$ test rounds and compute
$\hat T=\widehat{\TNR}(\hat P)$ by solving the SDP
\eqref{eq:TNR_SDP}. They accept the session if
$\hat T\geq T^*$ for a chosen threshold $T^*>0$, otherwise abort.
Conditioned on acceptance, the average state delivered to Bob over
the $\sim(1-\gamma) N$ send rounds is the
Heisenberg--Weyl-twirled output
\begin{equation}\label{eq:hwtwirl}
\bar\rho_{\rm out}=\frac{1}{d^2}\sum_{k\in\mathbb{Z}_d\times\mathbb{Z}_d}
W_k^\dagger\,\mathcal{E}(W_k|\psi\rangle\langle\psi|W_k^\dagger)\,W_k,
\end{equation}
which is the Heisenberg--Weyl-twirled (depolarizing-type) image
of $|\psi\rangle$ under $\mathcal{E}$. The Heisenberg--Weyl
twirling is the standard \emph{quantum one-time pad} construction
of Hayden \emph{et al.}~\cite{HaydenLeungShorWinter2004}: the
encoded ensemble $\{W_k|\psi\rangle\langle\psi|W_k^\dagger\}_k$ is
indistinguishable from $\openone/d$ to any party without the
shared key $k$, so the send-round transmission carries quantum
information only through the channel $\mathcal{E}$ itself.
Combined with the send-round Heisenberg--Weyl twirl, the protocol
delivers the depolarized channel image of $|\psi\rangle$
(Theorem~\ref{thm:ditit}); the test rounds certify $\TNR$, which
lower-bounds the twirl fidelity $\mathcal{F}_{\rm DI}$ on the standard
(twirl-covariant) noise families.

\paragraph*{Qubit-case verification.}
The protocol structurally parallels prepare-and-measure DI-QKD
\cite{Acin2007PRL} at $d=2$. Take
$W_k\in\{\openone,X,Y,Z\}$ (the four single-qubit Pauli operators
identified with the $d^2=4$ Heisenberg--Weyl elements
$X^{k_1}Z^{k_2}$ for $(k_1,k_2)\in\mathbb{Z}_2^2$); the test rounds
use the two MUBs $\{|0\rangle,|1\rangle\}$ (the $Z$ eigenbasis) and
$\{|+\rangle,|-\rangle\}$ (the $X$ eigenbasis), and the verification
witness is the temporal Bell-violation argument of
Theorem~\ref{thm:necessity}, which at $d=2$ reduces to the
Leggett--Garg-type contradiction on Bob's marginal once Alice's
input has a non-uniform diagonal in the chosen basis. The
Heisenberg--Weyl twirling in the send rounds reduces, via
$W_k|\psi\rangle\langle\psi|W_k^\dagger$ averaged uniformly over
$k$, to standard Pauli-twirling. On the identity channel the
twirl acts trivially, $\bar\rho_{\rm out}=|\psi\rangle\langle\psi|$,
giving $\mathcal{F}_{\rm DI}=1$ exactly (the noiseless limit $p=1$); on a
noisy channel the Pauli-twirl produces a depolarized image
$\bar\rho_{\rm out}=p|\psi\rangle\langle\psi|+(1-p)\openone/2$
with $p$ determined by the channel's average entanglement
fidelity. The qubit DI-TIT protocol therefore reduces to a
prepare-measure-twirl-broadcast arrangement structurally
analogous to the BB84-with-Bell-monitoring schemes of
\cite{Acin2007PRL}, with the test inequality replaced by the
temporal-Bell witness of Theorem~\ref{thm:necessity}.

\begin{theorem}[DI temporal-teleportation fidelity]
\label{thm:ditit}
Let $\mathcal{F}_{\rm DI}(\mathcal{E})$ be the asymptotic average
fidelity of the prepare-and-measure DI-TIT protocol of
section~\ref{sec:ditit} (Heisenberg--Weyl one-time-pad encoding,
public key, send through $\mathcal{E}$). Then $\mathcal{F}_{\rm DI}$
is exactly the depolarizing fidelity set by the channel's
entanglement fidelity
$F_e(\mathcal{E})=\langle\Phi^+|(\mathrm{id}\otimes\mathcal{E})
(\Phi^+)|\Phi^+\rangle$,
\begin{equation}\label{eq:ditit_exact}
\mathcal{F}_{\rm DI}(\mathcal{E})
=\frac{1}{d}+\frac{d-1}{d}\,p,
\qquad
p=\frac{d^2F_e(\mathcal{E})-1}{d^2-1}\in[0,1].
\end{equation}
In particular $\mathcal{F}_{\rm DI}>1/d \Leftrightarrow p>0
\Leftrightarrow F_e(\mathcal{E})>1/d^2$: temporal teleportation beats
the no-resource baseline iff the channel transmits coherence above
the fully-depolarizing floor.
\end{theorem}

\begin{proposition}[State-bound certification of the fidelity]
\label{prop:ditit_cert}
For the standard channel families, on the set of certifiers delimited in
Proposition~\ref{prop:mixsuff} the certified robustness lower-bounds the
twirl parameter, $\TNR(\rho_A,\mathcal{E})\leq p$, so that
\begin{equation}\label{eq:ditit_bound}
\mathcal{F}_{\rm DI}(\mathcal{E})\;\geq\;
\frac{1}{d}+\frac{d-1}{d}\,T,\qquad T:=\TNR(\rho_A,\mathcal{E}).
\end{equation}
The guaranteed fidelity rises with the certified value up to its supremum
$\mathcal{F}_{\rm DI}=(d^2-d+1)/d^2=7/9$ at $d=3$, attained at the maximal
\emph{honest} value $T=p=(d-1)/d$---the $|0\rangle$+phase-damping point
where the bound is tight, on the edge of over-certification
(section~\ref{sec:overcert}). Off that certifier set a channel-protected
probe over-certifies and the bound fails.
\end{proposition}

The certifier set is controlled by the probe's mixedness through a
universal cap on the robustness.

\begin{lemma}[Universal robustness cap]\label{lem:cap}
For every channel $\mathcal{E}$, every dimension $d$, and every input
$\rho$, with the two MUBs adapted to $\rho$,
\begin{equation}\label{eq:cap}
\begin{aligned}
\TNR(\rho,\mathcal{E}) &\leq \TSR(\rho,\mathcal{E}) \leq M(\mathcal{E}), \\
M(\mathcal{E}) &:= \max_k \tfrac12 \bigl\|\mathcal{E}(\Delta_k)\bigr\|_1 \leq \frac{d-1}{d},
\end{aligned}
\end{equation}
where $\Delta_k:=|k\rangle\langle k|-\openone/d$.
The cap $(d-1)/d$ is saturated by phase damping (any strength) and by
Heisenberg--Weyl unitaries; $M(\mathcal{E})=p\,(d-1)/d$ for the
depolarizing channel, while amplitude damping lies strictly below.
\end{lemma}

\begin{proof}
By unitary invariance of the adapted scheme we may take
$\rho=\mathrm{diag}(\lambda)$ with computational and Fourier MUBs; $\TSR$
is then the value of an SDP linear in $\lambda$, hence convex, with
supremum at a vertex $\rho=|k\rangle\langle k|$. There $\mu_{a|Z}=\delta_{ak}$,
$\mu_{a|F}=1/d$, and the hidden-state assignment
$\sigma_{(k,a')}=\tfrac1d\mathcal{E}(\Pi_{a'|F})+c_{a'}$, with
$c_{a'}\succeq0$ and $\sum_{a'}c_{a'}=(\mathcal{E}(\Delta_k))_+$, is
feasible for the steering program~\eqref{eq:TSR_SDP}: the $F$-constraints
hold as $c_{a'}\succeq0$, and the $Z$-constraint as
$\sum_{a'}\sigma_{(k,a')}=\mathcal{E}(\openone/d)+(\mathcal{E}(\Delta_k))_+
\succeq\mathcal{E}(|k\rangle\langle k|)$. Its cost is
$\tr(\mathcal{E}(\Delta_k))_+=\tfrac12\|\mathcal{E}(\Delta_k)\|_1$, because
$\mathcal{E}(\Delta_k)$ is traceless. Trace-norm contractivity
$\|\mathcal{E}(\Delta_k)\|_1\leq\|\Delta_k\|_1=2(d-1)/d$ and
$\TNR\leq\TSR$ (Theorem~\ref{thm:tsr_tnr}) give~\eqref{eq:cap}.
\end{proof}

\begin{proposition}[Mixedness suffices for honest certification]
\label{prop:mixsuff}
For the standard channel families, in the $\rho_A$-adapted scheme, the
certified robustness is bounded by the probe's mixedness through the cap,
\begin{equation}\label{eq:mixsuff}
\TNR(\rho_A,\mathcal{E})\;\leq\;(1-d\,\lambda_{\min})\,M(\mathcal{E})
\;\leq\;(1-d\,\lambda_{\min})\,\frac{d-1}{d},
\end{equation}
with $\lambda_{\min}:=\lambda_{\min}(\rho_A)$.
Hence the certificate~\eqref{eq:ditit_bound} is honest ($\TNR\leq p$)
whenever the probe is mixed enough,
\begin{equation}\label{eq:mixcond}
\lambda_{\min}(\rho_A)\;\geq\;\frac{1}{d}-\frac{p}{d-1} .
\end{equation}
For the \emph{depolarizing} channel $M=p\,(d-1)/d$, so
$\TNR\leq(1-d\lambda_{\min})\,p\,(d-1)/d<p$ for \emph{every} input
(honest unconditionally). For \emph{phase damping} $M=(d-1)/d$, saturated
at the eigenstate $|0\rangle$ ($\TNR=(d-1)/d$,
Proposition~\ref{prop:tsr_tnr_collapse})---so a probe too pure and aligned
with the dephasing basis over-certifies.
\end{proposition}

\begin{proof}
Write $\rho_A=d\lambda_{\min}(\openone/d)+(1-d\lambda_{\min})\tau$; the
state $\tau$ is co-diagonal with $\rho_A$, hence measured in its own
adapted scheme. The behavior is affine in the input at fixed settings and
$\TNR$ is convex in it, so
$\TNR(\rho_A,\mathcal{E})\leq d\lambda_{\min}\,\TNR(\openone/d,\mathcal{E})
+(1-d\lambda_{\min})\,\TNR(\tau,\mathcal{E})$. On the standard families
$\TNR(\openone/d,\mathcal{E})=0$ (Proposition~\ref{prop:product_lhv}) and
$\TNR(\tau,\mathcal{E})\leq M(\mathcal{E})$ by Lemma~\ref{lem:cap}, which
proves~\eqref{eq:mixsuff}; \eqref{eq:mixcond} follows from $M\leq(d-1)/d$.
The depolarizing value $M=p(d-1)/d$ is the cap on
$\mathcal{E}_{\rm depol}(\Delta_k)=p\,\Delta_k$, and the phase value
$M=(d-1)/d$ follows from $\mathcal{E}_{\rm phase}(\Delta_0)=\Delta_0$ (the
traceless part is diagonal, fixed by dephasing).
\end{proof}

\noindent\emph{Closed forms and onset.} For the standard families the twirl
parameter is $p_{\rm depol}=e$, $p_{\rm phase}=(1+de)/(d+1)$, and (at
$d=3$, cascade amplitude damping)
$p_{\rm amp}=[(1+\eta+\eta^2)^2-1]/8$ with $\eta=e^{-\kappa t/2}$,
$e=e^{-\kappa t}$. A channel-protected probe over-certifies once $p$ falls
below its pinned $\TNR$: the phase-damping eigenstate ($T=(d-1)/d$) does so
for $e<(d^2-d-1)/d^2$, i.e.\ $t>\ln\tfrac{9}{5}\approx0.59$ at $d=3$, and
the amplitude-damping $|+\rangle$ near $t\approx0.85$.

\paragraph*{Finite-statistics statement.}
The asymptotic statement of Theorem~\ref{thm:ditit} extends to a
finite-rounds bound through the standard Hoeffding/Azuma analysis
of device-independent protocols. With $N$ rounds at test fraction
$\gamma$, the empirical $\hat T = \widehat{\TNR}(\hat P)$ converges
to $T$ at rate $|\hat T - T| = O(1/\sqrt{N\gamma})$ with confidence
$1-\delta$, $\delta$ exponentially small in $N\gamma$. Combined
with the strict monotonicity of $T \mapsto 1/d+(d-1)T/d$, the achievable
fidelity in the accepted session lies within $O(1/\sqrt{N\gamma})$
of the asymptotic bound \eqref{eq:ditit_bound}; explicit
constants follow from the entropy-accumulation framework of
Arnon-Friedman \emph{et al.}~\cite{ArnonFriedman2018}.

\begin{figure*}[tb]
\centering
\includegraphics[width=\linewidth]{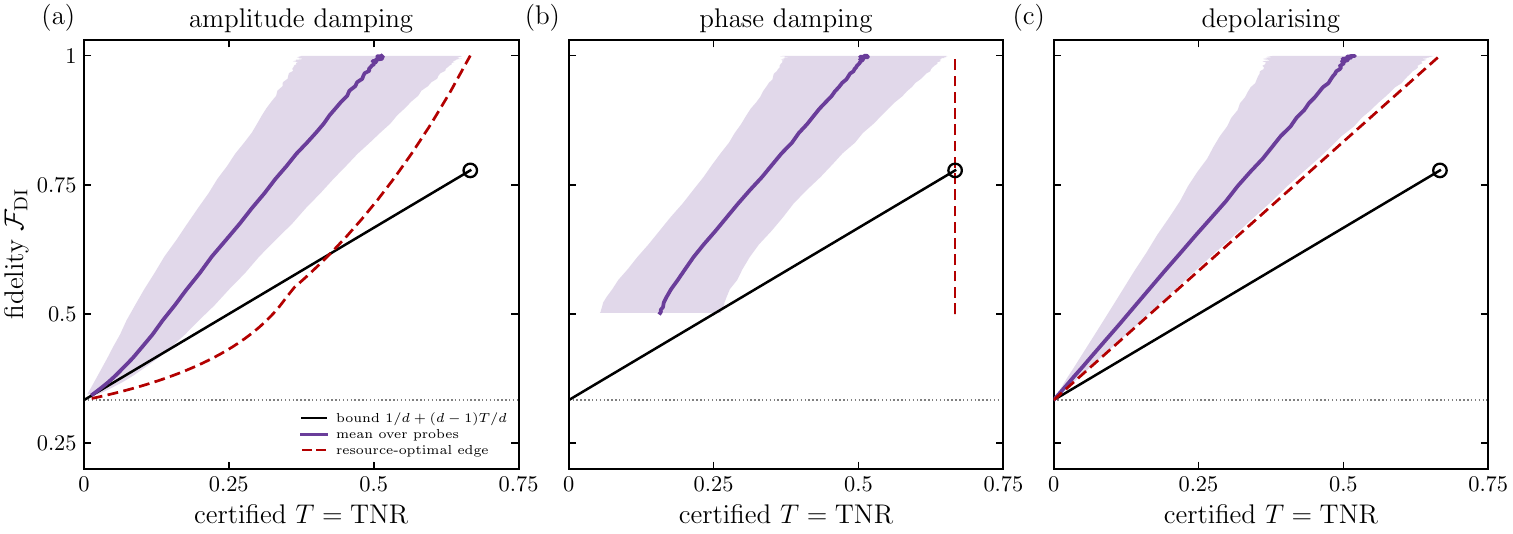}
\caption{Temporal-teleportation fidelity $\mathcal{F}_{\rm DI}=1/d+(d-1)p/d$
versus the certified $T=\TNR(\rho_A,\mathcal{E})$ at $d=3$, averaged over
random (Haar-pure and mixed) probe states, one column per channel. The
fidelity is fixed by the channel (the Heisenberg--Weyl twirl), while $T$
depends on the probe. A probe is a \emph{valid certifier} only when the
channel does not specially protect it; an invariant (channel-protected)
probe \emph{over-certifies}---its curve dips below the
bound~\eqref{eq:ditit_bound} (Proposition~\ref{prop:ditit_cert}). The
purple curve is the mean and the band is $\pm1\sigma$ over the random
probes; the black line is the bound; the red dashed curve is the exact
resource-optimal edge---the maximal certified $T$ over probes at each
channel time---which for phase damping is \emph{vertical} at
$T=(d-1)/d$, the $|0\rangle$+dephasing family holding $\TNR$ constant
while $\mathcal{F}_{\rm DI}$ ranges from the floor $2/(d+1)$ (here
$\tfrac12$) up to $1$; the circle marks the supremal
guaranteed point $\mathcal{F}_{\rm DI}=7/9$ at $T=(d-1)/d$, on the edge
of the certifiable regime. Both the mean and the resource-optimal edge
reach $\mathcal{F}_{\rm DI}=1$ only in the identity limit ($t\to0$, lossless
channel $p=1$): the edge at the maximal certified $T=(d-1)/d$, the mean at
the smaller average $T=\langle\mathrm{TNR}\rangle$. Perfect temporal
teleportation thus demands a \emph{lossless} channel; only the bound stays
strictly below $1$ (it tops at $7/9$ over the certifiable range
$T\le(d-1)/d$). The
\emph{typical} probe is an honest certifier---the mean lies on or above the
bound in every channel. The
individual-probe pathologies (specific eigenstate and superposition
inputs whose curves reveal the over-certification mechanism) are
displayed in the Supplementary Material.}
\label{fig:fdi}
\end{figure*}

\paragraph*{Reading the bound.}
At the boundary value $T=(d-1)/d$---realized by the
$|0\rangle$+phase-damping family at the dephasing time where
$p=T$---the certification bound~\eqref{eq:ditit_bound} meets the
exact fidelity~\eqref{eq:ditit_exact},
\begin{equation}\label{eq:ditit_canonical}
\mathcal{F}_{\rm DI}=\frac{1}{d}+\frac{d-1}{d}\cdot\frac{d-1}{d}
=\frac{d^2-d+1}{d^2},
\end{equation}
equal to $7/9$ at $d=3$; this is the supremal guaranteed fidelity, on
the edge of the certifiable regime (just beyond it, where the input is
$\mathcal{E}$-invariant, $T$ over-certifies and the bound fails;
section~\ref{sec:overcert}). Within the regime the
bound~\eqref{eq:ditit_bound} reduces to the classical baseline
$1/d$ as $T\to 0$ ($\rho_A\to\openone/d$). The protocol parallels
the spatial CHSH-bounded device-independent
quantum-key-distribution arguments of
Ac{\'\i}n \emph{et al.}~\cite{Acin2007PRL}.

\paragraph*{Fidelity, channel entanglement, and NSIT.}
The teleportation fidelity is set entirely by the channel: both
$\mathcal{F}_{\rm DI}$ and the entanglement robustness $\TER^{\rm sep}$
of the Choi state are channel-only quantities, and they move together.
For the depolarizing channel they are related \emph{affinely},
\begin{equation}\label{eq:fdi_ter}
\mathcal{F}_{\rm DI}=\frac{1}{d}
+\frac{\TER^{\rm sep}+(d-1)/d}{\,d+1\,}
\end{equation}
($\mathcal{F}_{\rm DI}=\tfrac12+\tfrac14\TER^{\rm sep}$ at $d=3$),
rising linearly from $\mathcal{F}_{\rm DI}=2/(d+1)$ ($=1/2$ at $d=3$) at
the separability threshold $\TER^{\rm sep}=0$ (depolarizing channel at
the separability boundary $p=1/(d+1)$, where $F_e=1/d$) to perfect
teleportation $\mathcal{F}_{\rm DI}=1$ at the identity channel
$\TER^{\rm sep}=d-1$. (The classical baseline $1/d$ itself is attained only
at the fully depolarizing point $p=0$; throughout the entanglement-breaking
region $p<1/(d+1)$ the robustness $\TER^{\rm sep}=0$, while
$\mathcal{F}_{\rm DI}$ ranges from $1/d$ up to $2/(d+1)$.)
The channel's entanglement robustness \emph{is} its
temporal-teleportation resource: more Choi entanglement means a higher
achievable fidelity, available in full when the devices are trusted.

The robustness tiers then certify this single channel fidelity under
decreasing trust. The fully device-independent floor uses $\TNR$
[equation~\eqref{eq:ditit_bound}], and the universal bound
$0\leq\TNR\leq\TSR\leq\TER^{\rm sep}+\tfrac12\mathcal{V}_{\rm NSIT}$
(Proposition~\ref{prop:univ_bound}) pins down what NSIT does here: a
maximally mixed probe ($\mathcal{V}_{\rm NSIT}=0$) certifies nothing,
$\TNR=0$, and the floor collapses to $1/d$; a probe that violates NSIT
unlocks certification, and $\mathcal{V}_{\rm NSIT}$ bounds how far the
steering resource can run \emph{above} the channel entanglement---the
signaling content of the assemblage that $\TER^{\rm sep}$, blind to
the input, cannot see. The NSIT violation is therefore the resource
that converts the channel's intrinsic teleportation capability
$\TER^{\rm sep}$ into a device-independently certifiable fidelity floor
through $\TNR$. These three quantities close a single operational
triangle---the fidelity $\mathcal{F}_{\rm DI}$ one wants, the channel
entanglement $\TER^{\rm sep}$ that supplies it, and the NSIT violation
$\mathcal{V}_{\rm NSIT}$ that certifies it---around which the rest of
the paper is organized.

Two further points are developed in the Supplementary Material: that
the NSIT violation $\mathcal{V}_{\rm NSIT}$ is a bona fide physical
resource---causal measurement back-action, accessed by the protocol's
public classical channel or by postselection, not a no-signaling
paradox---and that a one-sided device-independent variant (only Bob's
device untrusted) certifies the analogous steering-fraction bound
$\mathcal{F}_{\rm 1sDI}\geq\mathcal{F}_{\rm DI}$, strictly easier to
operate but requiring Alice's device to be trusted.

\subsection{Over-certification: when the certified value outruns the
channel}\label{sec:overcert}

It is worth pausing on a subtlety that decides when the certificate of
the previous subsection can be trusted. The certified number $T=\TNR$
and the fidelity $\mathcal{F}_{\rm DI}$ it is supposed to guarantee are
not the same kind of object, and now and then the certificate promises
more than the channel can deliver. We call this \emph{over-certification}.
The everyday version is certifying a leaky bucket by checking that it
holds a brick: the bucket passes perfectly, yet it still will not hold
water. The rest of this subsection explains why the same thing can
happen to a teleportation certificate, when it happens, and how to keep
it from happening.

\paragraph*{What is being compared.}
Two numbers meet in the bound~\eqref{eq:ditit_bound}. The first,
$\mathcal{F}_{\rm DI}$, is the fidelity for sending an \emph{unknown}
state through the channel; by Theorem~\ref{thm:ditit} it is an
\emph{average} over all the states one might send (this is what the
Heisenberg--Weyl twirl computes), so it depends on the channel alone.
The second, $T=\TNR(\rho_A,\mathcal{E})$, is read off from \emph{one}
chosen test state $\rho_A$ and the two fixed measurements made on it.
The certificate works by claiming that this single test never flatters
the channel, i.e.\ $T\leq p$, where $p$ is the channel's averaged
quality~\eqref{eq:ditit_exact}. \emph{Over-certification} is just the
case where this claim fails, $T>p$: the test looks better than the
channel really is, so the promised fidelity~\eqref{eq:ditit_bound} sits
above the true one~\eqref{eq:ditit_exact}. In Fig.~\ref{fig:fdi} it
shows up as a curve that dips \emph{below} the bound.

\paragraph*{Why it happens.}
The test is fair only if the channel treats the test state the way it
treats a typical state. Trouble starts when the channel \emph{protects}
the test state---keeps it intact far better than average. Then the test
keeps scoring well even as the channel quietly gets worse at everything
else. The clearest case is a test state the channel leaves completely
untouched (a \emph{fixed point}). Phase damping, for instance, slowly
erases coherence but leaves the basis states
$\ket{0},\ket{1},\ket{2}$ exactly in place. Certify with $\ket{0}$ and
the channel ``stores'' it perfectly forever, so $T$ stays at its maximum
$(d-1)/d$ for all times---while the channel is busy destroying the very
coherences that teleporting a superposition needs, dragging
$\mathcal{F}_{\rm DI}$ down to $2/(d+1)$ ($=1/2$ at $d=3$, already below
the best honest value $7/9$). We measured the one thing the channel is
perfect at and called it good: the brick in the leaky bucket.

Which states are ``protected'' depends on the channel (Fig.~\ref{fig:fdi}).
Phase damping protects its basis states, so those over-certify.
Amplitude damping instead drains everything toward the ground state, so
the even superposition $\ket{+}$ is the sheltered probe and
over-certifies, while the basis states stay honest. The depolarizing
channel treats every direction alike---it has no favorite to hide
behind---so no state can over-certify it and $T\leq p$ always. This is
the same fixed-point condition as in Proposition~\ref{prop:ditit_cert},
seen from the other side: a fixed point is the one place the channel is
not invertible on the test state, and there the clean rule
$\TNR=0\Leftrightarrow\rho_A=\openone/d$
(Corollary~\ref{cor:tnr_iff}) stops carrying over to what the channel
actually transmits.

\paragraph*{What it means in practice.}
This matters whenever $\TNR$ is used as a guarantee one cannot otherwise
check---temporal teleportation here, but equally a $\TNR$-certified
secure link or quantum-memory test. Three lessons follow.
\emph{(i) Choose the test state with care.} Avoiding a channel-fixed
test state is \emph{necessary} but, we find, not sufficient. The clean
sufficient rule is \emph{mixedness}: by Proposition~\ref{prop:mixsuff} a
probe mixed past the threshold $\lambda_{\min}(\rho_A)\geq 1/d-p/(d-1)$
certifies honestly on the standard families (and \emph{every} input does
for the depolarizing channel).
Since the user prepares $\rho_A$, this costs nothing---deliberately mixing
a little white noise into the probe guarantees honesty with no knowledge
of the channel. No test state is, however, a worst-case guarantee against
a completely arbitrary channel (next paragraph). \emph{(ii) The same effect sets the ceiling on
what can be promised.} The best certifiable value $T=(d-1)/d$ is reached
only right at the edge where the test stops being fair, $p=T$
[equation~\eqref{eq:ditit_canonical}]; that is exactly why the supremal
honest fidelity is $7/9$ at $d=3$, and why any genuinely fair test gives
a little less. So over-certification is not just a pitfall---it is the
boundary that \emph{pins} this supremum. \emph{(iii) It can be caught.}
Because over-certification needs a protected (frozen) test state, the
tell-tale sign is a certified $T$ that refuses to fall as the channel
degrades; a value stuck at $(d-1)/d$ is a red flag. Underneath all of
this is one clean message: what you \emph{certify} (the input's temporal
nonclassicality, $\TNR$) and what you \emph{use up} (the channel's
ability to carry coherence, $F_e$) are different resources, and they
agree only when the test state is not protected---which is precisely the
condition the certificate assumes.

\paragraph*{Can a clever choice of probes eliminate it?}
It is natural to hope that spreading the test over many well-chosen probes
removes over-certification. In the worst case it cannot be done---validity
is a property of the channel, not of the probe set. First, no \emph{single}
fixed probe is safe: for any $\rho_A$ a channel that fixes it (dephasing in
an aligned basis) over-certifies. Averaging over an ensemble does not rescue
it. A natural candidate is a $2$-design ($d^2$ states, a symmetric
informationally-complete measurement (SIC)~\cite{RenesBlumeKohout2004}, or
the $d+1$ mutually unbiased bases~\cite{WoottersFields1989}), because a $2$-design reproduces the
channel's average \emph{survival} fidelity, and hence $p$,
exactly~\cite{Scott2008}. But $\TNR$ is \emph{not} a survival fidelity---it
is unchanged by post-channel relabelings that destroy fidelity---so the
$2$-design property does not transfer. We verified directly ($d=2,3$) that a
complete $2$-design is over-certified both by generic channels and, once its
axes are not aligned with the measurement, by ordinary phase damping; an
informationally-complete \emph{non}-design behaves identically, so neither
$d^2$ nor the design property is the relevant quantity. The decisive
obstruction is an injective \emph{unitary}: the Heisenberg--Weyl shift has
$p=-1/(d^2-1)<0$ while $\TNR>0$ on every probe, over-certifying \emph{any}
ensemble. Honest certification is therefore confined to the twirl-covariant
standard families of Proposition~\ref{prop:ditit_cert}, and even there it
requires a \emph{sufficiently mixed} probe: by
Proposition~\ref{prop:mixsuff} the certificate is honest once
$\lambda_{\min}\geq 1/d-p/(d-1)$ (unconditionally for the depolarizing
channel). A near-pure probe aligned with the
channel's protected direction over-certifies; the honest mean of
Fig.~\ref{fig:fdi}(d--f) reflects exactly this---its spread is pulled
below the bound by the mixed members of the ensemble. Outside that channel
class no choice of probes restores the bound, and a tight characterization
of the certifiable channel class is left open.

\section{Numerical evidence}\label{sec:numerics}

We tested the framework on large $\rho_A$-adapted Monte-Carlo sweeps
of qutrit (state, channel, time) configurations: $10^6$ \emph{full}
configurations (each yielding TER, TSR, TNR, and the monitor
$\mathcal{V}_{\NSIT}$ from a $9\times9$ PDO SDP, an $\sim81$-strategy
TSR SDP, and an $\sim81$-vertex TNR LP) and $10^6$ TNR-only
configurations, with random states (Haar-pure and Ginibre-mixed at
ranks $2,3$), random channels, and continuous random times, plus the
canonical anchors $\{|0\rangle,|1\rangle,|2\rangle,\openone/3,
|+\rangle\}$. Every input is measured in its own eigenbasis and the
Fourier conjugate (the adapted scheme of
Corollary~\ref{cor:tnr_iff}); the SDPs are solved as prebuilt
DPP-parameterized problems and the LP via prebuilt HiGHS, parallelised
over $28$ cores (section~\ref{sec:soft}).

\subsection{State-boundness of TNR}\label{sec:numTNR}

Figure~\ref{fig:tnr_zero} (section~\ref{sec:TNRcoherence}) shows the
TNR-zero region across a $\rho_A$-adapted Monte-Carlo sweep of
$10^6$ configurations
(random states, channels, and continuous times). TNR is zero (within
$10^{-4}$) precisely on the configurations at $\rho_A=\openone/3$ and
is positive on every other state, vanishing \emph{continuously} as
$\rho_A\to\openone/3$: the radial sweep
$\rho_r=(1-r)\openone/d+r|0\rangle\langle0|$ in
figure~\ref{fig:tnr_zero}(c) shows TNR is \emph{numerically exactly
linear} in
$r$ (hence $\propto\sqrt{\mathrm{tr}(\rho_A^2)-1/d}$), so any apparent
finite floor is an artifact of how close a coarse sample lands to
$\openone/3$. Because the measurement is adapted to the eigenbasis of
$\rho_A$, the sweep tests exactly the iff of
Corollary~\ref{cor:tnr_iff} (TNR${}=0\Leftrightarrow\rho_A=\openone/3$),
with no reliance on avoiding the fixed-scheme zero family (remark after
Corollary~\ref{cor:three_way}); under adaptation the cloud also
\emph{saturates} the $\sqrt{\cdot}$ envelope across the whole purity
range. The corresponding NSIT-violation
$\mathcal{V}_{\NSIT}$ is exactly zero on the maximally mixed input
and positive---vanishing continuously---whenever the input has
nonzero purity offset, sustaining the operational reading:
device-independent temporal nonclassicality requires departure from
the
maximally-randomized input.

\subsection{Hierarchy gaps and NSIT}\label{sec:numHier}

\begin{figure*}[t]
\includegraphics[width=\linewidth]{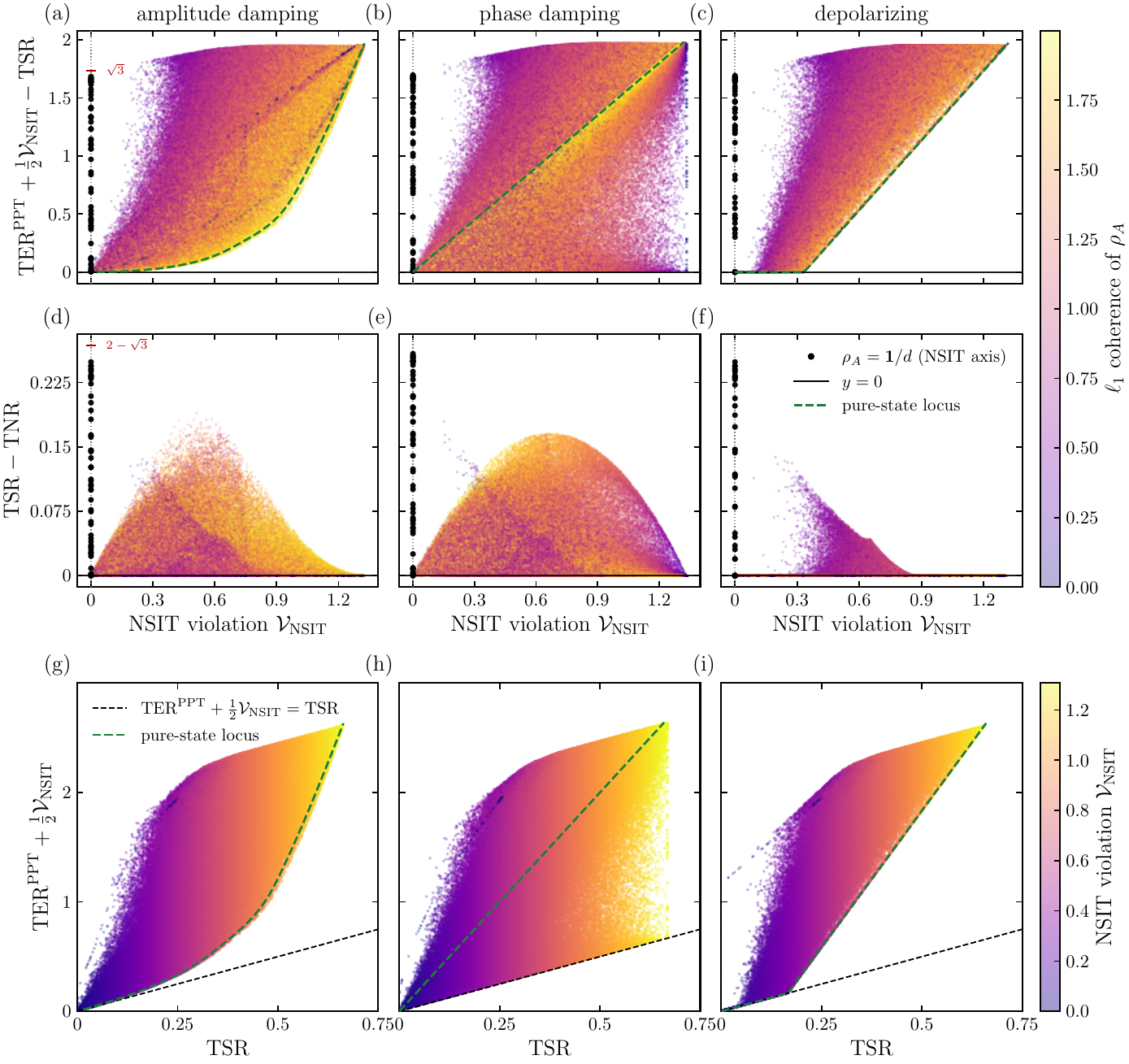}
\caption{The two universal bounds over a $\rho_A$-adapted Monte-Carlo
sweep of $10^6$ configurations (random states, channels and
continuous times) at $d=3$; columns are the three channels (amplitude
damping, phase damping, depolarizing). The NSIT-violation monitor
$\mathcal{V}_{\NSIT}$ [equation~\eqref{eq:Vnsit}] is the temporal
analogue of signaling, vanishing iff $\rho_A=\openone/d$
(Theorem~\ref{thm:nsit_state}). The plotted entanglement robustness is the
SDP-computable PPT proxy $\TER^{\rm PPT}\leq\TER^{\rm sep}$, so each bound
shown implies the corresponding separability-robustness statement of the
theorems. \emph{Panels (a)--(f):} the slacks
$\TER^{\rm PPT}+\tfrac12\mathcal{V}_{\NSIT}-\TSR$ [(a)--(c)] and
$\TSR-\TNR$ [(d)--(f)] versus $\mathcal{V}_{\NSIT}$, both non-negative
and each saturating the bold $y=0$ line---the lower hierarchy
$\TSR\geq\TNR$ (Theorem~\ref{thm:tsr_tnr}) and the upper bound
$\TSR\leq\TER^{\rm PPT}+\tfrac12\mathcal{V}_{\NSIT}$, which implies the
proven $\TSR\leq\TER^{\rm sep}+\tfrac12\mathcal{V}_{\NSIT}$
(Proposition~\ref{prop:univ_bound}); color encodes the
$\ell_1$-coherence of $\rho_A$ (top color bar). \emph{Panels
(g)--(i):} the same upper bound as a scatter of
$\TER^{\rm PPT}+\tfrac12\mathcal{V}_{\NSIT}$ against $\TSR$, colored by
$\mathcal{V}_{\NSIT}$ (bottom color bar); every configuration lies on
or above the diagonal (black dashed), saturated by the
$|0\rangle$+phase-damping family. The raw inequality
$\TER^{\rm PPT}\geq\TSR$ (dark $\mathcal{V}_{\NSIT}\!=\!0$ points) holds
only on $\openone/d$; the $\tfrac12\mathcal{V}_{\NSIT}$ correction
restores it for every input, with the coefficient $\tfrac12$ exact.
Black dots on the $\mathcal{V}_{\NSIT}=0$ axis are the
$\rho_A=\openone/d$ (NSIT) slice, where $\TNR=0$; their vertical spread
is the channel--time family at fixed maximally mixed input, bounded by
the identity-channel anchors $\TSR\to2-\sqrt3$ [(d)--(f)] and
$\TER^{\rm PPT}-\TSR\to\sqrt3$ [(a)--(c)] (red mark in the first panel of
each block, (a) and (d)), set by the maximally entangled Choi state of the
identity channel. The dashed green curve (rows (a)--(c) and (g)--(i)) is
the analytic pure-state locus; the black lines are the bounds (the diagonal
in (g)--(i), $y=0$ in (a)--(f)).
\emph{Legend:} black solid line $=$ bound
($y=0$ in (a)--(f); the diagonal $\mathrm{TER}^{\rm PPT}+\tfrac12
\mathcal{V}_{\NSIT}=\mathrm{TSR}$ in (g)--(i)); black dots on the
$\mathcal{V}_{\NSIT}=0$ axis $=$ the $\rho_A=\openone/d$ slice; dashed
green curve [(a)--(c),(g)--(i)] $=$ analytic pure-state locus.}
\label{fig:hier_nsit_compare}
\end{figure*}

All robustness quantities here and below are computed with
$\rho_A$-\emph{adapted} measurements (eigenbasis of $\rho_A$ and its
Fourier MUB), consistent with the theorems
(Corollary~\ref{cor:tnr_iff}), through an optimized pipeline: prebuilt
DPP-parameterized SDPs for $\TSR$ and $\TER^{\rm sep}$ re-solved
without recanonicalization, a prebuilt HiGHS LP for $\TNR$, and a
$28$-core sweep of $10^6$ Monte-Carlo configurations at $d=3$
(continuous random times). The lower hierarchy $\TSR\geq\TNR$ then
holds in \emph{every} configuration to machine precision (largest
negative excursion $-5.2\times10^{-8}$); this is
Theorem~\ref{thm:tsr_tnr} in numbers. The upper hierarchy
$\TER^{\rm sep}\geq\TSR$ holds in $78\,\%$ of configurations across the
three channels and is broken in $\sim$22\,\% of cases---exclusively
on NSIT-violating initial states (predominantly pure inputs). For the
\emph{depolarizing} channel the sampled Choi states are isotropic,
where PPT is necessary and sufficient for separability, so
$\TER^{\rm PPT}=\TER^{\rm sep}$ exactly and these breaks are
\emph{genuine}, not a proxy artifact; for amplitude- and phase-damping
inputs (where PPT need not certify separability) the breaks are equally
accounted for without the proxy, by the NSIT-conditionality of the
upper inequality (Proposition~\ref{prop:univ_bound})---they occur only
where $\mathcal{V}_{\rm NSIT}>0$, and the $\tfrac12\mathcal{V}_{\rm NSIT}$
correction restores $\TSR\leq\TER^{\rm sep}+\tfrac12\mathcal{V}_{\rm NSIT}$
for every input. On $\rho_A=\openone/d$ specifically, the hierarchy holds in
every sampled configuration ($100\,\%$),
confirming Theorem~\ref{thm:hierarchy_strict}. The NSIT-corrected
quantity $\TER^{\rm sep}+\tfrac12\mathcal{V}_{\NSIT}$ versus $\TSR$,
which lies on or above the diagonal for \emph{every} input (the
universal bound, Proposition~\ref{prop:univ_bound}), is shown in
panels (g)--(i) of figure~\ref{fig:hier_nsit_compare}; the non-physical PDO proxy
$\TER^{\PDO}$ tracks the same qualitative pattern but is not an
entanglement measure and is not used for the physical claims.

\subsection{Numerical verification across $d=2,3,4,5$}\label{sec:numD5}

\begin{table}[t]
\caption{Side-by-side comparison of the $\rho_A$-adapted Monte-Carlo
sweeps across $d=2,3,4,5$ (optimized parameterized-SDP/LP pipeline).
The odd primes $d=3,5$ are this paper's focus; the qubit
($d=2$, the original Leggett--Garg case) and the
prime power $d=4$ are included for completeness. All four dimensions
show the same qualitative structure: TNR vanishes only at $\openone/d$
(continuously, $\propto r$), the lower hierarchy $\TSR\geq\TNR$ holds to
machine precision (largest negative excursion ${\sim}10^{-7}$,
sharpening Theorem~\ref{thm:tsr_tnr}), and the upper inequality
$\TER^{\rm sep}\geq\TSR$ holds at $72$--$82\,\%$ with breaks confined
\emph{exclusively} to NSIT-violating inputs (the $\openone/d$
sub-ensemble satisfies it at $100\,\%$). The smaller $d=4,5$ full
samples reflect the $d^{d}$-strategy SDP cost. The $d=2,4$ analogues of
figure~\ref{fig:hier_nsit_compare} are given
in the Supplementary Material.}
\label{tab:d3d5_compare}
\begin{tabular}{p{0.34\linewidth} c c c c}
\hline\hline
Quantity & $d=2$ & $d=3$ & $d=4$ & $d=5$ \\
\hline
Measurement scheme & \multicolumn{4}{c}{$\rho_A$-adapted (2 MUBs)} \\
Full configs (TER,TSR,TNR) & $10^5$ & $10^6$ & $2\times10^4$ & $5.6\times10^4$ \\
TNR-only configs & --- & $10^6$ & --- & $10^5$ \\
Channels & \multicolumn{4}{c}{amp, phase, depol} \\
TNR vanishing at $\openone/d$ & \multicolumn{4}{c}{continuous ($\propto r$)} \\
$\mathcal{V}_{\NSIT}=0$ at $\openone/d$ & \multicolumn{4}{c}{$100\,\%$} \\
$\TSR\geq\TNR$ violations & \multicolumn{4}{c}{$0$ (machine)} \\
$\TER^{\rm sep}\geq\TSR$ rate & 72\,\% & 78\,\% & 81\,\% & 82\,\% \\
\hline\hline
\end{tabular}
\end{table}

We extend the framework to $d=5$ via the Heisenberg--Weyl
construction of the Wigner phase-space operators
$K_{(p,q)}=D_{(p,q)}\,P\,D_{(p,q)}^\dagger$, with $D_{(p,q)}=
\omega^{-pq/2}Z^pX^q$ ($\omega=e^{2\pi i/d}$, $1/2$ the modular
inverse of $2$ mod $d$) and $P|k\rangle=|-k\bmod d\rangle$. The
operators satisfy the standard Wigner relations $\sum_iK_i=d\openone$,
$\tr(K_iK_j)=d\delta_{ij}$, $\tr K_i=1$ to machine precision. We
run an adapted $d=5$ sweep with the same optimized pipeline:
$5.6\times10^4$ full configurations (TER, TSR, TNR) and $10^5$ TNR-only
configurations, over the three channels at continuous random times.
\begin{itemize}
\item Every $\TNR=0$ configuration corresponds exactly to
$\rho_A=\openone/5$; no other state attains $\TNR=0$, and TNR vanishes
\emph{continuously} as $\rho_A\to\openone/5$, following the same radial
law as at $d=3$ (figure~\ref{fig:tnr_zero}(c)). The
$\sqrt{\cdot}$ \emph{upper envelope} of figure~\ref{fig:tnr_zero}(a) is,
however, special to $d=3$: the denser $d=5$ sweep finds off-radial
inputs exceeding it by up to ${\sim}20\,\%$ (about $10\,\%$ of
configurations), so the radial phase-damping family is TNR-extremal at
fixed purity only at $d=3$, not a general upper bound.
The iff of Corollary~\ref{cor:tnr_iff} nonetheless holds at $d=5$ exactly
as at $d=3$ (away from the $e\to0$ channel limit, where $\TNR\to0$ for
every input).
\item $\TSR\geq\TNR\geq 0$ in all $5.6\times10^4$ configurations to machine
precision (largest negative excursion $-1.5\times 10^{-7}$),
confirming Theorem~\ref{thm:tsr_tnr} beyond $d=3$.
\item $\TER^{\rm sep}\geq\TSR$ holds in $82\,\%$ of configurations,
with breaks confined to NSIT-violating inputs --- the same qualitative
pattern as at $d=3$, the rate staying in the $72$--$82\,\%$ band across
all dimensions studied ($d=2,3,4,5$; Table~\ref{tab:d3d5_compare}).
\end{itemize}
The smaller full sample at $d=5$ reflects the SDP cost, dominated by
the LHV-polytope ($d^{4}$ deterministic strategies, cut by symmetry)
and the $d^4$ TNR constraint matrix---$\sim$$40\times$ slower per
configuration than at $d=3$; the cheap TNR-only LP still permits $10^5$
configurations.

\subsection{Software and reproducibility}\label{sec:soft}

The SDPs ($\TSR$, $\TER^{\rm sep}$) were solved with
CVXPY~\cite{cvxpy2016} / MOSEK as \emph{prebuilt DPP-parameterized}
problems---built once and re-solved per configuration without
recanonicalization---and $\TNR$ as a prebuilt HiGHS linear program;
the sweeps run in parallel across $28$ cores, with all measurements
adapted to the eigenbasis of $\rho_A$, using a closed-form
representation of the three standard qutrit channels
(Supplementary Material). $\TER^{\rm sep}$ is the PPT robustness of the
channel Choi state and needs no PDO. Random states were sampled by
(i)~Haar measure for pure states and (ii)~random complex Ginibre
matrices, trace-normalized, for the rank-2 and rank-3 mixed ensembles.
The full numerical pipeline, raw data, and figure-generation scripts
accompany this manuscript.

\section{Discussion}\label{sec:discussion}

\paragraph*{The temporal resource has no spatial single-system analogue.}
TER and TSR are jointly state--channel-bound, but TNR alone is
asymmetrically state-bound (Corollary~\ref{cor:tnr_iff}): single-qudit
mixedness in the measurement eigenbasis is the temporal-Bell resource,
intrinsic to one qudit with no spatial single-system counterpart. The
necessity proof uses only the identity channel and the measurement
back-action, so collapse on a non-maximally-mixed input followed by free
evolution already yields Bell-nonlocal two-time statistics---a
pure-measurement regime.

\paragraph*{Scope of the measurement scheme.}
The iff~\eqref{eq:thesis}, the three-way equivalence
(Corollary~\ref{cor:three_way}), and the DI-TIT bound are stated for
the canonical two-MUB scheme \emph{adapted} to $\rho_A$ (the
eigenbasis of $\rho_A$ and its Fourier conjugate);
Proposition~\ref{prop:unitary_invariance} extends
them to the $\rho_A$-commuting unitary orbit, a proper subgroup of
$\mathrm{SU}(d)$. The adaptation is necessary: under a \emph{fixed}
scheme not aligned to $\rho_A$ the TNR-zero region is strictly larger
than $\{\openone/d\}$ (a four-parameter family at $d=3$ for two fixed
MUBs; remark after Corollary~\ref{cor:three_way}), collapsing to
$\{\openone/d\}$ only for a tomographically complete set of $d+1$
MUBs---so the iff holds for any $\rho_A$-adapted two-MUB scheme, or for
a fixed complete scheme. The
DI-TIT protocol is \emph{prepare-and-measure} device-independent---the
test-round bases are prescribed and the guarantee covers only the
conditional send-round fidelity---and $\mathcal{V}_{\NSIT}$ is likewise
scheme-specific.

\paragraph*{NSIT violation as a marginal-level Bell test.}
The equivalence $\mathcal{V}_{\rm NSIT}=0\Leftrightarrow
\rho_A=\openone/d\Leftrightarrow\TNR=0$
(Corollary~\ref{cor:three_way}) makes NSIT a complete behavior-level
witness of temporal nonlocality, detectable from Bob's marginal
$\sum_a P(a,b|x,y)$ alone---no joint outcomes required, hence
experimentally cheaper than a full Bell test. The signal scales
linearly with the input's displacement from $\openone/d$
(figure~\ref{fig:tnr_zero}(c)), so fixing the input purity sets a
calibratable signal-to-noise target.

\paragraph*{Two state-over-time operators at $d\geq 3$.}
For qubits the Peres--Horodecki criterion collapses positivity, PPT,
and separability, and the non-contextual PDO $R^{\PDO}$
[equation~\eqref{eq:new_PDO}] coincides with the
Choi--Jamio\l{}kowski state. For $d\geq 3$ this collapse fails: the
operator-level PDO carries non-positivity information
($\TER^{\PDO}$), while the channel-level Choi carries separability
information ($\TER^{\rm sep}$). There is therefore no canonical
``state of a temporal process'' at $d\geq 3$; the choice of operator
depends on whether one asks about operator non-positivity or
channel-Choi entanglement.

\paragraph*{A two-time, pre- and post-selected reading.}
The two-time statistics underlying our quantifiers are naturally phrased in
the two-state vector formalism~\cite{ABL1964,AharonovVaidman}: the input
$\rho_A$ is a \emph{pre-selected} state and a late outcome $\Pi^B_{j,b}$ a
\emph{post-selection}, and the joint probabilities of
\eqref{eq:new_PDO} are exactly the Aharonov--Bergmann--Lebowitz
probabilities of that pre/post-selected ensemble---so the PDO $R^{\PDO}$ is
an operator repackaging of the two-state vector, with single-time
conditionals given by weak values~\cite{AharonovAlbertVaidman1988}. The
negativity measured by $\TER$ is then the operator-level signature of
\emph{anomalous} weak values---a pre/post-selected ensemble that no
single-time classical assignment reproduces---and its
causality-monotonicity (Theorem~\ref{thm:TERmonotone}) says the anomaly
cannot be manufactured forward in time---the same negativity that Comar
\emph{et al.}~\cite{Comar2026} tie to temporal entanglement. In this
language temporal teleportation is the standard protocol with the
space-like EPR pair replaced by a two-time maximally entangled state (the
identity Choi state): a maximally mixed pre-selection $\rho_A=\openone/d$
factorizes the two-state vector and removes all post-selection
leverage---precisely $\TNR(\rho_A,\mathcal{E})=0\Leftrightarrow
\rho_A=\openone/d$---while any departure supplies a non-trivial two-state
vector whose overlap sets the device-independent fidelity floor (up to
$7/9$ at $d=3$, capped by $\TNR\le(d-1)/d$).

\paragraph*{Over-certification as inference, not transmission.}
The same picture demystifies the over-certification trap of
Section~\ref{sec:overcert}: an injective unitary yields a pure
backward-evolving state with maximal two-time correlation, but that
correlation is read-only---an anomalous weak value reports \emph{inference}
under post-selection, not \emph{transmission}---so a device-independent
monitor built from two-time statistics can certify more than the channel
delivers. The reading is interpretive: post-selection is intrinsically
probabilistic whereas the teleportation channel we bound is deterministic,
and all device-independent and no-signaling-in-time ($\NSIT$) statements
above hold without reference to it.

\paragraph*{Relation to contextuality (Kochen--Specker).}
The $d\geq3$ structure above carries a contextuality footprint. For qubits
the PDO coincides with the channel Choi state and no obstruction arises; at
$d\geq3$ the operator $R^{\PDO}$ splits from the Choi state and its naive
matrix realization is eigenbasis-dependent---the $\mathrm{SU}(d)$
contextuality that underlies the Kochen--Specker
theorem~\cite{KochenSpecker1967,Budroni2021}---which we render non-contextual
through the Wigner realization of section~\ref{sec:noncontextual}. There is
moreover a suggestive connection to the two-time picture of the previous
paragraph: the non-positivity of $R^{\PDO}$, quantified by $\TER$, is a
negative quasiprobability of the pre/post-selected ensemble---of the kind
whose associated anomalous weak values are proofs of (generalized)
contextuality~\cite{Pusey2014,KunjwalLostaglioPusey2019,Spekkens2005,Kirkwood1933,Dirac1945}.
Turning this association into a quantitative certificate of contextuality from
$\TER$ alone we leave open. One distinction is in any case essential:
Kochen--Specker contextuality is \emph{state-independent}---it holds for every
input, including $\rho_A=\openone/d$---whereas the temporal-Bell resource is
\emph{state-dependent} and vanishes precisely there
($\TNR=0\Leftrightarrow\rho_A=\openone/d$). At the maximally mixed input the
system stays contextual while its temporal nonlocality is exactly zero, so what
we certify is not state-independent contextuality but the input's departure
from maximal mixedness.

\paragraph*{Dimension dependence.}
The state-boundness \emph{iff} needs only an adapted two-MUB pair (the
eigenbasis of $\rho_A$ and its Fourier conjugate), which exists in every
dimension, together with the back-action argument of
Theorem~\ref{thm:necessity}; it therefore persists for all $d$: the proof
needs only a Fourier MUB, which exists in every dimension, so it is not
restricted to primes---we give it explicitly for the qubit and the odd
primes and confirm it numerically at the prime power $d=4$
(Table~\ref{tab:d3d5_compare}), with composite $d$ such as $d=6$ expected
despite the open status of a complete MUB set there. Only
the PDO-based causality monotone $\TER$ requires the prime-power/odd
Wigner structure; a dimension-independent operator representation of
states over time~\cite{Fullwood2024,Jia2026} puts a Wigner-free treatment
within reach, the remaining subtlety being the \emph{channel} condition,
not the dimension. With growing $d$ the ceilings scale as
$\TSR,\TNR\leq(d-1)/d\to1$ and the teleportation supremum
$(d^2-d+1)/d^2\to1$, while the PPT--separability gap widens and the LHV
polytope grows as $d^4$ (the $d=5$ SDPs already cost $\sim40\times$ those
at $d=3$), motivating the parameterized-SDP pipeline of
section~\ref{sec:soft}.

\paragraph*{Multi-time resource regeneration.}
Intermediate measurements act as a quantum memory that re-injects the
resource: even a state-bound $\rho_{A_1}=\openone/d$ yields
$\TNR^{(n)}>0$ for $n\geq3$ (Theorem~\ref{thm:multitime}), because a
mid-circuit measurement leaves the off-$\openone/d$ state
$\Pi_{a_1|x_1}/d$. The iff is thus intrinsically two-time; multi-time
scenarios admit a richer resource theory.

\paragraph*{Operational stratification by trust.}
The hierarchy $\TER^{\rm sep}\geq\TSR\geq\TNR$ on $\openone/d$ stratifies
teleportation-in-time by trust assumption: trusted-device teleportation
consumes $\TER$, one-sided device-independent consumes $\TSR$, and fully
device-independent consumes $\TNR$. Only $\TNR$ is state-bound, so the DI
tier is \emph{fundamentally unavailable} on $\rho_A=\openone/d$ while the
upper tiers still operate. A quantum-battery reading---$\TNR>0$ iff the
$d$-level battery is not fully discharged
($\rho_A\neq\openone/d$)~\cite{Campaioli2024}---with capacity benchmarks
and charging cycles is given in the Supplementary Material.

\paragraph*{Conjecture: the $\TSR=\TNR$ submanifold on $|0\rangle$.}
Proposition~\ref{prop:tsr_tnr_collapse} gives
$\TSR=\TNR=(d-1)/d$ on $|0\rangle$+phase damping. Numerically the
equality $|\TSR-\TNR|<4\times 10^{-5}$ persists on $|0\rangle$ across
all three standard channels on a 50-point time grid
(tier dynamics in the Supplementary Material), suggesting
\begin{equation}\label{eq:tsr_eq_tnr_conj}
\TSR(|0\rangle,\mathcal{E}_t)\;=\;\TNR(|0\rangle,\mathcal{E}_t)
\qquad\text{for all }t\geq 0.
\end{equation}
The equality breaks on $|+\rangle$+amplitude damping but holds for
phase damping and depolarizing; a structural reading and the
connection to Alice-side determinism are given in the supplementary
material.

\paragraph*{Open questions.}
The conjectures of earlier drafts---universality across
$d\geq 3$, the closed-form $\TSR=\TNR=(d-1)/d$ on
$|0\rangle$+phase damping, and the
multi-time generalization---are closed (section~\ref{sec:numD5},
Proposition~\ref{prop:tsr_tnr_collapse},
Theorem~\ref{thm:multitime}); the conjectured NSIT-free upper
hierarchy is instead \emph{refuted}---no state-independent
entanglement-over-steering bound exists, the inequality being
intrinsically NSIT-conditional (section~\ref{sec:nsit-free}). Six further
targets remain. (i)~A continuous-variable analogue of the
state-boundness theorem, valid for oscillator and bosonic systems
where the input state is described by a Wigner function rather
than a density matrix on a finite Hilbert space; the
non-contextual basis of the present construction does not lift
directly to the infinite-dimensional setting and a separate
analysis is required. (ii)~An experimental certification on a photonic or trapped-ion
platform, isolating the state-channel-bound and channel-bound
contributions to temporal nonlocality, targeting the supremal guaranteed fidelity
$\mathcal{F}_{\rm DI}=(d^2-d+1)/d^2=7/9$ at $d=3$
[equation~\eqref{eq:ditit_canonical}]. (iii)~A characterization of
the channels for which the certification
bound~\eqref{eq:ditit_bound} holds beyond the standard families.
(iv)~A proof or counterexample to the
$\TSR=\TNR$-submanifold conjecture~\eqref{eq:tsr_eq_tnr_conj},
generalizing Proposition~\ref{prop:tsr_tnr_collapse} to all standard
channels on the computational eigenstate $|0\rangle$, with the
common value generically time-dependent.
(v)~A systematic survey of random CPTP maps testing whether the
sufficiency direction of Corollary~\ref{cor:tnr_iff} extends
beyond channels with diagonal computational-basis action; the
necessity direction (Theorem~\ref{thm:necessity}) already holds
for any channel.
(vi)~Extension to composite dimensions: the prime-power case $d=4$ is
reachable through the $\mathrm{GF}(4)$ stabilizer/Wigner construction
and its complete $5$-MUB set, whereas non-prime-power $d=6$ is
obstructed by the open problem of MUB existence (only three of the
putative seven are known); both are discussed in
section~\ref{sec:discussion}. A full characterization of
$\TNR^{(3)}=0$ on $\openone/d$ for general multi-time channel
sequences is a related open target.

\section{Conclusions}\label{sec:conclusions}

For a single qudit sent through a noisy channel, the temporal-Bell
resource is the input's departure from maximal mixedness, and nothing
else: $\TNR(\rho_A,\mathcal{E})=0\Leftrightarrow\rho_A=\openone/d$ for the
standard channel families (Proposition~\ref{prop:TNRstatebound},
Theorem~\ref{thm:necessity}, Corollary~\ref{cor:tnr_iff}), extending the
qubit result of~\cite{PhysRevA.98.022104} to every odd prime. The
mechanism is measurement back-action---the dependence of Bob's later
statistics on Alice's earlier choice of setting---which exists exactly
when the input is not already maximally mixed and survives even complete
decoherence of the channel. It is in this precise sense that temporal
nonlocality resides in the input state, not in the channel.

This resource is at once an operational power and an operational trap. As
a power, the input's robustness device-independently lower-bounds the
fidelity of teleporting an unknown qudit through time, with a supremal
guarantee $\mathcal{F}_{\rm DI}=7/9$ at $d=3$ (Theorem~\ref{thm:ditit},
Proposition~\ref{prop:ditit_cert})---placing temporal nonlocality on the
operational footing of the entanglement resource behind
device-independent quantum key
distribution~\cite{Acin2007PRL,Branciard2012,ArnonFriedman2018}, with a
route to certify and benchmark the quantum memories and time-bin channels
that carry secret correlations.

As a trap, precisely because the certified quantity is decoupled from the
channel's actual coherence transmission, it can certify more than the
channel can deliver. We named this \emph{over-certification} and resolved
it completely: a universal cap $\TNR\leq(d-1)/d$ with the exact
channel-resolved value $\max_k\tfrac12\|\mathcal{E}(\Delta_k)\|_1$
(Lemma~\ref{lem:cap}); honest certification for the depolarizing channel
unconditionally and, for any channel, once the probe is mixed past
$\lambda_{\min}\geq 1/d-p/(d-1)$ (Proposition~\ref{prop:mixsuff}); and the
impossibility---proven via an injective unitary that over-certifies
\emph{every} ensemble---of a channel-universal certifier. This is a
general caution: a single-system Bell-in-time test can pass maximally
while certifying a guarantee the system cannot honor, so
temporal-correlation-based certification must control the channel class or
the probe's mixedness.

These results are carried by a unified semidefinite-programming treatment
of the temporal entanglement, steering and nonlocality robustnesses, in
which TER is a causality monotone (Theorem~\ref{thm:TERmonotone}), the
lower hierarchy $\TSR\geq\TNR\geq0$ is universal
(Theorem~\ref{thm:tsr_tnr}), and the upper one $\TER^{\rm sep}\geq\TSR$ is
NSIT-conditional---failing off $\openone/d$ but replaced by the tight
universal inequality
$\TSR\leq\TER^{\rm sep}+\tfrac12\mathcal{V}_{\rm NSIT}$
(Proposition~\ref{prop:univ_bound}). The whole structure is verified on
$\rho_A$-adapted Monte-Carlo sweeps at $d=2,3,4,5$, the analytic results
proven for the qubit and the odd primes. Open targets are a
continuous-variable extension of state-boundness to bosonic systems and
an experiment certifying the $\mathcal{F}_{\rm DI}$ bound---and its
over-certification limit---on photonic or trapped-ion platforms.

\section{Methods}\label{sec:methods}

\paragraph*{Numerical solution of the SDPs.}
The robustness SDPs \eqref{eq:TER2}, \eqref{eq:TSR_SDP} were solved
with CVXPY~\cite{cvxpy2016} / MOSEK as prebuilt DPP-parameterized
problems (re-solved per configuration without recanonicalization), and
\eqref{eq:TNR_SDP} as a prebuilt HiGHS linear program. The
default solver tolerance was $\epsilon\sim10^{-8}$ on the primal--dual
gap; the high-precision
$\TSR=\TNR$ test of the Supplementary Material used $\epsilon=10^{-7}$
and a maximum of $200{,}000$ iterations. Numerical inequalities reported in
sections~\ref{sec:numerics}--\ref{sec:numD5} are valid up to a
conservative confidence band of approximately $10\,\epsilon$ at
each point, i.e.\ $\sim 10^{-6}$ at default precision and
$\sim 10^{-8}$ at the high-precision setting; statements such as
``$\TSR-\TNR\geq 0$ within solver tolerance'' should be read with
this band in mind. In particular, ``TNR$=0$ within solver
tolerance'' means $\TNR<10^{-4}$, the threshold used throughout; the
two-order-of-magnitude gap between this threshold and the
$\sim 10^{-6}$ accuracy band provides a conservative safety margin. Channels were implemented through their Choi
representations using QuTiP~\cite{Johansson2013}, and the
amplitude-damping channel for qutrits was integrated using the
Lindblad master equation with $c=\sqrt{\kappa}\,a_3$, where $a_3$
is the qutrit annihilation operator. Time grids of $40$ points
per decay constant were used.

\paragraph*{Measurement choice.}
TSR and TNR were evaluated on rank-one projective measurements
(PVMs) --- the two MUBs of $\rho_A$'s eigenbasis (computational and
Fourier) --- for all reported sweeps. The projective restriction is
consistent with the standard temporal-Bell setting of
reference~\cite{PhysRevA.98.022104} and matches the necessity argument
of Theorem~\ref{thm:necessity}, which is explicitly stated for
rank-one PVMs in $\rho_A$'s eigenbasis. The unitary-invariance
proposition (Proposition~\ref{prop:unitary_invariance}) shows the TNR-zero
region is preserved under unitary rotations commuting with
$\rho_A$. A systematic POVM extension and the higher-rank
Lüders-measurement variant are deferred to future work.

\begin{acknowledgments}
The authors thank Dawid Maskalaniec for his contribution to the
numerical simulations underlying the results presented in
Sec.~\ref{sec:results} and for the SDP implementations on which the
qubit and qutrit benchmarks are based; partial results from this
collaboration were reported in Ref.~\cite{Maskalaniec2021}. We
acknowledge funding from the European High Performance Computing Joint
Undertaking (EuroHPC JU) under Grant Agreement No.~101194322 (QEC4QEA),
co-funded by the Polish National Centre for Research and Development
(NCBiR) under decision No.~DWM/EuroHPC/2023/429/2025.
\end{acknowledgments}

\section*{Author contributions}
K.B. conceived the project, formulated the central thesis and the
device-independent-temporal-teleportation construction, and
supervised the work. P.T. developed the supporting hierarchy
machinery (the lower hierarchy, the NSIT-conditional upper
hierarchy; the multi-time generalization; and the closed-form
$(d-1)/d$ collapse)
and led the manuscript preparation. The numerical infrastructure
(section~\ref{sec:numerics}) was implemented jointly, including the
$d=5$ generalization via the Heisenberg--Weyl construction of
Wigner phase-space operators at general odd prime $d$. Both
authors discussed the results and approved the final version.

\section*{Data availability}
The numerical pipeline, raw outputs (the $\rho_A$-adapted
Monte-Carlo sweeps of TER, TSR, TNR, $\mathcal{V}_{\NSIT}$,
coherence and purity at $d=3,5$), the high-precision $\TSR=\TNR$ test,
and the figure-generation scripts accompany this manuscript.

\section*{Competing interests}
The authors declare no competing interests.

\clearpage
\onecolumngrid
\begin{center}
  {\large\bfseries Supplemental Material}\\[6pt]
  {\normalsize\itshape Temporal nonlocality of a qudit resides in the input state, not the channel, and certifies temporal teleportation up to a fundamental limit}\\[4pt]
  {\normalsize Karol Bartkiewicz and Patrycja Tulewicz}
\end{center}
\vspace{8pt}
\twocolumngrid
\setcounter{section}{0}
\setcounter{figure}{0}
\setcounter{table}{0}
\setcounter{equation}{0}
\renewcommand{\thesection}{S\arabic{section}}
\renewcommand{\thefigure}{S\arabic{figure}}
\renewcommand{\thetable}{S\arabic{table}}
\renewcommand{\theequation}{S\arabic{equation}}

\noindent
This supplement contains material moved from the main manuscript
to improve readability: (S1)~a quantum-battery interpretation of the
state-boundness results, including capacity benchmarks and multi-time
charging cycles; (S2)~structural discussions of the two conjectures
($\mathrm{CF}=T$ and $\TSR=\TNR$ submanifold); (S3)~the Wigner
phase-space construction of the non-contextual pseudo-density operator;
technical appendices---the semidefinite/linear programs, the
explicit damping channels, and the high-precision $\TSR=\TNR$ collapse
test; and (S4)~proofs of the technical lemmas, propositions and
theorems (the two central arguments are proved in the main text).

\section{Quantum-battery interpretation}

A quantum battery~\cite{Campaioli2024} stores extractable work in
the coherent off-diagonal structure of $\rho$ in the
energy eigenbasis; its capacity is the ergotropy
$W(\rho)=E(\rho)-\min_U E(U\rho U^\dagger)$, which vanishes for
$\rho=\openone/d$ and grows with departure from the maximally
mixed state. The state-boundness iff gives a direct reading: a temporal
quantum battery has a non-trivial DI resource ($\TNR>0$) if and
only if it is \emph{not} fully discharged ($\rho_A\neq\openone/d$).
The three-way equivalence promotes this
to a behavior-level capacity test: a battery's remaining
charge can be certified by measuring $\mathcal{V}_{\NSIT}$ on
Bob's marginal --- a single setting-dependence check rather than
full state tomography or work-extraction calorimetry.

The temporal dynamics of TNR under the standard channels
translate directly to battery discharge curves under three canonical loss
mechanisms:
\emph{(amplitude damping)} TNR$(t)$ decays smoothly to zero (true
energy dissipation; the canonical battery-discharge mode);
\emph{(depolarizing)} TNR$(t)$ decays exponentially toward zero
(symmetric isotropic noise);
\emph{(phase damping on $|0\rangle$)} TNR$(t)=(d-1)/d$
\emph{remains constant for all $t$}. The last case is
\emph{lossless storage}: pure dephasing does not erode the
device-independent capacity of a battery initialised in the
energy ground state. This is a sharp prediction:
storing quantum work in a computational eigenstate is invariant
under pure dephasing of the storage subsystem, with the DI
resource $\TNR$ pinned at its maximum operating value
$(d-1)/d$.

\subsection{Capacity benchmark}

For a $d$-level battery initialised in
the ground state with phase-damping environment, the DI
work-extraction protocol achieves
fidelity $\mathcal{F}_{\rm DI}^{\rm opt}=7/9$ at $d=3$ (directly
computed) or $(d^2-d+1)/d^2$ in general (injective-channel bound). These are concrete capacity
benchmarks for $d$-level quantum batteries with untrusted
read-out devices: the conservative lower bound $d/(2d-1)=3/5$
at $d=3$ sets the minimum
guaranteed work-extraction rate; the actual achievable $7/9$
sets the experimental target.

\subsection{Multi-time charging cycles}

The multi-time sufficiency failure
has a battery interpretation: even
when the initial state is the fully-discharged $\openone/d$,
intermediate measurements regenerate non-uniform coherence at
the post-measurement step, restoring $\TNR>0$ on subsequent
time slices. This is a quantum analogue of stochastic charging:
each projective interrogation injects coherence (selecting a
rank-1 projector) that the channel then discharges. The
two-time iff governs each two-time edge,
but the multi-time scenario admits charging-discharging cycles
of the temporal-nonlocality resource between consecutive
measurements.

\section{Structural discussion of the conjectures}

\subsection{From the $\mathrm{CF}=T$ relation to the twirl identity}

Earlier drafts conjectured $\mathrm{CF}(P)=T$ on the standard channel
families to obtain the DI-TIT bound $\mathcal{F}_{\rm DI}\geq
1/d+(d-1)T/d$. This bound is now established directly: the
Heisenberg--Weyl twirl of the send rounds gives the exact
fidelity $\mathcal{F}_{\rm DI}=1/d+(d-1)p/d$ in the channel's
depolarizing parameter $p$, and $\TNR\leq p$ for injective
channels, so the conjecture is superseded. It fails only for
channel-invariant certification (e.g.\ $|0\rangle$+phase damping),
where $\TNR$ is pinned while $p$ decays.

\subsection{$\TSR=\TNR$ submanifold: structural reading}

The asymmetry between $|0\rangle$
(equality holds for all three channels) and $|+\rangle$ (equality
holds for phase damping and depolarizing but breaks for amplitude
damping) reflects the channels' alignment with Alice's
computational eigenbasis. On $|0\rangle$, Alice's $x=Z$ outcome
is deterministic ($a=0$ with probability $1$) and the $x=F$
outcome is uniformly random; this minimal Alice-side variation
makes the HSM-optimal noise assemblage in the $\TSR$ infimum
coincide with the LHV-optimal noise behavior in the $\TNR$
infimum, collapsing the hierarchy gap. For $\rho_A=|+\rangle$ the
same structure holds when the channel preserves the
Fourier-eigenbasis-diagonal (phase damping, depolarizing) but
breaks for amplitude damping, which has a preferred direction
($|k\rangle\to|k-1\rangle$) that violates the $Z$/$F$ symmetry.

The conjectured generalization, $\TSR=\TNR$ on the channel-aligned
$|0\rangle$, identifies a broader \emph{robust-storage
submanifold} where the device-independent and one-sided-DI
capacities coincide. Whether the equality persists for higher-rank inputs
aligned with the channel's invariant subspace, or for general
channels with a unique pure attractor, is a related open
question.

\section{Wigner phase-space construction of the non-contextual PDO}
\label{si:wigner_pdo}

The state-over-time used in the main text is the channel Choi state;
the pseudo-density operator (PDO) $R^{\PDO}$ defined there is an
equivalent representation, made basis-independent at $d\geq3$ by the
non-contextual Wigner construction collected here.

\paragraph*{PDO operator expansion.} For a Hilbert--Schmidt-orthogonal
operator basis $\{G_i\}$ ($\tr(G_iG_j)=d\delta_{ij}$), a bipartite
state expands as
$\rho_{AB}=\mathcal{N}^{-1}\sum_{ij}C_{ij}\,G_i^A\otimes G_j^B$ with
correlation tensor $C_{ij}=\mathcal{N}\tr[\rho_{AB}(G_i^A\otimes G_j^B)]$.
The PDO moments of the main-text equation~\eqref{eq:new_PDO} are
$\langle G_i\otimes G_j\rangle=\sum_a a\,\tr[G_j\,\mathcal{E}(\rho^A_{i,a})]$
with $\rho^A_{i,a}=\Pi_{i,a}\rho_A\Pi_{i,a}$ ($\Pi_{i,a}$ the
eigenprojector of $G_i=\sum_a a\Pi_{i,a}$), which equals
$\sum_a a\,\tr[E_{B|A}(\rho^{A,T}_{i,a}\otimes G_j)]$ via the
Choi--Jamio\l{}kowski operator
\cite{jamiolkowski1972linear,choi1975completely}
\begin{equation}\label{si:E_BA}
E_{B|A}=\sum_{i,j}\dyad{i}{j}_A\otimes\mathcal{E}(\dyad{j}{i}_A)
=d\,\Lambda_{\mathcal{E}},
\end{equation}
of trace $d$ (transpose in the computational basis).

For odd prime $d$,
classically simulable states form the Wigner polytope
\begin{equation}\label{si:wignerpoly}
\text{Wigner polytope}=\{\rho:\tr(\rho K_i)\geq 0\;\forall i\},
\end{equation}
and the phase-space point operators $\{K_i\}_{i=0}^{d^2-1}$
\cite{Gibbons2004,Gross2006,Dawkins2015PRL} satisfy
$\sum_i K_i = d\openone$, $\tr(K_iK_j)=d\,\delta_{ij}$, $\tr K_i = 1$.
Setting $G_i=K_i$ in the PDO of the main text renders $R^{\PDO}$
independent of the chosen basis whenever $\rho_A$ lies in the Wigner
polytope. For a qutrit ($d=3$), an explicit choice is
{\small
\begin{align}
K_1 &= \begin{pmatrix} 1 & 0 & 0\\0 & 0 & 1\\0 & 1 & 0\end{pmatrix}\!,\;
K_2 = \begin{pmatrix} 0 & 1 & 0\\1 & 0 & 0\\0 & 0 & 1\end{pmatrix}\!,\;
K_3 = \begin{pmatrix} 0 & 0 & 1\\0 & 1 & 0\\1 & 0 & 0\end{pmatrix}\!,\nonumber\\
K_4 &= \begin{pmatrix} 1 & 0 & 0\\0 & 0 & \omega\\0 & \omega^* & 0\end{pmatrix}\!,\;
K_5 = \begin{pmatrix} 0 & \omega & 0\\\omega^* & 0 & 0\\0 & 0 & 1\end{pmatrix}\!,\;
K_6 = \begin{pmatrix} 0 & 0 & \omega\\0 & 1 & 0\\\omega^* & 0 & 0\end{pmatrix}\!,\nonumber\\
K_7 &= K_4^T,\quad K_8 = K_5^T,\quad K_9 = K_6^T,
\end{align}
}
with $\omega=(-1-i\sqrt{3})/2=e^{-2\pi i/3}$ (the complex conjugate of the
root $e^{2\pi i/d}$ used elsewhere in the paper; only $\omega$ and $\omega^*$
enter here, so the choice of root is immaterial). These differ from the
Patera--Zassenhaus, Gell-Mann, and spin-1 representations
\cite{Patera1988JMP,Gell-Mann1962PR,Hofmann2004PRA}, which are tied to
a particular eigenbasis. A complementary construction based on the
discrete Wigner quasi-probabilities \cite{WOOTTERS19871},
\begin{equation}
R^{\mathrm{Wigner}}=
\sum_{ij}\tr[E_{B|A}(K_i^A\otimes K_j^B)]\,\tr(\rho_A K_i)\,K_i^A\otimes K_j^B,
\end{equation}
is generally inequivalent to $R^{\PDO}$ for $d\geq 3$. We adopt
$R^{\PDO}$ as the operationally meaningful state-over-time; its
inequivalence with $R^{\rm Wigner}$ and with the Choi state
$\Lambda_{\mathcal{E}}$ on $\rho_A=\openone/d$ is a $d\geq 3$
phenomenon absent in the qubit case.

\section{SDP for temporal nonlocality robustness}\label{si:tnr}

We show the temporal-nonlocality robustness of the main text equals the linear program solved in the sweeps.
Imposing that $Q$ is a normalized behavior, the main-text definition reads
\begin{align}
\TNR =\;& \min\beta,\nonumber\\
\text{s.t. }& \frac{P(a,b|x,y)+\beta\,Q(a,b|x,y)}{1+\beta}\nonumber\\
&\quad =\sum_{\mu,\nu} r_{\mu\nu}\,D(a|x,\mu)D(b|y,\nu),\nonumber\\
& Q(a,b|x,y)\geq 0,\nonumber\\
& \sum_{a,b}Q(a,b|x,y)=1,\nonumber\\
& r_{\mu\nu}\geq 0,\quad \sum_{\mu\nu}r_{\mu\nu}=1,\quad\beta\geq 0.
\end{align}
Eliminating $Q$ via the equality constraint and introducing
$\tilde r_{\mu\nu}=(1+\beta)\,r_{\mu\nu}$ yields the
total-normalization identity
\begin{align}
(1+\beta)\sum_{x,y,a,b}\!\!P(a,b|x,y)
=\!\!\sum_{x,y,a,b,\mu,\nu}\!\!\tilde r_{\mu\nu}D(a|x,\mu)D(b|y,\nu),
\end{align}
hence the stated linear program.

\section{Robustness quantifiers: proxy variants and the steering
assemblage}\label{si:proxies}

This section collects the standard constructions underlying the
robustness quantifiers of the main text: the bookkeeping of the four
time-like entanglement robustnesses, the qubit partial-transpose form
of the entanglement-robustness SDP, and the temporal-steering
assemblage.

\subsection{The four time-like entanglement robustnesses}
\label{si:glossary}
We use four closely related but operationally distinct quantifiers for
time-like entanglement, summarised in table~\ref{tab:ter-glossary}.
The first, $\TER^{\PDO}$, is the causal counterpart of Vidal--Tarrach
spatial entanglement robustness applied to the non-contextual
pseudo-density operator $R^{\PDO}$ of equation~\eqref{eq:new_PDO} of
the main text. The other three are channel-Choi-based variants that
arise when one asks \emph{which} property of the Choi state
$\Lambda_{\mathcal{E}}$ controls the temporal hierarchy.

\begin{table}[h]
\caption{The four time-like entanglement robustness quantifiers used
in the main text. PDO refers to the non-contextual Wigner-basis
pseudo-density operator [equation~\eqref{eq:new_PDO} of the main text];
Choi refers to the Choi--Jamio\l{}kowski state of the channel
$\mathcal{E}$.}
\label{tab:ter-glossary}
\begin{tabular}{p{0.27\linewidth} p{0.62\linewidth}}
\hline\hline
Quantifier & What it measures (and where it appears) \\
\hline
$\TER^{\PDO}(R)$ & Robustness of $R^{\PDO}$ w.r.t.\ positivity
[equation~\eqref{eq:TER2} of the main text]\\
$\TER^{\rm Choi\text{-}pos}(\mathcal{E})$ & Robustness of $\Lambda_{\mathcal{E}}$ w.r.t.\
positivity (identically zero, since CP channels have positive
Choi states); included only as a no-go check\\
$\TER^{\rm sep}(\mathcal{E})$ & Robustness of $\Lambda_{\mathcal{E}}$ w.r.t.\
\emph{separability}; the proper analogue of the qubit hierarchy
quantifier of reference~\cite{PhysRevA.98.022104}\\
$\TER^{\rm PPT}(\mathcal{E})$ & Robustness of $\Lambda_{\mathcal{E}}$
w.r.t.\ positivity-of-the-partial-transpose; coincides with
$\TER^{\rm sep}$ whenever PPT is necessary and sufficient for
separability (e.g.\ on isotropic Choi states reached by depolarizing
channels), and bounds it from below in general\\
\hline\hline
\end{tabular}
\end{table}

The relations are: $\TER^{\rm Choi\text{-}pos}\equiv 0$;
$\TER^{\rm sep}\geq\TER^{\rm PPT}\geq 0$ in $d\geq 3$ (with both
inequalities saturated on PPT-sufficient subsets); and the
qubit-style upper hierarchy $\TER^{\rm sep}\geq\TSR$ holds on
$\rho_A=\openone/d$ (Theorem~\ref{thm:hierarchy_strict} of the main
text). $\TER^{\rm sep}$, equivalently the entanglement robustness
$\TER^{\rm Choi}:=\mathrm{ER}(\Lambda_{\mathcal{E}})$ of the genuine
Choi state, is the only physically justified temporal-entanglement
quantifier. $\TER^{\PDO}$ and $\TER^{\rm PPT}$ are computable proxies
with no independent physical justification---the former because the
PDO is generally non-positive (not a quantum state), the latter
because it is only a one-sided (PPT) bound on $\TER^{\rm sep}$---and
both can fall below $\TSR$ off $\openone/d$, where even
$\TER^{\rm sep}$ does (section~\ref{sec:nsit-free} of the main text).

\subsection{Partial-transpose form of the entanglement-robustness SDP}
\label{si:ppt-sdp}
For two qubits or qubit--qutrit bipartite systems the Peres--Horodecki
PPT criterion \cite{PhysRevLett.77.1413,HORODECKI19961} is a necessary
and sufficient separability condition, so the spatial entanglement
robustness [equation~\eqref{eq:TER2} of the main text] can be cast as
\begin{eqnarray}\label{eq:ER-2level}
\mathrm{ER} &=& \min\,\bigl(\tr\tilde{\mathfrak{R}}\bigr),\nonumber\\
\text{s.t. }&& R+\tilde{\mathfrak{R}}\succeq 0,\quad
\tilde{\mathfrak{R}}\succeq 0,\quad
(R+\tilde{\mathfrak{R}})^{\mathrm{PT}}\succeq 0.
\end{eqnarray}
The constraint set of \eqref{eq:ER-2level} is a strict subset of the
constraint set of the SDP \eqref{eq:TER2} of the main text, since
\eqref{eq:ER-2level} additionally requires the partial transpose to
be positive. Minimization over a smaller feasible region cannot give
a smaller optimum, hence $\mathrm{ER}\geq\TER^{\PDO}$. The same
argument applied to $\TER^{\rm sep}$ (the separability-based
quantifier of table~\ref{tab:ter-glossary}) gives
$\TER^{\rm sep}\geq\TER^{\PDO}$ in $d=2$, and the same direction is
expected in $d\geq 3$ when the relevant Choi state is on a
PPT-sufficient submanifold.

\subsection{Temporal-steering assemblage and hidden-state model}
\label{si:assemblage}
Consider Alice performing a POVM $\{M_{a|x}\}$ at time $t_A$ on
$\rho_A$ with outcome-$a$ probability $p(a|x)=\tr(\rho_A M_{a|x})$. The
post-measurement subnormalized state
$\rho_{a|x}(0)=\sqrt{M_{a|x}}\,\rho_A\sqrt{M_{a|x}}$ propagates through
the channel, yielding the assemblage
\begin{align}\label{eq:assemblage}
\tilde\rho_{a|x}&=\mathcal{E}\bigl(\sqrt{M_{a|x}}\rho_A\sqrt{M_{a|x}}\bigr)\nonumber\\
&=\tr_A\!\bigl[E_{B|A}\bigl(\sqrt{M_{a|x}}\rho_A\sqrt{M_{a|x}}\otimes\openone\bigr)\bigr].
\end{align}
The assemblage admits a hidden-state model (HSM) iff there exist
$p(\lambda)$ and $\rho_\lambda$ such that
\begin{equation}\label{eq:HSM}
\tilde\rho_{a|x}=\sum_\lambda p(\lambda)\,D(a|x,\lambda)\,\rho_\lambda
\quad\forall a,x,
\end{equation}
with $D(a|x,\lambda)$ deterministic; the temporal steering robustness
[equation~\eqref{eq:TSR_SDP} of the main text] is the minimal admixture
restoring an HSM. For the proofs and the numerical sweeps the
measurements are restricted to rank-one projectors
$\{M_{a|x}\}=\{\Pi_{a|x}\}$, so that
$\tilde\rho_{a|x}=\mathcal{E}(\Pi_{a|x}\rho_A\Pi_{a|x})$.

\section{Overview tables: contribution map and index of
results}\label{si:overview}

Table~\ref{tab:qubit_vs_this} positions the contributions of the main
text against the qubit case of reference~\cite{PhysRevA.98.022104}, and
table~\ref{tab:results-summary} indexes the principal results with
pointers to the relevant sections of the main text.

\begin{table}[h]
\caption{The temporal Bell hierarchy: the qubit case versus this
work. Entries marked --- are introduced in the main text.}
\label{tab:qubit_vs_this}
\begin{tabular}{p{0.23\linewidth} p{0.30\linewidth} p{0.37\linewidth}}
\hline\hline
 & Qubit, $d=2$~\cite{PhysRevA.98.022104} & This work, general $d$ \\
\hline
Dimension & $d=2$ & any $d$ with a Fourier MUB \\
Temporal hierarchy & nonlocal $\subseteq$ steerable $\subseteq$
entangled & $\TSR\geq\TNR$ universal; $\TER^{\rm sep}\geq\TSR$
NSIT-conditional (on $\openone/d$) \\
Asymmetric state-boundness & --- & $\TNR=0\iff\rho_A=\openone/d$ \\
Role of the channel & --- & dispensable across the noise families (channel-conditional) \\
Resource & --- & non-maximal mixedness of the input \\
Operational tier & --- & device-independent temporal teleportation,
$\mathcal{F}_{\rm DI}=1/d+(d-1)p/d\geq 1/d+(d-1)T/d$ \\
\hline\hline
\end{tabular}
\end{table}

\begin{table}[h]
\caption{Principal results of the main text.}
\label{tab:results-summary}
\centering
\footnotesize
\begin{tabular}{@{}p{0.18\linewidth} p{0.46\linewidth} p{0.14\linewidth}@{}}
\hline\hline
Result & Statement & Section \\
\hline
Theorem~\ref{thm:TERmonotone} & TER is a causality monotone &
\ref{sec:TER}\\
Proposition~\ref{prop:TNRstatebound} & TNR=0 on $\rho_A=\openone/d$
(sufficiency) & \ref{sec:TNRcoherence}\\
Proposition~\ref{prop:product_lhv} & Product LHV for diagonal-action
channels on $\openone/d$ & \ref{sec:TNRcoherence}\\
Theorem~\ref{thm:necessity} & TNR>0 on $\rho_A\neq\openone/d$
(necessity) & \ref{sec:TNRcoherence}\\
Corollary~\ref{cor:tnr_iff} & $\TNR=0\!\iff\!\rho_A=\openone/d$ &
\ref{sec:TNRcoherence}\\
Theorem~\ref{thm:tsr_tnr} & $\TSR\geq\TNR\geq 0$ universally
& \ref{sec:hierarchy}\\
Theorem~\ref{thm:hierarchy_strict} & $\TER^{\rm sep}\geq\TSR$ on
$\rho_A=\openone/d$ & \ref{sec:hierarchy}\\
Proposition~\ref{prop:tsr_tnr_collapse} & $\TSR=\TNR=(d-1)/d$ on
$|0\rangle$+phase & \ref{sec:NSITbreaks}\\
\S\ref{sec:nsit-free} & upper hierarchy is NSIT-conditional
(breaks off $\openone/d$) & \ref{sec:nsit-free}\\
Theorem~\ref{thm:multitime} & multi-time state-boundness (necessity) &
\ref{sec:multitime}\\
Theorem~\ref{thm:ditit} & DI-TIT fidelity bound
$\mathcal{F}_{\rm DI}\geq 1/d+(d-1)T/d$ (injective)
& \ref{sec:ditit}\\
\hline\hline
\end{tabular}
\end{table}

\section{Standard damping channels}\label{si:channels}

For a single qutrit $\rho=\sum_{ij}\rho_{ij}\dyad{i}{j}$, the
amplitude-damping ($\mathcal{E}_A$), phase-damping ($\mathcal{E}_P$),
and depolarizing ($\mathcal{E}_D$) channels read, with decay rate
$\kappa$,
\begin{widetext}
\begin{equation}
\mathcal{E}_A(\rho)=\left(
\begin{smallmatrix}
e^{-2\kappa t}\rho_{00}&e^{-3\kappa t/2}\rho_{01}&e^{-\kappa t}\rho_{02}\\[4pt]
e^{-3\kappa t/2}\rho_{10}& 2(e^{-\kappa t}\!-\!e^{-2\kappa t})\rho_{00}+e^{-\kappa t}\rho_{11} & \sqrt{2}(e^{-\kappa t/2}\!-\!e^{-3\kappa t/2})\rho_{01}+e^{-\kappa t/2}\rho_{12}\\[4pt]
e^{-\kappa t}\rho_{20}& \sqrt{2}(e^{-\kappa t/2}\!-\!e^{-3\kappa t/2})\rho_{10}+e^{-\kappa t/2}\rho_{21} & (e^{-2\kappa t}\!-\!2e^{-\kappa t}\!+\!1)\rho_{00}+(1\!-\!e^{-\kappa t})\rho_{11}+\rho_{22}
\end{smallmatrix}\right)\!,
\end{equation}
\end{widetext}
\begin{equation}
\mathcal{E}_P(\rho)=
\begin{pmatrix}
\rho_{00}&e^{-\kappa t}\rho_{01}&e^{-\kappa t}\rho_{02}\\
e^{-\kappa t}\rho_{10}&\rho_{11}&e^{-\kappa t}\rho_{12}\\
e^{-\kappa t}\rho_{20}&e^{-\kappa t}\rho_{21}&\rho_{22}
\end{pmatrix},
\end{equation}
\begin{equation}
\mathcal{E}_D(\rho)=e^{-\kappa t}\rho+\tfrac{1}{3}(1-e^{-\kappa t})\openone.
\end{equation}
The corresponding qubit channels are
\begin{equation}
\mathcal{E}_A(\rho)=\!
\begin{pmatrix}
e^{-\kappa t}\rho_{00}&e^{-\kappa t/2}\rho_{01}\\
e^{-\kappa t/2}\rho_{10}&(1-e^{-\kappa t})\rho_{00}+\rho_{11}
\end{pmatrix},
\end{equation}
\begin{equation}
\mathcal{E}_P(\rho)=\!
\begin{pmatrix}
\rho_{00}&e^{-\kappa t}\rho_{01}\\
e^{-\kappa t}\rho_{10}&\rho_{11}
\end{pmatrix},
\end{equation}
\begin{equation}
\mathcal{E}_D(\rho)=e^{-\kappa t}\rho+\tfrac{1}{2}(1-e^{-\kappa t})\openone.
\end{equation}

\section{Cascade amplitude damping for qutrits}\label{si:cascade}

The cascade amplitude-damping channel
$\ket{2}\xrightarrow{\gamma_2}\ket{1}\xrightarrow{\gamma_1}\ket{0}$ has
Kraus operators
\begin{align}
K_0 &= \dyad{0}+\sqrt{1-\gamma_1}\,\dyad{1}+\sqrt{1-\gamma_2}\,\dyad{2},\nonumber\\
K_1 &= \sqrt{\gamma_1}\,\dyad{0}{1},\nonumber\\
K_2 &= \sqrt{\gamma_2}\,\dyad{1}{2},
\end{align}
with $\sum_i K_i^\dagger K_i=\openone$.

\section{Full tier dynamics under the standard channels}\label{si:dynamics}

Figure~\ref{fig:dynamics_full} traces the three robustness tiers as a
function of time for the maximally mixed and the pure input, making
visible both the state-boundness of $\TNR$ and the NSIT-conditional
break of the upper hierarchy discussed in
section~\ref{sec:NSITbreaks} of the main text.

\begin{figure*}[t]
\includegraphics[width=\linewidth]{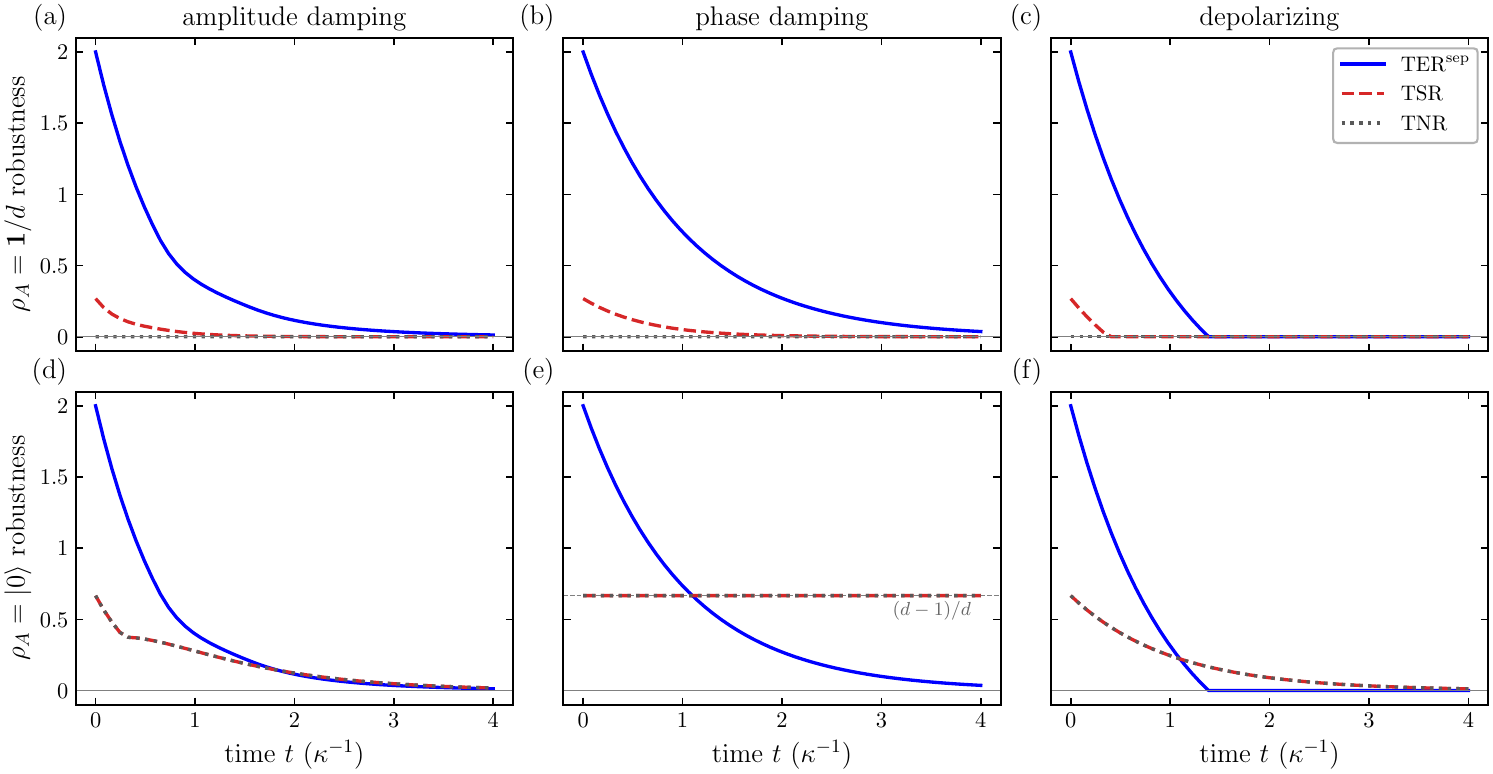}
\caption{Full dynamics of the genuine tiers $\TER^{\rm sep}$
(blue, solid; entanglement robustness of the channel Choi state,
computed via PPT) $\geq\TSR$ (red, dashed) $\geq\TNR$ (grey, dotted),
for a single qutrit on a dense 50-point time grid, under (left)
amplitude damping, (center) phase damping, (right) depolarization.
$\TER^{\rm sep}$ is a property of the \emph{channel} alone, hence
identical in both rows. Top row [(a)--(c)]: the maximally mixed input
$\rho_A=\openone/d$, which satisfies NSIT
(Theorem~\ref{thm:nsit_state} of the main text); the hierarchy
$\TER^{\rm sep}\geq\TSR\geq\TNR$ holds with $\TNR=0$ identically
(state-boundness, Corollary~\ref{cor:tnr_iff} of the main text). Bottom row
[(d)--(f)]: the pure, NSIT-violating input $\rho_A=|0\rangle$, where
$\TSR=\TNR$ collapse (dashed and dotted coincide); panel (e) shows the
closed-form plateau $\TSR=\TNR=(d-1)/d=2/3$ on $|0\rangle$+phase
damping (Proposition~\ref{prop:tsr_tnr_collapse} of the main text),
\emph{below which $\TER^{\rm sep}$ itself crosses} (here exact, since on
the depolarizing/isotropic Choi PPT${}={}$separability): the upper
inequality $\TER^{\rm sep}\geq\TSR$ genuinely \emph{fails} off
$\openone/d$ because $\TSR$ then carries a signaling contribution the
channel-only $\TER^{\rm sep}$ cannot bound
(section~\ref{sec:nsit-free} of the main text). The break is governed
by NSIT violation, not by the choice of TER quantifier.}
\label{fig:dynamics_full}
\end{figure*}

\section{High-precision $\TSR=\TNR$ collapse test}\label{si:numS4}

\begin{figure}[t]
\includegraphics[width=\linewidth]{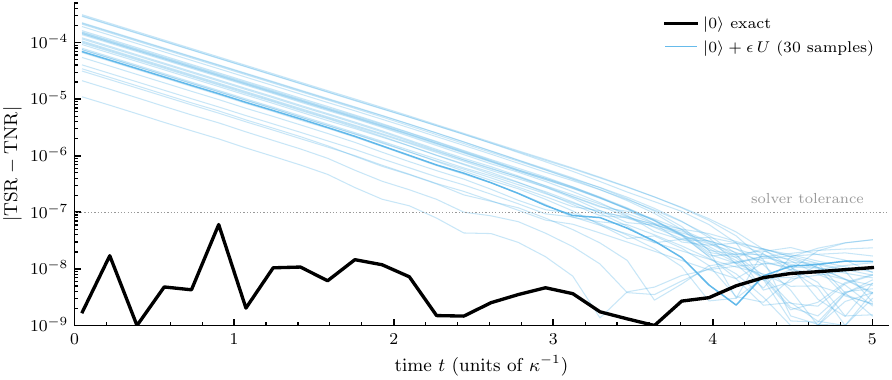}
\caption{High-precision test of the conjectured
$\TSR=\TNR$ equality on $|0\rangle$ + phase damping. The bold black
curve is the exact $|0\rangle$ initial state; thin curves are 30
small unitary perturbations ($\epsilon=0.05$). The exact case is
flat at solver tolerance ($\sim 10^{-8}$), confirming
$\TSR=\TNR$ \emph{exactly}; the perturbations escape this surface
immediately. The equality is therefore measure zero in the space of
initial states.}
\label{si:figs4}
\end{figure}

The numerical equality $\TSR=\TNR$ observed in
Fig.~\ref{fig:dynamics_full} on $|0\rangle$ + phase damping is
\emph{exact}, not coincidental. Solving the relevant SDPs at SCS
tolerance $10^{-7}$ on a fine time grid yields
$|\TSR-\TNR|\leq 6\times 10^{-8}$ for the precise initial state
$|0\rangle$, while a generic unitary perturbation moves the difference
to $\sim 10^{-4}$ within numerical tolerance. The equality is
therefore a measure-zero submanifold of (state, channel) space
(figure~\ref{si:figs4}). On $|0\rangle$ + phase damping the numerical
value is $\TSR=\TNR=2/3$ at $d=3$ and $\TSR=\TNR=4/5$ at $d=5$,
consistent with the closed-form expression $(d-1)/d$. An explicit
HSM ansatz (proof of the $\TSR=\TNR$ collapse proposition of the main text) saturates the upper
bound $\TSR\leq(d-1)/d$ for $|0\rangle$ + phase damping; SDP
optimality is confirmed numerically.

\section{Individual-probe pathologies of the DI teleportation fidelity}\label{si:fdi_probes}

Figure~\ref{fig:fdi} of the main text plots the fidelity averaged over
random probe states, showing that the \emph{typical} probe is an honest
certifier. Figure~\ref{fig:fdi_si} below resolves the same relation for
four individual probes---the computational eigenstates $|0\rangle$ and
$|1\rangle$, a mixed eigenstate $0.55|0\rangle\langle0|+0.45\,\openone/d$,
and the uniform superposition $|+\rangle$---and makes visible which probes
over-certify under which channel. The pattern is a direct consequence of
Proposition~\ref{prop:ditit_cert} of the main text: any channel-invariant
input over-certifies.

\begin{figure*}[t]
\includegraphics[width=\linewidth]{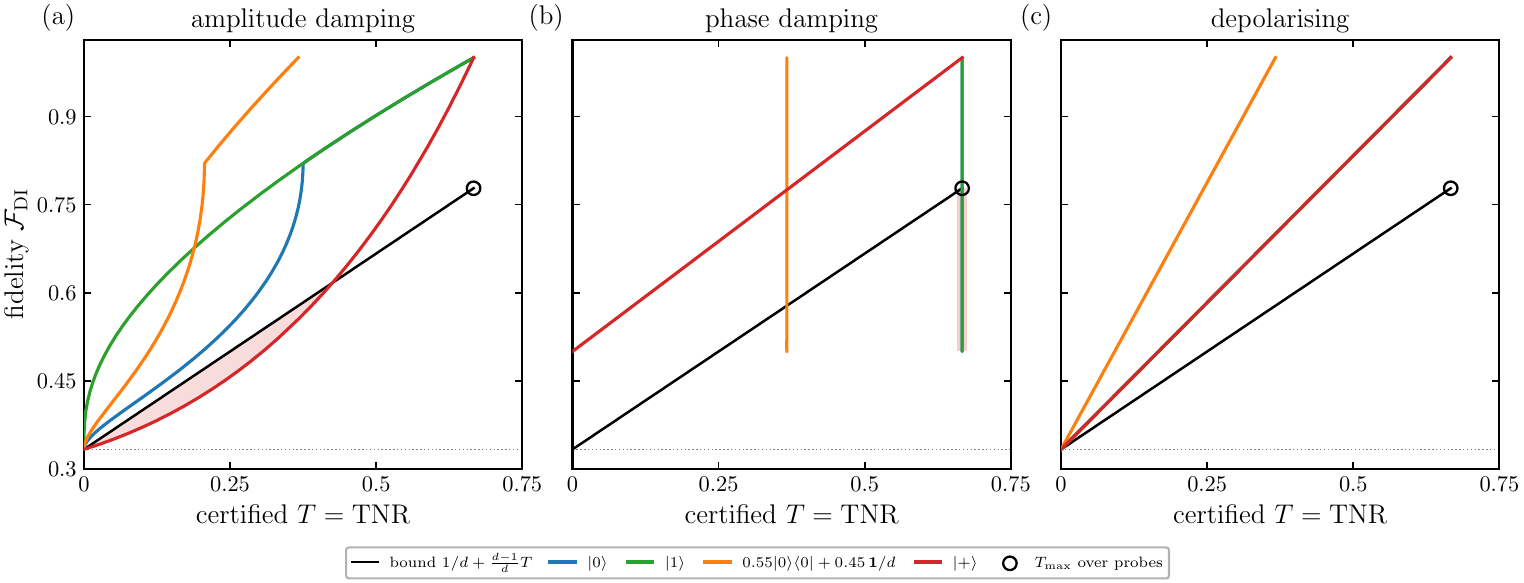}
\caption{Temporal-teleportation fidelity $\mathcal{F}_{\rm DI}$ versus the
certified $T=\TNR(\rho_A,\mathcal{E})$ at $d=3$, one column per channel,
for four individual probe states. Over-certification (curve dipping
\emph{below} the black bound~\eqref{eq:ditit_bound}) is highlighted by
the pink shading. (a)~Amplitude damping (biased toward the ground state):
the eigenbasis inputs $|0\rangle$, $|1\rangle$ and the mixed eigenstate
are honest, while the uniform superposition $|+\rangle$ over-certifies.
(b)~Phase damping: the eigenstates are \emph{fixed points} (lossless
storage), so $T$ is pinned at $(d-1)/d$ while $\mathcal{F}_{\rm DI}$
decays to $2/(d+1)$ (${=}1/2$ at $d=3$)---they over-certify below $7/9$;
the diagonal mixed probe is likewise $\mathcal{E}$-invariant and
over-certifies, whereas the non-invariant $|+\rangle$ is honest.
(c)~Depolarizing: isotropic (twirls to itself), so all pure states give
identical $T$ and every input certifies honestly. The circle marks the
supremal guaranteed point $\mathcal{F}_{\rm DI}=7/9$ at $T=(d-1)/d$.}
\label{fig:fdi_si}
\end{figure*}

\section{Analytic boundary of the achievable region}\label{si:envelopes}

The dense scatter clouds of Fig.~\ref{fig:hier_nsit_compare} display the
achievable region directly; we characterize its outer boundary here
\emph{analytically}, as a reference (it is no longer overlaid on the
figure). The boundary admits \emph{no global closed form}: it is traced by a
continuum of mixed inputs whose optimal spectrum \emph{and} eigenbasis
orientation both depend on the channel. We compute it as an optimization over
the generalized input state,
\begin{equation}\label{si:bopt}
B(x)=\max_{\rho\succeq0,\,\tr\rho=1,\,t\ge0}
\Big[\TER^{\rm sep}(\mathcal{E}_t)+\tfrac12\mathcal{V}(\rho,\mathcal{E}_t)\Big]
\ \text{s.t.}\ \TSR(\rho,\mathcal{E}_t)=x,
\end{equation}
and analogously for the dome $D(\mathcal{V})=\max\,[\TSR-\TNR]$. Symmetry
reduces the search: $\TER^{\rm sep}$ is input-independent (precomputed per
$t$); the metrics are basis-covariant, so $\rho$ is parametrized by its
spectrum together with an eigenbasis orientation modulo the channel's
stabilizer (full $U(d)$ for depolarizing, hence spectrum only; the
ground-state/computational-basis cosets for amplitude/phase damping). We solve
\eqref{si:bopt} by a dense symmetry-reduced sweep (real and complex
orientations) followed by gradient-free local polishing; the resulting
boundary contains all $10^6$ ($d=3$) and $5.6\times10^4$ ($d=5$)
Monte-Carlo samples to within solver tolerance. The exact analytic results
below give its value at special points and families that lie on the boundary;
every expression was verified against the $\rho_A$-adapted sweep to
$\leq 10^{-7}$.

\paragraph*{Channel parametrisation.} With $\kappa=1$ the depolarizing channel
(Sec.~\ref{si:channels}) carries survival parameter $\eta=e^{-t}$,
$\mathcal{E}_\eta(\rho)=\eta\rho+(1-\eta)\openone/d$, and the cascade
amplitude-damping channel (Sec.~\ref{si:cascade}) carries single-step damping
$\gamma=1-e^{-t}$.

\paragraph*{The NSIT monitor is sandwiched.} Writing
$\mathcal{V}\equiv\mathcal{V}_{\rm NSIT}$, the adapted two-MUB structure gives,
for \emph{every} input and channel,
\begin{equation}\label{si:sandwich}
\TNR\;\le\;\tfrac12\mathcal{V}\;\le\;\TSR ,
\end{equation}
with no violation over $>10^{3}$ random configurations (both inequalities tight
to $10^{-7}$). The lower inequality is an equality \emph{iff} the channel is the
identity or depolarizing; the upper one is an equality throughout the decohered
regime (amplitude/phase damping at any $t>0$; depolarizing for
$\eta\lesssim0.8$) and is strict only near the identity, where a steering excess
$\Delta\equiv\TSR-\tfrac12\mathcal{V}\ge0$ opens. Since the two plotted gaps are
$y_{\rm top}=\TER^{\rm sep}+\tfrac12\mathcal{V}-\TSR=\TER^{\rm sep}-\Delta$ and
$y_{\rm bot}=\TSR-\TNR$, relation~\eqref{si:sandwich} organises both rows.

\paragraph*{Maximally mixed probe (exact corner).} At the identity channel the
maximally mixed input has $\TNR=0$ (Proposition~\ref{prop:TNRstatebound}) and
\begin{equation}\label{si:tsr0}
\TSR(\openone/d,\,\mathrm{id})\;=\;2-\sqrt3\;=\;0.2679492\ldots\quad(d=3),
\end{equation}
the $\mathcal{V}\to0$ anchor of the bottom-row envelope. \emph{Proof}
(steering-robustness SDP dual): with two MUB settings the dual reads
$\max\sum_{x,a}\tr(F_{a|x}\sigma_{a|x})-1$ over $F_{a|x}\succeq0$ with
$F_{a|0}+F_{b|1}\preceq\openone$ for all $a,b$, where
$\sigma_{a|x}=\tfrac13\ket{\psi_{a|x}}\!\bra{\psi_{a|x}}$. The rank-one ansatz
$F_{a|x}=\alpha\ket{\psi_{a|x}}\!\bra{\psi_{a|x}}$ with the MUB overlap
$|\langle\psi_{a|0}|\psi_{b|1}\rangle|=1/\sqrt3$ has constraint top eigenvalue
$\alpha(1+1/\sqrt3)$, so feasibility forces
$\alpha=\sqrt3/(\sqrt3+1)=(3-\sqrt3)/2$; the objective is
$6\alpha\cdot\tfrac13-1=2\alpha-1=2-\sqrt3$, and a matching primal certifies
optimality.\hfill$\square$

\paragraph*{Exact closed forms} (eigenvalues $\lambda_k$ of $\rho_A$):
\begin{align}
\mathcal{V}_{\rm NSIT}(\rho,\mathcal{E}_\eta)
  &= \eta\textstyle\sum_k\lvert\lambda_k-1/d\rvert, \label{si:Vdepol}\\
\TER^{\rm sep}(\mathcal{E}_\eta)
  &= \tfrac{d-1}{d}\,[(d+1)\eta-1]_+ = \tfrac23[4\eta-1]_+, \\
\TNR=\TSR &= \tfrac12\mathcal{V}_{\rm NSIT}
  \quad(\text{depolarizing, all }\rho), \\
\TNR(\rho_r,\mathrm{id}) &= \tfrac{d-1}{d}\,r, \\
\rho_r &=(1-r)\openone/d+r\ketbra{\psi}{\psi}\ (\text{any pure }\psi), \\
\TNR(\ket\psi,\mathrm{id}) &= \tfrac{d-1}{d}\quad\text{for every pure }\psi .
\end{align}
For the basis state $\ket0$ under amplitude damping the monitor is elementary
with a kink at $\gamma=\tfrac14$,
\begin{equation}
\begin{gathered}
\mathcal{V}_{\ket0}(\gamma)=
\begin{cases}\tfrac43(1-\gamma)^2,&\gamma\le\tfrac14,\\[3pt]
\tfrac23(1-\gamma)(1+2\gamma),&\gamma\ge\tfrac14,\end{cases}\\[4pt]
\TSR_{\ket0}=\TNR_{\ket0}=\tfrac12\mathcal{V}_{\ket0},
\end{gathered}
\end{equation}
and the amplitude-damping Choi negativity is
$N(\gamma)=\tfrac{e}{3}\big(1+e+\sqrt{e^2-e+1}\big)$, $e=1-\gamma$.

\paragraph*{Envelopes, panel by panel.}
\begin{itemize}
\item \emph{High-$\mathcal{V}$ boundary segment: the identity limit.} Over the
large-$\mathcal{V}$ portion the boundary is traced by the identity-channel
family $\rho_r$, independent of the pure direction $\ket\psi$ (it
depends only on the spectrum), with $\mathcal{V}=2(d-1)r/d$. There
$\TER^{\rm sep}=d-1$ and $\TNR=\tfrac12\mathcal{V}$, so
$y_{\rm top}+y_{\rm bot}=d-1$ exactly, with
$y_{\rm bot}(\mathcal{V})=\TSR_{\rm id}(\tfrac34\mathcal{V})-\tfrac12\mathcal{V}$
running from $2-\sqrt3$ at $\mathcal{V}\to0$ down to $0$ at
$\mathcal{V}^\ast\approx0.943$, beyond which $y_{\rm top}=d-1=2$. Here
$\TSR_{\rm id}$ is a genuine SDP value with \emph{no elementary closed form}
(tabulated below).
\item \emph{Depolarizing top-row pure-state locus (exact):}
$y_{\rm top}=\max(0,\,2\mathcal{V}-\tfrac23)$, kink at $\mathcal{V}=\tfrac13$.
This is the locus of pure inputs; it is \emph{interior} to the cloud (mixed
inputs fall $\sim0.05$ below), hence not the true lower envelope, which has no
elementary form.
\item \emph{Bottom-row floor (exact):} $y_{\rm bot}=0$, attained by all pure
inputs and by basis / maximally mixed inputs.
\item \emph{Amplitude-damping lower (top) and dome (bottom):} traced by $\ket+$
respectively $\ket0$--$\ket1$ coherent inputs swept over $\gamma$; exact only
\emph{parametrically} ($\mathcal{V}_+(\gamma)$ is an irreducible cubic, while
$\TER^{\rm sep}(\gamma)$ and $\TNR_{\rm amp}(\gamma)$ are SDP/LP values). The
bottom dome peaks at $y\approx0.148$ near $\mathcal{V}\approx0.68$.
\end{itemize}

\noindent Identity-limit envelope (non-elementary part), $d=3$:
\begin{center}\small
\begin{tabular}{l|ccccccc}
$\mathcal{V}$ & $0$ & $0.27$ & $0.53$ & $0.67$ & $0.80$ & $0.93$ & $\geq0.94$\\
\hline
$\TSR_{\rm id}$ & $0.2679$ & $0.2829$ & $0.3284$ & $0.3637$ & $0.4089$ & $0.4667$ & $\tfrac12\mathcal{V}$\\
$y_{\rm bot}$ & $0.2679$ & $0.1495$ & $0.0618$ & $0.0304$ & $0.0089$ & $0.0000$ & $0$\\
\end{tabular}
\end{center}

\noindent These envelopes \emph{refine but do not alter} the central results:
the state-boundness equivalence (Corollary~\ref{cor:tnr_iff}), the universal
bound $\TSR\le\TER^{\rm sep}+\tfrac12\mathcal{V}$
(Proposition~\ref{prop:univ_bound}, whose constant $\tfrac12$ is reconfirmed
exactly tight here), and the strict hierarchy
(Theorem~\ref{thm:hierarchy_strict}) all stand; relation~\eqref{si:sandwich} is
a sharpening, inserting $\tfrac12\mathcal{V}_{\rm NSIT}$ between $\TNR$ and
$\TSR$.

\section{Even and prime-power dimensions ($d=2,4$)}\label{si:dvar}

For completeness we reproduce the hierarchy-gap and universal-bound panels
of Fig.~\ref{fig:hier_nsit_compare} at the two non-odd-prime
dimensions of Table~\ref{tab:d3d5_compare}: the qubit $d=2$ and the prime power
$d=4$. The qualitative picture is identical to the odd primes $d=3,5$: TNR
vanishes only at $\openone/d$, the hierarchy $\TNR\le\TSR$ holds throughout, and
the universal bound $\TSR\le\TER^{\rm sep}+\tfrac12\mathcal{V}_{\rm NSIT}$ is
respected (saturated by the $\ket0$+phase-damping family). The envelopes carry
over with the $d$-dependent identity anchor
$\TSR(\openone/d,\mathrm{id})=(\sqrt d-1)/(\sqrt d+1)$ [$=3-2\sqrt2\approx0.172$
at $d=2$ and $1/3$ at $d=4$; cf.\ $2-\sqrt3$ at $d=3$, eq.~\eqref{si:tsr0}],
maximal NSIT violation $2(d-1)/d$, and ceiling $\TER^{\rm sep}=d-1$. (Even $d$ is
outside the odd-prime scope of the state-boundness \emph{theorem}, but the
numerics show the same phenomenology; $d=4$ uses the prime-power MUB pair.)

\begin{figure*}[t]
\includegraphics[width=\linewidth]{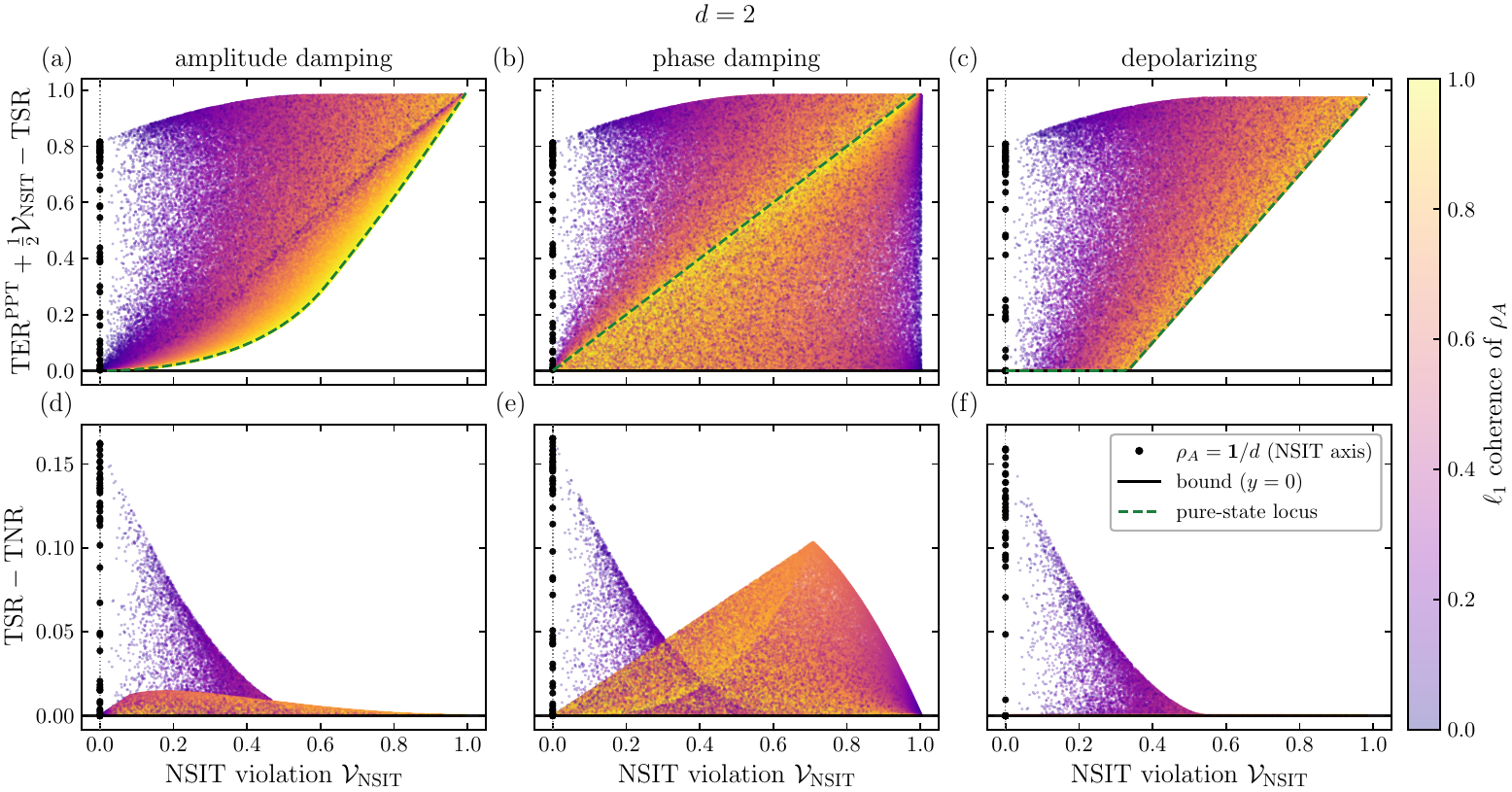}\\[4pt]
\includegraphics[width=\linewidth]{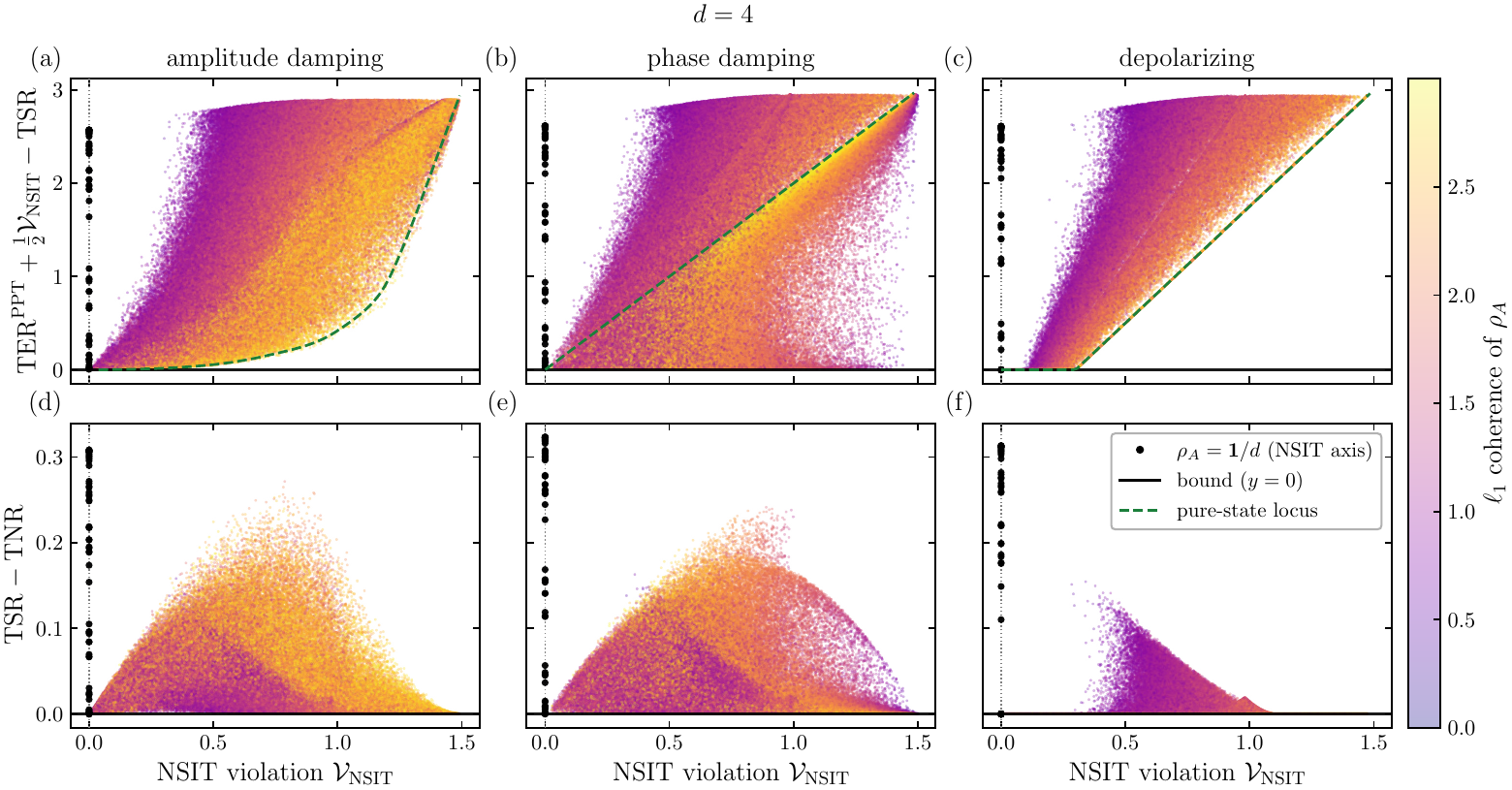}
\caption{Hierarchy gaps versus $\mathcal{V}_{\rm NSIT}$ at $d=2$ (qubit, top)
and $d=4$ (prime power, bottom)---the analogues of
Fig.~\ref{fig:hier_nsit_compare}. Dashed green: the analytic pure-state
locus; black dots: the $\rho_A=\openone/d$ axis ($\TNR=0$).}
\label{fig:dvar_hier}
\end{figure*}

\begin{figure*}[t]
\includegraphics[width=\linewidth]{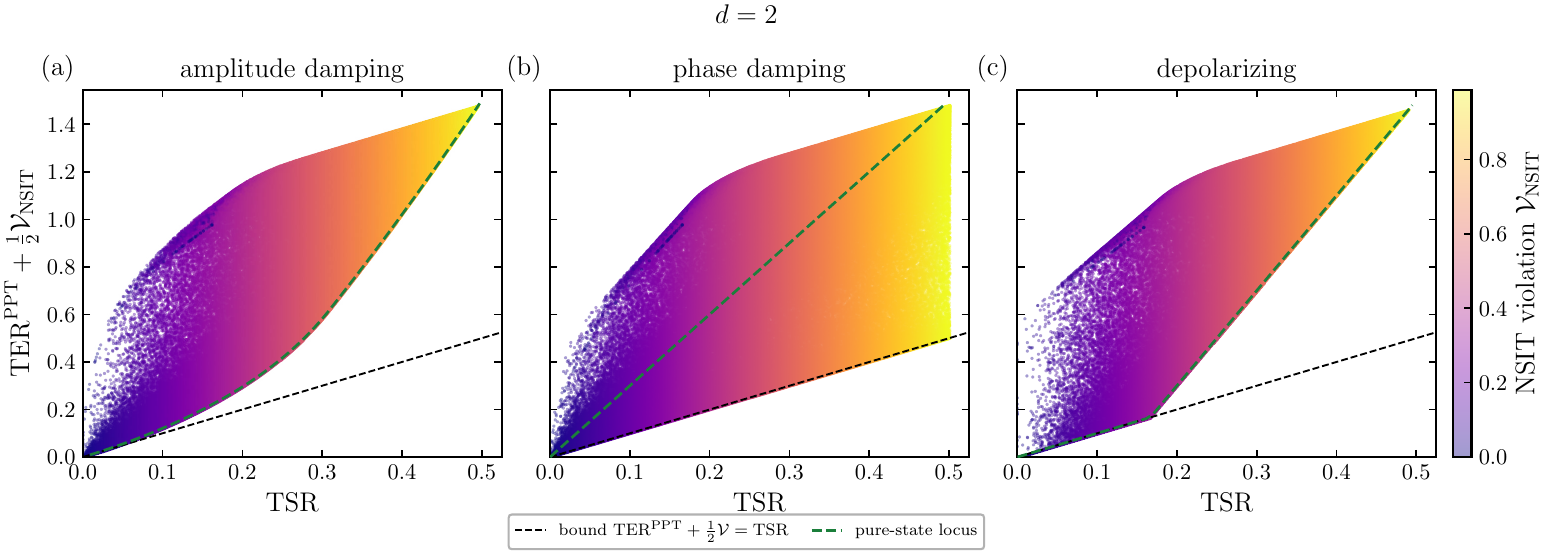}\\[4pt]
\includegraphics[width=\linewidth]{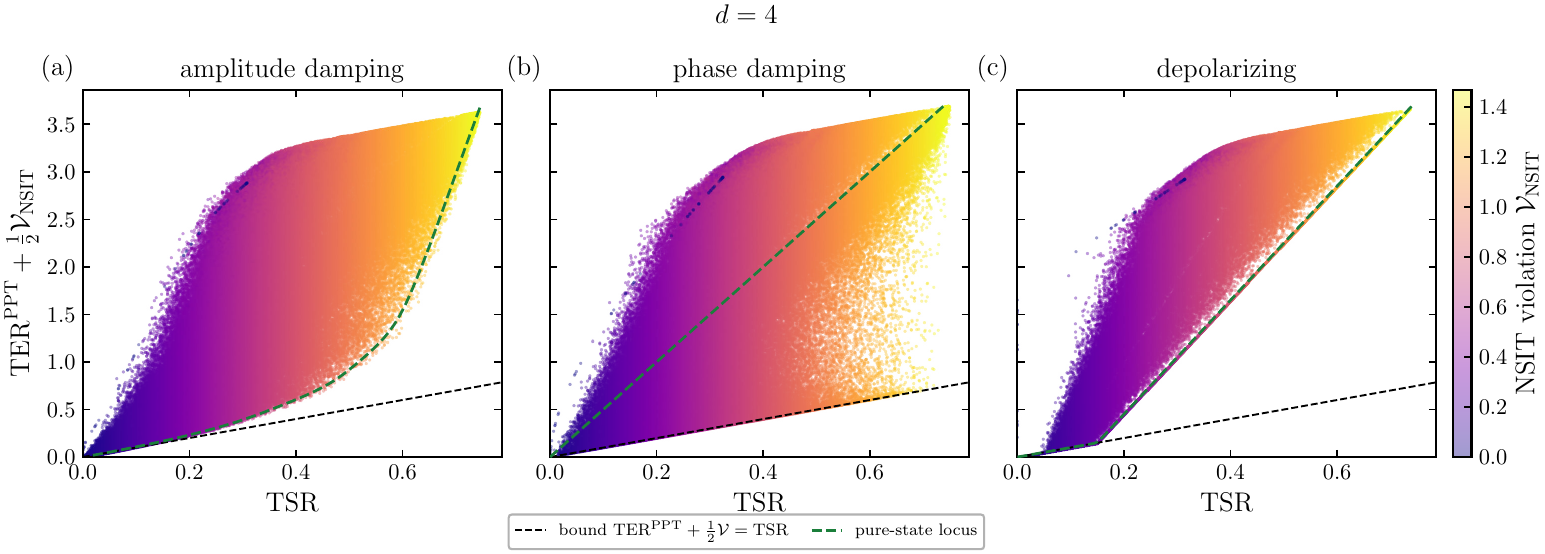}
\caption{The universal bound in the SDP-computable PPT proxy
$\TER^{\rm PPT}+\tfrac12\mathcal{V}_{\rm NSIT}$ ($\TER^{\rm PPT}\leq\TER^{\rm sep}$)
versus $\TSR$ at $d=2$ (top)
and $d=4$ (bottom)---the analogues of panels (g)--(i) of
Fig.~\ref{fig:hier_nsit_compare}. Every
point lies on or above the dashed diagonal; dashed green is the analytic
pure-state locus.}
\label{fig:dvar_terpdo}
\end{figure*}

\section{Operational supplements}\label{si:operational}

\subsection{A qubit to the future: the BB84 picture}
Here is the qubit ($d=2$) story behind the protocol of
section~\ref{sec:ditit} of the main text. Alice wants to hand an unknown
qubit $|\psi\rangle$ to her future self (Bob) at $t_B$, but the only
thing that survives from $t_A$ to $t_B$ is one noisy quantum memory
cell---the channel $\mathcal{E}$---together with an ordinary classical
notebook. No entangled partner is stored away: the two-time correlations
of the memory are the entire quantum resource. She proceeds as in a
\emph{quantum one-time pad}: (i)~she flips two fair coins $k=(k_1,k_2)$
and applies the Pauli $W_k=X^{k_1}Z^{k_2}$, so
$|\psi\rangle\mapsto W_k|\psi\rangle$, which to anyone lacking $k$ is
indistinguishable from $\openone/2$; (ii)~the scrambled qubit is written
into the cell and read out at $t_B$, i.e.\ passed through $\mathcal{E}$;
(iii)~Alice publishes $k$ and Bob undoes the pad,
$W_k^\dagger(\cdot)\,W_k$. Because every send-round input looks maximally
mixed without $k$, the qubit's quantum content reaches the future
\emph{only} through $\mathcal{E}$; Bob recovers the depolarised image
$p\,|\psi\rangle\langle\psi|+(1-p)\openone/2$ with fidelity
$\mathcal{F}_{\rm DI}=\tfrac12+\tfrac12 p$, so a faithful hand-off
requires $p>0$.

To trust the memory \emph{before} risking $|\psi\rangle$, Alice
interleaves \emph{test rounds} feeding a known probe $\rho_A$, measuring
(at random) in one of the two MUBs ($Z$ and $X$, the BB84 pair at
$d=2$). From the two-time statistics she computes $\TNR$ and, as BB84
aborts on excess disturbance, aborts unless $\TNR\geq T^*$; a passing
test certifies $\mathcal{F}_{\rm DI}\geq\tfrac12+\tfrac12\TNR$. Temporal
teleportation is thus a quantum one-time pad sent through the
time-channel and monitored, BB84-style, by two-MUB test rounds. The
general-$d$ protocol replaces the Pauli pad by the Heisenberg--Weyl group
and the two BB84 bases by the $\rho_A$-adapted two-MUB scheme.

\subsection{The NSIT violation as a physical resource}
At first sight $\mathcal{V}_{\rm NSIT}>0$ looks unphysical: Bob's
marginal carries information about Alice's setting, forbidden by
no-signaling in the spatial Bell scenario. Temporally there is no
paradox---Alice and Bob act on the \emph{same} system at ordered times
$t_A<t_B$, so $\mathcal{V}_{\rm NSIT}$ is causal measurement back-action,
the lawful disturbance one party imprints on the system the other later
receives. The appearance of unphysicality is an artifact of demanding a
no-signaling assemblage---exactly the demand that makes the
entanglement-over-steering bound NSIT-conditional. The back-action is
operationally accessible in two ways native to teleportation.
(i)~\emph{Communication}: the protocol provides a public classical
channel carrying Alice's setting and outcome; the signaling part of the
assemblage is then a resource Bob exploits, and the
$\tfrac12\mathcal{V}_{\rm NSIT}$ term of
Proposition~\ref{prop:univ_bound} of the main text measures precisely the
steering this unlocks beyond the channel's entanglement.
(ii)~\emph{Postselection}: conditioning on Alice's outcome replaces the
outcome-averaged marginal by a postselected conditional state---a valid
sub-normalized object---accessing the same resource by selecting rounds.
Either way $\mathcal{V}_{\rm NSIT}$ is a bona fide physical resource.

\subsection{One-sided device-independent teleportation}
For the one-sided variant in which only Bob's device is untrusted, the
resource is the assemblage $\{\tilde\rho_{a|x}\}$ rather than the full
behavior. The analogous Abramsky--Brandenburger argument for the
steering fraction $\mathrm{SF}(\{\tilde\rho_{a|x}\})$
\cite{Branciard2012,PhysRevLett.114.060404,PhysRevA.93.052112} gives
\begin{align}\label{eq:1sdi_bound}
\mathcal{F}_{\rm 1sDI}(\rho_A,\mathcal{E})
&\geq\frac{1}{d}+\frac{d-1}{d}\,\mathrm{SF}\bigl(\{\tilde\rho_{a|x}\}\bigr)\nonumber\\
&\geq\frac{1}{d}+\frac{d-1}{d}\,\frac{\TSR(\rho_A,\mathcal{E})}{1+\TSR(\rho_A,\mathcal{E})},
\end{align}
the second step using the steerable-weight versus steering-robustness
identity of Cavalcanti--Skrzypczyk \cite{PhysRevA.93.052112}. Since
$\TSR\geq\TNR$ (Theorem~\ref{thm:tsr_tnr} of the main text) and
$x/(1+x)$ is monotone, $\mathcal{F}_{\rm 1sDI}\geq\mathcal{F}_{\rm DI}$
pointwise, with equality on the $\TSR=\TNR$ submanifold (e.g.\ the
$|0\rangle$+phase-damping pair of
Proposition~\ref{prop:tsr_tnr_collapse} of the main text): the one-sided
protocol is strictly easier to operate but requires Alice's measurement
device to be trusted. An $\mathrm{SF}=\TSR$ conjecture on the standard
families would tighten~\eqref{eq:1sdi_bound} to
$\mathcal{F}_{\rm 1sDI}\geq 1/d+(d-1)\TSR/d$.

\section{Proofs of the technical results}\label{si:proofs}

\paragraph*{Proof of Theorem~\ref{thm:TERmonotone}.}
Write any PDO as $R=(1+\gamma)\rho-\gamma\,\mathfrak{R}$. Diagonalising
$R=\sum_i\mu_i\dyad{\mu_i}$, the optimum is attained for
$\gamma\,\mathfrak{R}=\sum_{i:\mu_i<0}|\mu_i|\,\dyad{\mu_i}\geq 0$,
proving (i): TER $\geq 0$ with equality iff $R\geq 0$, and TER is
maximized on closed-system unitary evolution. Since the trace of $R$ is
invariant under unitaries, and $\rho,\mathfrak{R}$ have unit trace,
TER is invariant under the same, proving (ii). For (iii), let
$\mathcal{E}$ be a local CPTP map. Then
$\mathcal{E}(R)=(1+\gamma)\mathcal{E}(\rho)-\gamma\,\mathcal{E}(\mathfrak{R})$;
since $\mathcal{E}$ is CP, $\mathcal{E}(\rho)\geq 0$ and
$\mathcal{E}(\mathfrak{R})\geq 0$ with
$\tr\mathcal{E}(\mathfrak{R})\leq\tr\mathfrak{R}$ by trace
non-increase, so this is a feasible decomposition of
$\mathcal{E}(R)$ at cost $\leq\gamma$; hence
$\TER(\mathcal{E}(R))\leq\TER(R)$. For (iv), pick optimal
decompositions $\gamma_k$ for $R_k$ and combine them: the convex
combination $R=\sum_k p_kR_k$ satisfies
$R=(1+\gamma)\rho-\gamma\,\mathfrak{R}$ with
$\gamma=\sum_k p_k\gamma_k$, so $\TER(R)\leq\sum_k p_k\,\TER(R_k)$.

\paragraph*{Proof of Proposition~\ref{prop:TNRstatebound}.}
On $\rho_A=\openone/d$ and projective $\Pi_{a|x}$, the Born rule
\eqref{eq:born} gives
$P(a,b|x,y) = \tr[M_{b|y}\mathcal{E}(\Pi_{a|x}\rho_A\Pi_{a|x})]
=(1/d)\tr[M_{b|y}\mathcal{E}(\Pi_{a|x})]$, which is
\eqref{eq:TNRHSM}. By assumption this admits an LHV
decomposition, hence $\TNR=0$ by definition.

\paragraph*{Proof of Proposition~\ref{prop:lhv_reduction}.}
\emph{Depolarizing.} Direct computation gives
\begin{equation*}
\begin{aligned}
P(a,b|x,y)&=\tfrac1d\tr[M_{b|y}\mathcal{E}_{\rm depol}^{(e)}(\Pi_{a|x})]\\
&=\begin{cases}
(e/d)\delta_{a,b}+(1-e)/d^{2} & x=y,\\
1/d^{2} & x\neq y.
\end{cases}
\end{aligned}
\end{equation*}
This admits an LHV with hidden variable
$\lambda=(\alpha_{0},\alpha_{1},\gamma_{0},\gamma_{1},\varepsilon)
\in\mathbb{Z}_{d}^{4}\times\{C,U\}$, distributed as $\alpha_{x}$ and
$\gamma_{x}$ independently uniform on $\mathbb{Z}_{d}$ for $x=0,1$
and $\Pr[\varepsilon=C]=e$, $\Pr[\varepsilon=U]=1-e$. Alice's
deterministic response is $D_A(a|x,\lambda)=\delta_{a,\alpha_{x}}$;
Bob's is $D_B(b|y,\lambda)=\delta_{b,\alpha_{y}}$ when
$\varepsilon=C$ and $D_B(b|y,\lambda)=\delta_{b,\gamma_{y}}$ when
$\varepsilon=U$. Direct verification: for $x=y$,
$\Pr[a=\alpha_{x},b=\alpha_{x}]=e/d\cdot\delta_{a,b}$ (correlated
branch) plus $\Pr[a=\alpha_{x},b=\gamma_{y}]=(1-e)/d^{2}$
(uncorrelated branch) reproduces the same-basis $P$; for $x\neq y$
both branches give $1/d^{2}$, reproducing the cross-basis $P$.

\emph{Phase damping.} The same computation gives
\begin{equation*}
P(a,b|x,y)=
\begin{cases}
(1/d)\delta_{a,b} & x=y=Z,\\
(e/d)\delta_{a,b}+(1-e)/d^{2} & x=y=F,\\
1/d^{2} & x\neq y,
\end{cases}
\end{equation*}
which differs from depolarizing only in the $x=y=0$ slot, where
phase damping leaves the diagonal computational projector invariant.
An LHV with
$\lambda=(\alpha_{0},\alpha_{1},\gamma_{1},\varepsilon)
\in\mathbb{Z}_{d}^{3}\times\{C,U\}$, distributed as $\alpha_{0},
\alpha_{1},\gamma_{1}$ independently uniform on $\mathbb{Z}_{d}$ and
$\Pr[\varepsilon=C]=e$, $\Pr[\varepsilon=U]=1-e$, with responses
$D_A(a|x,\lambda)=\delta_{a,\alpha_{x}}$,
$D_B(b|0,\lambda)=\delta_{b,\alpha_{0}}$, and
$D_B(b|1,\lambda)=\delta_{b,\alpha_{1}}$ if $\varepsilon=C$,
$\delta_{b,\gamma_{1}}$ if $\varepsilon=U$, reproduces all four
joint probabilities by the same direct verification.

\paragraph*{Proof of Proposition~\ref{prop:product_lhv}.}
Non-negativity and normalization are immediate (product of two
probability distributions). We verify the four measurement-pair
blocks. (i)~$P_{\rm LHV}(a,b|Z,Z)
=\sum_{\alpha_F,\beta_F}P(a,b|Z,Z)\,P(\alpha_F,\beta_F|F,F)
=P(a,b|Z,Z)$. (ii)~$P_{\rm LHV}(a,b|F,F)
=\sum_{\alpha_Z,\beta_Z}P(\alpha_Z,\beta_Z|Z,Z)\,P(a,b|F,F)
=P(a,b|F,F)$. (iii)~$P_{\rm LHV}(a,b|Z,F)
=\bigl[\sum_{\beta_Z}P(a,\beta_Z|Z,Z)\bigr]
\bigl[\sum_{\alpha_F}P(\alpha_F,b|F,F)\bigr]$.
On $\rho_A=\openone/d$, both marginals are $1/d$: Alice's $Z$
marginal is $(1/d)\tr[\mathcal{E}(\Pi_{a|Z})]=1/d$, and Alice's
$F$ marginal is $(1/d)\tr[\mathcal{E}(\Pi_{a|F})]=1/d$ (the Fourier
projectors have uniform diagonal). Hence
$P_{\rm LHV}(a,b|Z,F)=1/d^{2}$, which matches the actual behavior
because $\tr[\Pi_{b|F}\,\mathcal{E}(|a\rangle\langle
a|)]=(1/d)\sum_j[\mathcal{E}(|a\rangle\langle a|)]_{jj}=1/d$ by
the diagonal-action hypothesis. (iv)~$P_{\rm LHV}(a,b|F,Z)
=P_A(a|F)\cdot P_B(b|Z)=(1/d)\cdot P_B(b|Z)$, which matches
$P(a,b|F,Z)$ because on $\openone/d$ every Fourier projector
$|+_a\rangle\langle +_a|$ has the \emph{same} diagonal $(1/d,
\ldots,1/d)$ and the diagonal-action channel therefore produces
the same output diagonal $[\mathcal{E}(\openone/d)]_{bb}$ for
every $a$.

\paragraph*{Proof of Proposition~\ref{prop:unitary_invariance}.}
The two-time behavior obeys
\begin{align*}
P'(a,b|x,y)&=\tr\!\big[(UM_{b|y}U^\dagger)\,U\mathcal{E}U^\dagger
(U\Pi_{a|x}U^\dagger\,\rho_A\,U\Pi_{a|x}U^\dagger)\big]\\
&=\tr[M_{b|y}\,\mathcal{E}(\Pi_{a|x}\rho_A\Pi_{a|x})]=P(a,b|x,y)
\end{align*}
by the cyclicity of trace and $[U,\rho_A]=0$. An LHV decomposition
of $P$ is therefore an LHV decomposition of $P'$, and the TNR
optimum is invariant.

\paragraph*{Proof of Theorem~\ref{thm:nsit_state}.}
($\Leftarrow$) On $\rho_A = \openone/d$, the dephased state in any
setting $x$ is
$\sum_a \Pi_{a|x}(\openone/d)\Pi_{a|x} = (1/d)\sum_a \Pi_{a|x} = \openone/d$,
independent of $x$. Consequently Bob's marginal
$\sum_a P(a,b|x,y) = \tr[M_{b|y}\,\mathcal{E}(\openone/d)]$ is
independent of $x$, so $\mathcal{V}_{\rm NSIT}=0$ for any
$\mathcal{E}$.

($\Rightarrow$) Contrapositive. Suppose $\rho_A \neq \openone/d$.
For the eigenbasis setting $x = Z$, the dephased state is
$\sum_a \Pi_{a|Z}\rho_A\Pi_{a|Z} = \rho_A$ (diagonal in its own
eigenbasis). For the Fourier-MUB setting $x = F$, since the Fourier
MUB is unbiased to the eigenbasis,
$\Pi_{a|F}\rho_A\Pi_{a|F}
= \langle +_a|\rho_A|+_a\rangle\,\Pi_{a|F}
= (1/d)\,\Pi_{a|F}$
(since $\langle +_a|\rho_A|+_a\rangle=(1/d)\sum_j p_j=1/d$ by
unbiasedness), so the sum equals
$(1/d)\sum_a\Pi_{a|F}=\openone/d$.
The two dephased states $\rho_A$ and $\openone/d$ differ when
$\rho_A \neq \openone/d$, and the injectivity of $\mathcal{E}$ on
their affine span makes $\mathcal{E}(\rho_A) \neq \mathcal{E}(\openone/d)$.
Therefore Bob's marginal $\tr[M_{b|y}\,\mathcal{E}(\cdot)]$ differs
between $x=Z$ and $x=F$ for at least one $(b,y)$, giving
$\mathcal{V}_{\rm NSIT} > 0$.

\paragraph*{Proof of Theorem~\ref{thm:tsr_tnr}.}
The steering robustness $\TSR(R)$ is the least weight $\alpha$ for
which the mixed assemblage
$\sigma_{a|x}^{(\alpha)}=(\tilde\rho_{a|x}+\alpha\,\nu_{a|x})/(1+\alpha)$
admits a hidden-state model (HSM)
$\sigma_{a|x}^{(\alpha)}=\sum_\lambda p(\lambda)D(a|x,\lambda)\rho_\lambda$
for some noise assemblage $\nu$. Applying Bob's measurements
$\{M_{b|y}\}$ to both sides,
$P^{(\alpha)}(a,b|x,y)=\tr[M_{b|y}\sigma_{a|x}^{(\alpha)}]
=\sum_\lambda p(\lambda)D(a|x,\lambda)\tr(M_{b|y}\rho_\lambda)$
is local-hidden-variable, with $P^{(\alpha)}=(P+\alpha P_\nu)/(1+\alpha)$
and the valid noise behavior $P_\nu(a,b|x,y)=\tr[M_{b|y}\nu_{a|x}]$.
Hence every HSM at weight $\alpha$ for the assemblage yields an LHV
decomposition at the same weight for the behavior, so
$\TNR(R)\leq\TSR(R)$ ($\alpha=0$ recovers
$\TSR=0\Rightarrow\TNR=0$); non-negativity is immediate.

\paragraph*{Proof of Theorem~\ref{thm:hierarchy_strict}.}
Suppose the Choi state $\Lambda_{\mathcal{E}}$ has separability
robustness $\TER^{\rm sep}=\alpha$. Then there exist a separable
state $\Lambda_{\rm sep}$ and a state $\Lambda_\tau$ such that
$(\Lambda_{\mathcal{E}}+\alpha\Lambda_\tau)/(1+\alpha)
=\Lambda_{\rm sep}$. By the Choi--channel correspondence, the
mixture $\mathcal{E}_{\rm sep}
=(\mathcal{E}+\alpha\tau)/(1+\alpha)$ is entanglement-breaking and
admits a measure-and-prepare form
$\mathcal{E}_{\rm sep}(\sigma)
=\sum_k\tr(F_k\sigma)\omega_k$ for a POVM $\{F_k\}$ and states
$\{\omega_k\}$~\cite{HORODECKI20031}. On $\rho_A=\openone/d$, where $\Pi_{a|x}\rho_A\Pi_{a|x}
=\tfrac1d\Pi_{a|x}$, the assemblage element satisfies
$(\tilde\rho_{a|x}+\alpha\tilde\rho_{a|x}^\tau)/(1+\alpha)
=\sum_k\tfrac1d\tr(F_k\Pi_{a|x})\,\omega_k
=\sum_k q(k)\,D(a|x,k)\,\omega_k$,
with $x$-independent weights $q(k)=\tr(F_k)/d$ ($\sum_k q(k)=1$) and
the valid stochastic response
$D(a|x,k)=\tr(F_k\Pi_{a|x})/\tr(F_k)$
($\sum_a D(a|x,k)=1$, $D\geq 0$). This is a hidden-state model for the
assemblage (made deterministic, if desired, by the standard
convex-decomposition of stochastic responses), hence
$\TSR\leq\alpha=\TER^{\rm sep}$. The $x$-independence of $q(k)$ is
exactly what fails off $\openone/d$ (section~\ref{sec:nsit-free}). The lower part
$\TSR\geq\TNR\geq 0$ is Theorem~\ref{thm:tsr_tnr}.

\paragraph*{Proof of Proposition~\ref{prop:tsr_tnr_collapse}.}
\label{sec:tsrtnr_proof_sketch}
Take the deterministic-strategy index $\lambda=(a^0,a^1)$ and the
HSM ansatz
$\sigma_{(0,a)}=\tilde\rho_{a|1}+\frac{d-1}{d^2}|0\rangle\langle 0|$,
$\sigma_{(a^0,a^1)}=0$ for $a^0\neq 0$, where
$\tilde\rho_{a|1}=\frac{1-e}{d^2}\openone+\frac{e}{d}|+_a\rangle\langle+_a|$
is the assemblage element. The $x=0$ constraint
$\sum_{a^1}\sigma_{(0,a^1)}\geq|0\rangle\langle 0|$ holds because
$\sum_{a}\tilde\rho_{a|1}=\openone/d$ contributes $\openone/d$ on the
$|0\rangle$ subspace, and the additional
$\frac{d-1}{d^2}|0\rangle\langle 0|$ summed $d$ times provides the
remaining $\frac{d-1}{d}|0\rangle\langle 0|$. The $x=1$ constraint
$\sum_{a^0}\sigma_{(a^0,a)}\geq\tilde\rho_{a|1}$ is direct since
$\sigma_{(0,a)}\geq\tilde\rho_{a|1}$. The total cost is
$\sum_\lambda\tr\sigma_\lambda - 1
=d\cdot[1/d+(d-1)/d^2]-1=(d-1)/d$, establishing the analytic upper
bound $\TSR\leq(d-1)/d$ for all $e\in(0,1]$. By
Theorem~\ref{thm:tsr_tnr} the same bound applies to the nonlocality
robustness, $\TNR\leq\TSR\leq(d-1)/d$. Saturation---that this upper
bound is attained, so $\TSR=\TNR=(d-1)/d$ independent of $e$---is
confirmed by the numerical SDP optimum at $d=3$ ($2/3$) and $d=5$
($4/5$); we do not have a closed-form dual certificate for the
matching lower bound.

\paragraph*{Proof of Proposition~\ref{prop:univ_bound}
(NSIT-corrected universal bound).}
The lower chain $0\leq\TNR\leq\TSR$ is Theorem~\ref{thm:tsr_tnr}. For
the upper bound we construct an explicit hidden-state over-cover. Put
$\alpha=\TER^{\rm sep}=\mathrm{ER}(\Lambda_{\mathcal{E}})$, so
$(\Lambda_{\mathcal{E}}+\alpha\,\Omega)/(1+\alpha)$ is separable for
some state $\Omega$; the associated entanglement-breaking channel has a
measure-and-prepare form
$\mathcal{E}_{\rm sep}(\cdot)=\sum_k\tr(F_k\,\cdot)\,\omega_k$ with
$\{F_k\}$ a POVM and $\{\omega_k\}$ states, and dominates the channel,
\begin{equation*}
\begin{gathered}
\tilde\rho_{a|x}=\mathcal{E}\bigl(\Pi_{a|x}\rho_A\Pi_{a|x}\bigr)
\;\preceq\;(1+\alpha)\sum_k w_{a|x,k}\,\omega_k,\\
w_{a|x,k}=\tr\!\bigl(F_k\Pi_{a|x}\rho_A\Pi_{a|x}\bigr)\geq0 .
\end{gathered}
\end{equation*}
Take hidden variable $k$ with states $\omega_k$, weights
$q(k)=(1+\alpha)\max_x\tr\!\bigl(F_k D_x(\rho_A)\bigr)$ and responses
$r(a|x,k)=(1+\alpha)w_{a|x,k}/q(k)\in[0,1]$ (with
$\sum_a r(a|x,k)\leq1$). Then
$\sum_k q(k)\,r(a|x,k)\,\omega_k=(1+\alpha)\sum_k w_{a|x,k}\omega_k
\succeq\tilde\rho_{a|x}$ is a feasible local-hidden-state over-cover of
the steering SDP \eqref{eq:TSR_SDP} of the main text, of cost
$\sum_k q(k)-1$. With $\Delta=D_Z(\rho_A)-D_F(\rho_A)$ traceless
Hermitian and $\{F_k\}$ a POVM,
\begin{equation*}
\sum_k\max_x\tr\!\bigl(F_kD_x(\rho_A)\bigr)
=1+\tfrac12\sum_k\bigl|\tr(F_k\Delta)\bigr|
\;\leq\;1+\tfrac12\|\Delta\|_1,
\end{equation*}
the factor $\tfrac12$ being the ``traceless $\Rightarrow$ half the
trace-norm'' identity. Hence
\begin{equation}\label{eq:univ_rig}
\TSR\;\leq\;\TER^{\rm sep}
+\tfrac12\,(1+\TER^{\rm sep})\,
\bigl\|D_Z(\rho_A)-D_F(\rho_A)\bigr\|_1 ,
\end{equation}
a rigorous universal bound. It is tight: on $|0\rangle$+phase damping
at strong dephasing $\TER^{\rm sep}=0$,
$\|\Delta\|_1=2(d-1)/d$ and $\TSR=(d-1)/d$, attaining equality.

Equation~\eqref{eq:univ_bound} of the main text is the sharp
post-channel refinement of \eqref{eq:univ_rig}: since
$\mathcal{V}_{\rm NSIT}=\|\mathcal{E}(\Delta)\|_1\leq\|\Delta\|_1$
(channel contraction) and $\TER^{\rm sep}\geq0$, it lies below the
right-hand side of \eqref{eq:univ_rig} and coincides with it in the
saturating limit. We verify it with zero violations to machine
precision (minimum slack $-3.9\times10^{-8}$) across the full
$\rho_A$-adapted sweeps at $d=3$ ($10^6$) and $d=5$
($5.6\times10^4$), and with no violation also at $d=2,4$
(Table~\ref{tab:d3d5_compare} of the main text); removing the prefactor
$(1+\TER^{\rm sep})$ and replacing the pre-channel signaling by
$\mathcal{V}_{\rm NSIT}$ is the only step that rests on the numerics
rather than the construction
(figure~\ref{fig:hier_nsit_compare} of the main text). The sharp form
reduces to a single entanglement-theoretic statement: the
no-signaling part of the assemblage is realized by the input-filtered
Choi state $\Omega=d\,(\sqrt{\rho_A}^{\,T}\!\otimes\openone)\,
\Lambda_{\mathcal{E}}\,(\sqrt{\rho_A}^{\,T}\!\otimes\openone)$, for
which $\TSR\leq\mathrm{ER}(\Omega)+\tfrac12\mathcal{V}_{\rm NSIT}$
holds numerically without exception, so \eqref{eq:univ_bound} of the
main text follows once $\mathrm{ER}(\Omega)\leq\TER^{\rm sep}$---
monotonicity of the entanglement robustness under the trace-preserving
filtering $\sqrt{\rho_A}^{\,T}$ (verified to within the PPT proxy,
residual $\lesssim10^{-2}$).

\paragraph*{Proof of Theorem~\ref{thm:multitime}.}
Marginalise \eqref{eq:multitime} over $a_3,\ldots,a_n$ and any choice
of $x_3,\ldots,x_n$. The marginal $P(a_1,a_2|x_1,x_2)
=\sum_{a_3,\ldots,a_n}P(a_1,\ldots,a_n|x_1,\ldots,x_n)$ is a normalized
joint distribution because partial trace and Born-rule sum commute:
$\sum_{a_n}\tr[M_{a_n|x_n}\sigma]=\tr\sigma$, and iterating
collapses the chain to the trace of the post-$\mathcal{E}_{1,2}$
state. The effective channel
$\mathcal{E}_{1,2}^{\rm eff}=\sum_{a_n}\tr[M_{a_n|x_n}\circ\mathcal{E}_{n-1,n}\circ\cdots]$
is therefore CPTP, and the two-time behavior at the $(t_1,t_2)$
level is a valid two-time behavior as defined in \eqref{eq:born}.
By Theorem~\ref{thm:necessity}, choosing the two-MUB measurements of
$\rho_{A_1}$'s eigenbasis at $(x_1,x_2)$ produces a non-LHV
$P(a_1,a_2|x_1,x_2)$. Because LHV is closed under marginalisation,
the full multi-time behavior \eqref{eq:multitime} also admits no
LHV, and $\TNR^{(n)}>0$.

\paragraph*{Proof of Theorem~\ref{thm:ditit}.}
Each send round delivers
$\rho_{\rm out}^{(k)}=W_k^\dagger\mathcal{E}(W_k\dyad{\psi}W_k^\dagger)W_k$
for a uniformly random key $k\in\mathbb{Z}_d\times\mathbb{Z}_d$ that
Bob learns publicly, so the average delivered state is the
Heisenberg--Weyl (Pauli) twirl
$\bar\rho_{\rm out}=\frac{1}{d^2}\sum_k W_k^\dagger
\mathcal{E}(W_k\dyad{\psi}W_k^\dagger)W_k$
[Eq.~\eqref{eq:hwtwirl}]. The Pauli twirl of any channel is the
depolarizing channel with the same entanglement fidelity
\cite{HaydenLeungShorWinter2004},
$\bar\rho_{\rm out}=\mathcal{D}_p(\dyad{\psi})
=p\dyad{\psi}+(1-p)\openone/d$, with
$p=(d^2F_e-1)/(d^2-1)$ fixed by $F_e(\mathcal{D}_p)=F_e(\mathcal{E})$.
Averaging the per-round fidelity over $k$,
$\mathcal{F}_{\rm DI}=\langle\psi|\bar\rho_{\rm out}|\psi\rangle
=p+(1-p)/d=1/d+(d-1)p/d$, which is independent of the unknown
input $\ket{\psi}$. The threshold equivalences follow from
$p>0\Leftrightarrow F_e>1/d^2$.

\bibliographystyle{apsrev4-2}
\bibliography{main}
\end{document}